\documentclass[12pt]{article}

\usepackage[english]{babel}
\usepackage{a4wide}
\usepackage[latin1]{inputenc}
\usepackage[T1]{fontenc}
\usepackage{color}
\usepackage{eurosym}
\usepackage[top=4cm,bottom=4cm,left=3cm,right=3cm]{geometry}
\usepackage{fancyhdr}

\usepackage{amssymb,amsmath, amsthm}
\usepackage{url,graphicx}
\usepackage{float}

\begin{document}
 \title{Estimation of Integrated Quadratic Covariation With Endogenous Sampling Times \footnote{We would like to thank Simon Clinet, Mathieu Rosenbaum, Steven Lalley, Jianqing Fan (the Editor), an anonymous Associate Editor,
and one anonymous referee for helpful discussions and advice. Financial support from the National Science Foundation under grant DMS 14-07812 is greatly acknowledged.}}
\author{Yoann Potiron\footnote{Faculty of Business and Commerce, Keio University. 2-15-45 Mita, Minato-ku, Tokyo, 108-8345, Japan. Phone:  +81-3-5418-6571. Email: potiron@fbc.keio.ac.jp}  and Per A. Mykland\footnote{Department of Statistics, The University of Chicago. 5734 S. University Avenue Chicago, IL 60637. Phone:  + 1 (773) 702 8044/8333. Fax: + 1 (773) 702 9810. Email: mykland@pascal.uchicago.edu} \vspace{-2ex}}
\date{Forthcoming in \emph{Journal of Econometrics}}

\maketitle

\newcommand{\reels}{\mathbb{R}}
\newcommand{\naturels}{\mathbb{N}}
\newcommand{\relatifs}{\mathbb{Z}}
\newcommand{\rat}{\mathbb{Q}}
\newcommand{\complex}{\mathbb{C}}
\newcommand{\esp}{\mathbb{E}}
\newcommand{\proba}{\mathbb{P}}
\newcommand{\var}{\operatorname{Var}}
\newcommand{\cov}{\operatorname{Cov}}
\newcommand{\Tau}{\mathrm{T}}

\theoremstyle{plain}
\newtheorem{main1}{Theorem}
\newtheorem{maincor}[main1]{Corollary}
\newtheorem{estcor}[main1]{Corollary}
\newtheorem{approxreturn}[main1]{Definition}
\newtheorem{distrib}[main1]{Proposition}
\newtheorem{variance}[main1]{Definition}
\newtheorem{holdingconst}[main1]{Lemma}
\newtheorem{approxdistrib}[main1]{Lemma}
\newtheorem{approxtau}[main1]{Lemma}
\newtheorem{qs}[main1]{Definition}
\newtheorem{avgpsis}[main1]{Definition}
\newtheorem{avgapproxexp}[main1]{Definition}
\newtheorem{distribconst}[main1]{Lemma}
\newtheorem{termsvanishing}[main1]{Lemma}
\newtheorem{distribexists}[main1]{Lemma}
\newtheorem{snk}[main1]{Lemma}
\newtheorem{esptauk}[main1]{Lemma}
\newtheorem{numb}[main1]{Lemma}
\newtheorem{sumtight}[main1]{Lemma}
\newtheorem{s}[main1]{Lemma}
\newtheorem{tauapp}[main1]{Lemma}
\newtheorem{scale}[main1]{Lemma}
\newtheorem{estimating}[main1]{Lemma}
\newtheorem{lemmajumps}[main1]{Lemma}

\theoremstyle{definition}
\newtheorem{stableconvergence}{Definition}

\theoremstyle{remark}
\newtheorem{continuousass}{Remark}
\newtheorem{A3eq}[continuousass]{Remark}
\newtheorem{pl}[continuousass]{Remark}
\newtheorem{asymptoticbias}[continuousass]{Remark}
\newtheorem{asympbias}[continuousass]{Remark}
\newtheorem{rkendomodel}[continuousass]{Remark}
\newtheorem{mainrk}[continuousass]{Remark}
\newtheorem{estrk1}[continuousass]{Remark}
\newtheorem{estrk2}[continuousass]{Remark}
\newtheorem{pathbiased}[continuousass]{Remark}
\newtheorem{cvrate}[continuousass]{Remark}
\newtheorem{ass3rk}[continuousass]{Remark}
\newtheorem{jumps}[continuousass]{Remark}
\newtheorem{logprice}[continuousass]{Remark}
\newtheorem{N2start}[continuousass]{Remark}

\theoremstyle{remark}
\newtheorem{hittingbarrier}{Example}
\newtheorem{hittingbarriernoise}[hittingbarrier]{Example}
\newtheorem{hittingbarriernoisejump}[hittingbarrier]{Example}
\newtheorem{uncertaintyzones}[hittingbarrier]{Example}
\newtheorem{irregulargrid}[hittingbarrier]{Example}
\newtheorem{autoregressive}[hittingbarrier]{Example}

\theoremstyle{definition}
\newtheorem*{assumptions}{Assumptions}
\newtheorem*{assumptionA1}{Assumption (A1)}
\newtheorem*{assumptionA2}{Assumption (A2)}
\newtheorem*{assumptionA3}{Assumption (A3)}
\newtheorem*{assumptionA4}{Assumption (A4)}

\begin{abstract}
 When estimating high-frequency covariance (quadratic covariation) of two arbitrary assets observed asynchronously, 
 simple assumptions, such as independence, are usually imposed on the relationship between 
 the prices process and the observation times. In this paper, we introduce a general endogenous two-dimensional 
 nonparametric model. Because an observation is generated whenever an auxiliary process called \emph{observation time process} 
 hits one of the two boundary processes, it is called the \emph{hitting boundary process 
 with time process} (HBT) model. We establish a central limit theorem for the Hayashi-Yoshida (HY)
 estimator under HBT in the case where the price process and the observation price process follow a continuous It\^{o} process. We obtain an 
 asymptotic bias. We provide an estimator of the latter as well as a bias-corrected HY estimator of the 
 high-frequency covariance. In addition, we give a consistent estimator of the associated 
 standard error.  
\end{abstract}

\section{Introduction}
Covariation between two assets is a crucial quantity in finance. Fundamental examples include optimal asset allocation 
and risk management. In the past few years, using the increasing amount of high-frequency data available, many papers have been published about estimating 
this covariance. Suppose that the latent log-price of two arbitrary assets $X_t = (X_t^{(1)}, X_t^{(2)})$ follows a continuous It\^{o} process 
\begin{eqnarray}
\label{ito1} dX_t^{(1)} & := & \mu_t^{(1)} dt + \sigma_t^{(1)} dW_t^{(1)}, \\
\label{ito2} dX_t^{(2)} & := & \mu_t^{(2)} dt + \sigma_t^{(2)} dW_t^{(2)},
\end{eqnarray}
where $\mu_t^{(1)}, \mu_t^{(2)}, \sigma_t^{(1)}, \sigma_t^{(2)}$ are random processes, and $W_t^{(1)}$ and $W_t^{(2)}$ are standard Brownian motions, with 
(random) high-frequency correlation $d \langle W^{(1)}, W^{(2)} \rangle_t = \rho_t dt$. 
Econometrics usually seeks to infer the \emph{integrated covariation}
$$\langle X^{(1)}, X^{(2)} \rangle_t = \int_0^t \rho_u \sigma_u^{(1)} \sigma_u^{(2)} du.$$
Earlier results 
were focused on estimating the integrated variance of a single asset, starting from the probabilistic point of view 
(Genon-Catalot and Jacod (1993), Jacod (1994)). 
Barndorff-Nielsen and Shephard (2001, 2002) introduced the 
problem in econometrics. Adapted to two dimensions, if each process is observed simultaneously 
at (possibly random) times $\tau_{0,n} := 0$, $\tau_{1,n}$ , \ldots , $\tau_{N_n,n}$ the \emph{realized covariation} 
$\big[X^{(1)},X^{(2)} \big]_t$ is defined as the sum 
of cross log returns 
\begin{eqnarray}
\label{volest} \big[X^{(1)} , X^{(2)} \big]_t = \sum_{\tau_{i,n} \leq t} \Delta X_{\tau_{i,n}}^{(1)} \Delta X_{\tau_{i,n}}^{(2)},
\end{eqnarray}
where for any positive integer $i$, $\Delta X_{\tau_{i,n}}^{(k)} = X_{\tau_{i,n}}^{(k)} - X_{\tau_{i-1,n}}^{(k)}$ corresponds to 
the increment of the $k$th process between the last two sampling times. As the observation intervals 
$\Delta \tau_{i,n}$ get closer (and the number of observations $N_n$ goes to 
infinity), $\big[ X^{(1)} , X^{(2)} \big]_t \overset{\proba}{\rightarrow} \langle X^{(1)}, X^{(2)} \rangle_t$ 
(see e.g. Theorem I.4.47 in Jacod and Shiryaev (2003)). Furthermore, when the observation times $\tau_{i,n}$ are independent 
of the prices process $X_t$, its estimation error follows a mixed normal distribution (Jacod and Protter (1998), Zhang (2001), Mykland and Zhang (2006)). 
This gives us insight on how 
to estimate the integrated covariation. However, in 
practice, these two assumptions are usually not satisfied. The observation times of the two assets are rarely 
\emph{synchronous} and there is \emph{endogeneity} in the price sampling times.

\bigskip
The first issue has been studied for a long time. The lack of synchronicity often creates undesirable effects in inference. 
If we sample at very high frequencies, we observe the Epps effect (Epps (1979)), i.e. the correlation estimates are drastically decreased compared to 
an estimate with sparse observations.  Hayashi and Yoshida (2005) introduced the so-called 
\emph{Hayashi-Yoshida estimator} (HY)
\begin{eqnarray}
\label{HY} \langle \widehat{X^{(1)}, X^{(2)}} \rangle_t^{HY} = \sum_{\tau_{i,n}^{(1)} , \tau_{j,n}^{(2)} < t} \Delta X_{\tau_{i,n}^{(1)}}^{(1)} 
\Delta X_{\tau_{j,n}^{(2)}}^{(2)} 
\mathbf{1}_{ \big\{ [ \tau_{i-1,n}^{(1)}, \tau_{i,n}^{(1)} ) \cap [ \tau_{j-1,n}^{(2)}, \tau_{j,n}^{(2)} ) \neq \emptyset \big\} },
\end{eqnarray}
where $\tau_{i,n}^{(k)}$ are the observation times of the $k$th asset. Note that if the observations of both 
processes occur simultaneously, 
(\ref{volest}) and (\ref{HY}) are equal. The consistency of this estimator was achieved in Hayashi and Yoshida (2005) and Hayashi and Kusuoka (2008). The corresponding central limit theorems were investigated in Hayashi and Yoshida 
(2008, 2011) under strong predictability of observation times, which is a more restrictive assumption than only assuming they are stopping times 
but still allows some dependence between prices and observation times. Recently, Koike (2014, 2015) extended the pre-averaged 
Hayashi-Yoshida estimator first under predictability of observation times, and then under a more general endogenous setting of stopping times. Other examples of high-frequency covariance estimators can be found in Zhang (2011), Barndorff-Nielsen et al. (2011), A\"{i}t-Sahalia et al. (2010), Christensen et al. (2010, 2013).  

\bigskip
In a general one-dimensional endogenous model, the asymptotic behaviour of the realized volatility (\ref{volest}) 
has been investigated in the case of sampling times given by hitting times on a grid (Fukasawa (2010a), Robert and Rosenbaum (2011, 2012), Fukasawa and Rosenbaum (2012)). Due to the regularity of those three models (see the discussion in the latter paper), they don't obtain any bias in the 
limit distribution of the normalized error. 
Also, the case of strongly predictable stopping times is treated in Hayashi et al. (2011). Finally, two general results (Fukasawa (2010b), Li and al. (2014)) showed that we can identify and estimate the asymptotic bias.

\bigskip
The primary goal of this paper is to bias-correct the HY. Note that estimating the bias is more challenging than in 
the volatility case because observations are asynchronous. In particular, the estimator will involve a quantity that can be considered as the 
\emph{tricity} of Li et al. (2014), but with a more intricate definition because of the asynchronicity in sampling 
times. This new definition can be seen as an analogy with the generalization of the RV estimator (\ref{volest}) 
by the HY estimator (\ref{HY}).

\bigskip
Another very important issue to address is the estimation of the asymptotic standard deviation. First, because the model 
is more general than in the no-endogeneity work, the theoretical asymptotic variance will be different. Consequently, 
a new variance estimator, which takes into proper account the endogeneity, will be given. 

\bigskip
The authors want to take no position on the joint distribution of the log-return and the next observation time 
that corresponds to an asset price change because they know that their unknown relationship is most 
likely contributing to the bias and the variance of the high-frequency covariance's estimate 
when we (wrongly) assume full independence between the price process and observation times. For this purpose, they 
introduce the \emph{hitting boundary process with time process} (HBT) model.

\bigskip
Finally, techniques developed in the proofs are innovative in the sense that they reduce the normalized error of the Hayashi-Yoshida estimator to a discrete process, which is locally a uniformly ergodic homogeneous Markov chain. Thus, the problem can be solved locally, and because we assume that the volatility of assets is continuous, the error of approximation between the 
local Markov structure and the real structure of the normalized error vanishes asymptotically. This technique is not problem-specific, and it can very much be applied to other estimators dealing with temporal data. 

\bigskip
The paper is organized as follows. We introduce the HBT model in Section 2. Examples covered by this model are given in Section 3. The main theorem of this work, the limit distribution of the normalized error is given in Section 4. Estimators of the asymptotic bias and variance are provided in Section 5. We carry out numerical simulations in Section 6 to corroborate the theory. Proofs are developed in Appendix.

\section{Definition of the HBT model}

We first introduce the model in $1$-dimension. We assume that for any positive integer $i$, $\tau_{i+1}$ is the next arrival time (after $\tau_i$) that corresponds to an actual change of price. In particular, several trades can occur at the same price $Z_{\tau_i}$ between $\tau_i$ and $\tau_{i+1}$, but no trade can occur with a price different than $Z_{\tau_i}$ before $\tau_{i+1}$. We also assume that $X_t$ is the efficient (log) price of the security of interest. In addition, we assume that the observations are noisy and that we observe $Z_{\tau_i} := X_{\tau_i} + \epsilon_{\tau_i}$ where the microstructure noise $\epsilon_{\tau_i}$ can be expressed as a known function of the observed prices $Z_0, \ldots, Z_{\tau_{i}}$. As an example, Robert and Rosenbaum (2012) showed in $(2.3)$ in p. 5 that the model with uncertainty zones can be written with that noise structure if we assume that we know the friction parameter $\eta$. Finally, we define $\alpha > 0$ as the tick size, and we assume that the observed price $Z_{\tau_i}$ lays on the tick grid, i.e. there exists positive integers $m_i$ such that $Z_{\tau_i} := m_i \alpha$.

\bigskip
Empirically, no economical model based on rational behaviors of agents on the stock markets, that shed light on 
the relationship between the efficient return $\Delta X_{\tau_i}$ and time before the next price 
change $\Delta \tau_i = \tau_i - \tau_{i-1}$, has won unanimous support. When arrival times are independent of the asset price, it follows directly from the continuous It\^{o}-assumption that the dependence structure 
is such that the return $\Delta X_{\tau_i}$ is a function of $\Delta \tau_i$. The longer we wait, the bigger the variance of the return 
is expected to be. In this paper, we take the opposite point of view by building a model in which $\tau_i$ is defined 
as a function of the efficient price path. For that purpose, we define the \emph{observation time process} 
 $X_t^{(t)}$ that will drive the 
observation times. We also define the \emph{down process} $d_t(s)$ and the \emph{up process} $u_t(s)$. Note 
that for any $t \geq 0$, we assume that $d_t$ and $u_t$ are functions on 
$\reels^+$. We also assume 
that the down process takes only negative values and that the up 
process takes only positive values. A new 
observation time will be generated whenever one of those two processes is hit by the increment of the observation time process. 
Then, the increment of the observation time process will start again from $0$, and the next observation time will be generated whenever it 
hits the up or the down process. Figure \ref{illustration} illustrates the HBT model. Formally, we define $\tau_0 := 0$ and for any positive integer $i$ as
\begin{eqnarray}
 \label{generateobstimes} \tau_i := \inf \Big\{ t > \tau_{i-1} : \Delta X_{[\tau_{i-1}, t]}^{(t)} \notin \big[ d_{t} 
 \left( t - \tau_{i-1} \right), u_{t} \left(t - \tau_{i-1} \right) 
\big] \Big\},
\end{eqnarray}
where $\Delta Y_{[a,b]} := Y_b - Y_a$. Note that if the observation time process $X_t^{(t)}$ is equal to the price process 
$X_t$ itself, then the price will go up (respectively go down) whenever it hits the up process (down process). Note also that if the time process, the up process and the down process are independent of the efficient price process, then the arrival times are independent of the efficient price process. 
We assume that the two-dimensional process $( X_t, X_t^{(t)} )$ is an It\^{o}-process. Section $3.1$ provides examples 
of the literature identifying the observation time process, the down process and the up process. 

\bigskip
Generalizing to two dimensions is straightforward. We define $X_t^{(t,k)}$ for $k=1,2$ to be 
the observation time process associated with 
the $k$th price process , $u_t^{(k)}$ the up process, $d_t^{(k)}$  the down process, and the arrival times $\tau_i^{(k)}$ generated by (\ref{generateobstimes}). We also define the four dimensional process $Y_t := (X_t^{(1)}, X_t^{(2)}, X_t^{(t,1)}, X_t^{(t,2)})$, 
and assume $Y_t$ follows an It\^{o}-process with volatility
$$\sigma_t := \begin{pmatrix}
               \sigma_t^{1,1} & \sigma_t^{1,2} & \sigma_t^{1,3} & \sigma_t^{1,4}\\
               \sigma_t^{2,1} & \sigma_t^{2,2} & \sigma_t^{2,3} & \sigma_t^{2,4}\\
               \sigma_t^{3,1} & \sigma_t^{3,2} & \sigma_t^{3,3} & \sigma_t^{3,4}\\
               \sigma_t^{4,1} & \sigma_t^{4,2} & \sigma_t^{4,3} & \sigma_t^{4,4}\\
              \end{pmatrix}.
$$
In particular, we have $d Y_t = \mu_t dt + \sigma_t dW_t$, where $W_t$ is a four dimensional standard Brownian motion (for 
$i=1, \ldots, 4$ and $j = 1, \ldots, 4$ such that $i \neq j$, $W_t^{(i)}$ is independent of $W_t^{(j)}$). 
If we set $\zeta_t = \sigma_t \sigma_t^T$, then the integrated covariance (or quadratic covariation) process is given by 
$\langle Y, Y \rangle_t = \int_0^t \zeta_s ds$. Let $\rho_t$ be the associated correlation process of $Y_t$, 
i.e. for $i = 1, \ldots, 4$ and $j = 1, \ldots, 4$ we set $\rho_t^{i,j} = \zeta_t^{i,j} (\zeta_t^{i,i})^{-1}$. 
Finally, it is useful sometimes to see $Y_t$ as a four dimensional vector expressed as in equations (\ref{ito1}) and (\ref{ito2}). 
For $k= 1, \ldots, 4$ we define the volatility of the $k$th process as $\sigma_t^{(k)} := (\zeta_t^{k,k})^{\frac{1}{2}}$, we can thus 
express $Y_t^{(k)}$ as
$$d Y_t^{(k)} = \mu_t^{(k)} dt + \sigma_t^{(k)} dB_t^{(k)}$$
where $B_t^{(k)}$ is a standard Brownian motion, which typically depends on $B_t^{(l)}$ for $l = 1, \ldots, 4$.

\section{Examples} \label{p1examples}
We insist on the fact that estimators of covariance and associated asymptotic variance given in this paper don't require any knowledge of the structure of 
 the observation time process, the up process and the down process. Nonetheless, for financial and economic interpretation purposes, the reader might be interested in getting an idea on how those processes behave in practice. We provide in this section several examples from the literature as well as possible extensions of the model with uncertainty zones of Robert and Rosenbaum (2011) that can be expressed as HBT models.
 
\subsection{Endogenous models contained in the HBT class} \label{ex1}
\begin{hittingbarrier} \label{hittingbarrier}
(hitting constant boundaries) The simplest endogenous semi-parametric model we can think of is a model where the time process $X_t^{(t)}$ is equal to the price process $X_t$, and 
times are generated by hitting a constant barrier. Formally, it means that there exists a two-dimensional parameter 
$(\theta_u, \theta_d)$ such that the up process is equal to $\theta_u$ and the down process is equal to $\theta_d$. We don't assume noise in that model.
\end{hittingbarrier}

\begin{hittingbarriernoise} \label{hittingbarriernoise}
(hitting constant boundaries of the tick size) One issue with Example \ref{hittingbarrier} is that the efficient price $X_{\tau_i}$, which is observed because no microstructure noise is assumed in the model, is not necessarily a modulo of the tick size $\alpha$ if $\theta_u$ and $\theta_d$ are not multiples of $\alpha$. To make Example \ref{ex1} feasible in practice, we assume here that the constant barriers $\theta_u$ and $\theta_d$ are respectively equal to the tick size $\alpha$ and its additive inverse $- \alpha$. We also assume that $Z_{\tau_i} := X_{\tau_i}$.
\end{hittingbarriernoise}

\begin{hittingbarriernoisejump} \label{hittingbarriernoisejump}
(hitting constant boundaries of the jump size) The issue with Example \ref{hittingbarriernoise} is that the absolute jump size of the observed price $Z_{\tau_i}$ is $\alpha$. On the contrary, in practice the absolute jump size can actually be bigger than the tick size $\alpha$. In the notation of 
Robert and Rosenbaum (2011), for any positive integer $i$, we introduce a discrete variables $L_i$ which 
corresponds to the observed price jump's tick number between $\tau_{i}$ and $\tau_{i+1}$, with $L_i \geq 1$. We assume that $L_i$ is bounded. The arrival times are defined recursively as $\tau_0 :=0$ and for any positive 
integer $i$ as
$$\tau_i := \inf \Big\{ t > \tau_{i-1} : X_t = X_{\tau_{i-1}} - L_{i-1} \alpha 
\text{ or } X_t = X_{\tau_{i-1}} + L_{i-1} \alpha \Big\}.$$
We assume that $L_i$ are IID and independent of the other quantities. We finally assume that $Z_{\tau_i} := X_{\tau_i}$. The up and down processes 
are piecewise constant in $t$ and constant in $s$, defined for any $s \geq 0$ as
\begin{eqnarray*}
d_t (s) = & - L_{i-1} \alpha  & \text{ for } 
t \in (\tau_{i-1}, \tau_i]\\
u_t (s)  = & L_{i-1} \alpha  & \text{ for } 
t \in (\tau_{i-1}, \tau_i]
\end{eqnarray*}
\end{hittingbarriernoisejump}

\begin{uncertaintyzones} \label{uncertaintyzones}
(model with uncertainty zones) We go one step further than Example \ref{hittingbarriernoisejump} and introduce now the model with uncertainty zones of Robert and Rosenbaum (2011). In a frictionless market, we can assume that a trade with change of price $Z_{\tau_i}$ will occur whenever the efficient price process crosses one of the mid-tick values $Z_{\tau_{i-1}} + \frac{\alpha}{2}$ or $Z_{\tau_{i-1}} - \frac{\alpha}{2}$. In that case, if the efficient price process hits the former value, we would observe an increment of the observed price $Z_{\tau_{i}} = Z_{\tau_{i-1}} + \alpha$ and if it hits the former value, we would observe a decrement $Z_{\tau_{i}} = Z_{\tau_{i-1}} - \alpha$. There are two reasons why in practice such a frictionless model is too simplistic. The first reason is that the absolute value of the increment (or the decrement) of the observed price can be bigger than the tick size $\alpha$ and was already pointed out in Example \ref{hittingbarriernoisejump}. We will thus keep the notation $L_i$ in this example. The second reason is that the frictions induce that the transaction will not exactly occur when the efficient process is equal to the mid-tick values. For this purpose in the notation of Robert and Rosenbaum (2012), let $0 < \eta < 1$ be a parameter that quantifies 
the aversion to price changes of the market participants. If we let $X_t^{(\alpha)}$ be the value of $X_t$ rounded to the nearest multiple of $\alpha$, the sampling times are defined recursively as $\tau_0 :=0$ and for any positive 
integer $i$ as
$$\tau_i := \inf \Big\{ t > \tau_{i-1} : X_t = X_{\tau_{i-1}}^{(\alpha)} - \alpha \big(L_{i-1} - \frac{1}{2} + \eta \big) 
\text{ or } X_t = X_{\tau_{i-1}}^{(\alpha)} + \alpha \big(L_{i-1} - \frac{1}{2} + \eta \big) \Big\}$$
The observed price is equal to the rounded efficient price $Z_{\tau_i} := X_{\tau_i}^{(\alpha)}$. The time process $X_t^{(t)}$ is again equal to the price process $X_t$ itself in this model. The up and down processes 
are piecewise constant in $t$ and constant in $s$, defined for any $s \geq 0$ as
\begin{eqnarray*}
d_t (s) = & - L_{i-1} \alpha \mathbf{1}_{ \{ X_{\tau_{i-1}} < X_{\tau_{i-2}} \} } - 
\left( 2 \eta + L_{i-1} - 1 \right) \alpha \mathbf{1}_{ \{ X_{\tau_{i-1}} > X_{\tau_{i-2}} \} } & \text{ for } 
t \in (\tau_{i-1}, \tau_i]\\
u_t (s)  = & L_{i-1} \alpha \mathbf{1}_{ \{ X_{\tau_{i-1}} > X_{\tau_{i-2}} \} } +
\left( 2 \eta + L_{i-1} - 1 \right) \alpha \mathbf{1}_{ \{ X_{\tau_{i-1}} < X_{\tau_{i-2}} \} } & \text{ for } 
t \in (\tau_{i-1}, \tau_i]
\end{eqnarray*}
where $\mathbf{1}_A$ is the indicator function of A. Note that in the case where $\eta = \frac{1}{2}$, we are back to Example \ref{hittingbarriernoisejump}.
\end{uncertaintyzones}

\begin{irregulargrid} \label{irregulargrid}
(times generated by hitting an irregular grid model) The fourth model we are looking at is called \emph{times generated by hitting an irregular grid model}. We follow the 
notation of Fukasawa and Rosenbaum (2012) and consider the irregular grid 
$\mathcal{G} = \{ p_k \}_{k \in \relatifs}$, with $p_k < p_{k+1}$. We set $\tau_0 = 0$ and for $i \geq 1$
$$\tau_i = \inf \Big\{ t > \tau_{i-1} : X_t \in \mathcal{G} - \{ X_{\tau_{i-1}} \} \Big\} ,$$
where $\mathcal{G} - \{ X_{\tau_{i-1}} \}$ is the set obtained by removing $\{ X_{\tau_{i-1}} \}$ from $\mathcal{G}$.  
We can rewrite it as an element of the HBT model where the time process is equal to the price process, and for all $s \geq 0$ the up and down processes 
are defined as
\begin{eqnarray*}
d_t (s)  = & p_{k-1} - p_k & \text{ for } t \in (\tau_{i-1}, \tau_i ]\\
u_t (s)  = & p_{k+1} - p_k & \text{ for } t \in (\tau_{i-1}, \tau_i ],
\end{eqnarray*}
where $k$ is the (random) index such that $p_k = X_{\tau_{i-1}}$. 
\end{irregulargrid}

\begin{autoregressive} \label{autoregressive}
(structural autoregressive conditional duration model) There have been several drafts for this 
model. We follow here a former version 
(Renault et al. (2009)), because we can directly express it as an element of the HBT 
model\footnote{Generating the sampling times (\ref{generateobstimes}) of the HBT model as a first hitting-time of a 
unique barrier 
instead of the first hitting time of one of two barriers as in the latter version of Renault et al. (2014) 
wouldn't change much the proofs of 
this paper, but we chose the two-boundaries setting 
because it seems more natural if interpretation of time processes, up processes and down processes is needed.}. 
In the structural autoregressive conditional duration 
model, the time $\tau_{i}$ when the next event occurs is given by $\tau_0=0$ and for $i>0$
\begin{equation}
\label{autoregressiveACD}
\tau_i = \inf \Big\{ t > \tau_{i-1} : A_t - A_{\tau_{i-1}} = \tilde{d}_{\tau_{i-1}} \text{ or } 
A_t - A_{\tau_{i-1}} = \tilde{c}_{\tau_{i-1}} \Big\}
\end{equation}
where $A_t$ is a standard Brownian motion (not necessarily independent of $X_t$). Expressed as an element of the HBT 
model, we have that the time process $X_t^{(t)}$ is equal to the Brownian motion $A_t$ and for all $s \geq 0$
\begin{eqnarray*}
d_t (s)  = & \tilde{d}_{\tau_{i-1}} \text{ for } & t \in (\tau_{i-1}, \tau_i ]\\
u_t (s)  = & \tilde{c}_{\tau_{i-1}} \text{ for } & t \in (\tau_{i-1}, \tau_i ].
\end{eqnarray*}
\end{autoregressive}

\subsection{Possible extensions of the model with uncertainty zones}  
The model with uncertainty zones of Robert and Rosenbaum (2011) introduced in Example \ref{uncertaintyzones}, which is semi-parametric, assumes that the observed price is the efficient price rounded to the nearest tick value $Z_{\tau_i} = X_{\tau_i}^{(\alpha)}$ and thus the 
noise is equal to $\epsilon_i := \alpha (\frac{1}{2} - \eta)$ if the last trade increased the price and $\epsilon_i := - \alpha (\frac{1}{2} - \eta)$ if the last trade decreased the price. In particular, the noise is auto-correlated and correlated to the efficient price. Because of this specific noise distribution, it is directly possible to estimate the underlying friction parameter $\eta$ without any data pre-processing such as preaveraging (see Robert and Rosenbaum (2012)). We believe the model with uncertainty zones is a very interesting starting point, because all the 
endogenous and noise structure of the model 
is reduced to the estimation of the $1$-dimensional friction parameter $\eta$. Nevertheless, as this semi-parametric model wants to be the 
simplest, it suffers from several issues. We will investigate two of them in the following.

\bigskip
First, the model doesn't allow for asymmetric information between the buyers and the sellers. Define $\eta^+$ and 
$\eta^-$, which are respectively the aversion to a positive price change and a negative price change. As a positive 
price change means that a buyer decided to put an order at the best ask price and a negative price change corresponds to 
a seller that puts an order at the best bid price (if we assume that cancel and repost orders are not the reason why the price changed), the difference 
$\eta^+ - \eta^-$ can be seen as a measure of 
information asymmetry. We define $\tau_0 := 0$ and recursively for $i$ any positive integer 
$$\tau_i := \inf \Big\{ t > \tau_{i-1} : X_t = X_{\tau_{i-1}}^{(\alpha)} - \alpha \big(L_i - \frac{1}{2} + \eta^- \big) 
\text{ or } X_t = X_{\tau_{i-1}}^{(\alpha)} + \alpha \big(L_i - \frac{1}{2} + \eta^+ \big)\Big\} .$$
Note that the HBT class contains this model and that it can be directly fitted if we slightly modify $\hat{\eta}$ in Robert and Rosenbaum (2012) to estimate 
$\eta^+$ and $\eta^-$. One possible application would be to build a test of asymmetric information $\eta^+ := \eta^-$. This is beyond the scope of this paper.

\bigskip
One other issue is that the authors don't do any model checking in their work. According to their empirical work (see pp. 359-361 
of Robert and Rosenbaum (2011)), the estimated values for $\eta$ are stable accross days for the ten French assets 
tested. Stability of $\eta$ favors their model but by doing so, the model doesn't allow any other structure than the 
full-endogeneity for the sampling times. Even if the true structure of sampling times is (mostly) independent of the
asset price, we 
will still estimate an $\eta$ that will be stable across days. If we allow the time process to be different from 
the price process itself, we can estimate the correlation $\rho^{1,3}$ between them and see how endogenous the sampling times 
are (the bigger $\big| \rho^{1,3} \big|$ is, the more endogenous the sampling times are). We would need to add more general microstructure noise in the model, and thus this is left for further work.

\section{Main result}
\subsection{Assumptions and Theorem}
Without loss of generality, we fix the horizon time $T:=1$, and we consider $[0,1]$ to represent the course of an 
economic event, such as a trading day. We first introduce the definition of stable convergence, which is a little bit 
stronger than usual convergence in distribution 
and needed for statistical purposes of inference, such as the prediction value of the high-frequency covariance and the 
construction of a confidence interval at a given confidence level.
\begin{stableconvergence}
We suppose that the random processes $Y_t$, $\mu_t$ and $\sigma_t$ are adapted to a 
filtration $\left( \mathcal{F}_t \right)$. Let $Z_n$ be a sequence of 
$\mathcal{F}_1$-measurable random variables. We say that $Z_n$ converges stably in distribution to $Z$ as $n \rightarrow \infty$ if $Z$ is measurable with 
respect to an extension of $\mathcal{F}_1$ so that for all $A \in \mathcal{F}_1$ and for all bounded continuous\footnote{Note that the continuity of $f$ refers 
to continuity with respect to the Skorokhod topology of $\mathbb{D} [0,1]$. Nevertheless, we can also use continuity 
given by the sup-norm, because all our limits are in $\mathbb{C} [0,1]$. One can look at Chapter $VI$ of Jacod and Shiryaev (2003) as a reference. For 
further definition of stable convergence, one can look at R\'{e}nyi (1963), Aldous and Eagleson (1978), Chapter 3 (p. 56) of Hall and Heyde (1980), 
Rootz\'{e}n (1980), and Section 2 (pp. 169-170) of Jacod and Protter (1998).} functions $f$, $\esp \left[ 
\mathbf{1}_A f \left( Z_n \right) \right] \rightarrow \esp \left[ \mathbf{1}_A f \left( Z \right) \right]$ as $n \rightarrow \infty$.
\end{stableconvergence}

\bigskip
In the setting of Section 2, the target of inference, the integrated covariation, can be written  for all $t \in [0, 1]$ as
\begin{eqnarray*}
\langle X^{(1)}, X^{(2)} \rangle_t := \int_0^{t} \sigma_s^{(1)} \sigma_s^{(2)} \rho_s^{1,2} ds . 
\end{eqnarray*}
We are providing now the asymptotics. We want to make the number of observations go to infinity asymptotically. The idea is to scale and thus keep the 
structure that drives the next return and the next observation time, while making the tick size vanish (and thus the 
number of observations explode on $[0,1]$). Formally, we let the tick size $\alpha > 0$ and we define the observation times 
$\Tau_{\alpha} := \big\{ \tau_{i,\alpha}^{(k)} \big\}_{i \geq 0}^{k=1,2}$ such that for $k=1,2$ we have	
$\tau_{0,\alpha}^{(k)} := 0$ and for $i$ any positive integer
$$\tau_{i,\alpha}^{(k)} := \inf \Big\{ t > \tau_{i-1,\alpha}^{(k)} : \Delta X_t^{(t,k)} \notin \big[ \alpha d_{t}^{(k)}(t - 
\tau_{i-1, \alpha}^{(k)}), \alpha u_{t}^{(k)}(t - \tau_{i-1, \alpha}^{(k)}) \big] \Big\} .$$
We define the HY estimator when the tick size is equal to $\alpha$ as
\begin{eqnarray}
\label{HY0} \langle \widehat{X^{(1)}, X^{(2)}} \rangle_{t,\alpha}^{HY} := \sum_{0 < \tau_{i,\alpha}^{(1)}\text{ , } \tau_{j,\alpha}^{(2)} < t} 
\Delta X^{(1)}_{\tau_{i,\alpha}^{(1)}} \Delta X^{(2)}_{\tau_{j,\alpha}^{(2)}} 
\mathbf{1}_{ \big\{ [ \tau_{i-1, \alpha}^{(1)}, \tau_{i, \alpha}^{(1)} ) \cap [ \tau_{j-1, \alpha}^{(2)}, 
\tau_{j, \alpha}^{(2)} ) \neq \emptyset \big\} } .
\end{eqnarray}
We now give the assumptions needed to prove the central limit theorem of (\ref{HY0}). We need to introduce some definitions for this purpose. 
In view of the different models introduced in Section $3$, there are three different possible assumptions regarding the correlation between the time processes 
$X_t^{(t)}$ and the price processes $X_t$. The first possibility is that they can be equal for all $0 \leq t \leq T$. In this case we define $\lambda_t^{\min}$ as the smallest eigen-value of $(\sigma_t^{(i,j)})_{i=1,2}^{j=1,2}$. The second scenario is that for one $k \in {1,2}$ we
have $X_t^{(k)} := X_t^{(t,k)}$, but the other time process is different from its associated price process. In that case, we define $\lambda_t^{\min}$ the smallest eigen-value of 
$(\sigma_t^{(i,j)})_{i \in \{1,2,3,4\} - \{k+2\}}^{j \in \{1,2,3,4\} - \{k+2\}}$. The third possible setting is that the time process is different from its 
associated asset price for both assets, and we let $\lambda_t^{\min}$ the smallest eigen-value of 
$\sigma_t$ in that case. Assumption $(A1)$ provides conditions 
on the price processes $X_t^{(1)}$ and $X_t^{(2)}$, the time processes $X_t^{(t,1)}$ and $X_t^{(t,2)}$ as well as 
their covariance matrix $\sigma_t$. There are two types of assumptions in $(A1)$. First, we want to get rid of the 
drift in the proofs, and this will be done using condition $(A1)$ together with the Girsanov theorem and local arguments 
(see e.g. pp.158-161 in Mykland and Zhang (2012)). This is a very standard assumption in the literature of 
financial econometrics. Furthermore, we assume that the covariance matrix $\sigma_t$ is continuous. 

\begin{assumptionA1} 
 The drift $\mu_t$, the volatility matrix $\sigma_t$ and the (four dimensional) Brownian motion $W_t$ are adapted 
  to a filtration $( \mathcal{F}_t )$. Also, $\mu_t$ 
  is integrable and locally bounded. Furthermore, $\sigma_t$ is continuous. Finally, we assume that $\underset{t \in (0,1]}{\inf} \lambda_t^{\min} > 0$ a.s.
\end{assumptionA1}
\begin{continuousass} (robustness to jumps in volatility) The proof techniques, holding the volatility constant on small blocks, require the "continuity of volatility". This is the same strategy as in Mykland and Zhang (2009) and Mykland (2012) where the volatility process follows a continuous It\^{o} process. Nonetheless, following the same line of reasoning as for the proof of Remark \ref{rkjumps}, we can add a finite number of jumps in the volatility matrix. The proof of Theorem \ref{main1} will break in the case of infinite number of jumps in $\sigma_t$.
\end{continuousass}

The following condition roughly assumes that both time processes can't be equal to each other, even on a very small time interval. 
Specifically, we will assume that there is a constant strictly smaller than $1$ such that the module of the 
\emph{instantaneous high-frequency correlation} $\rho_t^{3,4}$ can't be bigger than this constant. In practice, 
assumption $(A2)$ is harmless.
\begin{assumptionA2}
 For all $t \in [0,1]$ we have
  \begin{eqnarray}
   \label{A2rho34} \rho_t^{3,4} \in [ \rho_{-}^{3,4} , \rho_{+}^{3,4} ] ,
   \end{eqnarray}
  where $\max ( \mid \rho_{-}^{3,4} \mid, \mid \rho_{+}^{3,4} \mid ) < 1$.
\end{assumptionA2}
The next assumption deals with the down process $d_t$ and the up process $u_t$. It is clear that $d_t$ and $u_t$ have to be known 
with information at time $t$, which is why we assume that they are adapted to $(\mathcal{F}_t)$. The rest of assumption 
$(A3)$ is very technical and we only try to be as general as we can with respect to the proof techniques we will use. 
The reader should understand Assumption $(A3)$ as ``assume the worst dependence structure possible between the return 
$\Delta X_{\tau_i}$ and the time increment $\Delta \tau_i$, knowing that they follow the HBT model''. We insist once again on the fact that we only make the dependence structure as bad as we can in our model so that we can investigate how biased the HY estimator can be in practice, and how much the estimates of 
the variance assuming no endogeneity are wrong.
\begin{assumptionA3}
   For both assets $k=1,2$, define the couple of the down process and the up process 
   $g_t^{(k)} := ( d_t^{(k)}, u_t^{(k)} )$ and let $g_t := (g_t^{(1)}, g_t^{(2)})$. We assume that 
  $$
   \begin{array}{r c  c l}
      g^{(k)} : & \reels^{+} & \rightarrow & \left( \reels^{+} \rightarrow \reels^{-} \times \reels^{+} \right)\\
  & t & \mapsto & g_t^{(k)}
   \end{array}
$$
is adapted to $\left( \mathcal{F}_t \right)$. 
Moreover, there exists two non-random constants $0 < g^- < g^+$ such that a.s.
for any $t \in [0,1]$ and for any $s \geq 0$
\begin{eqnarray}
\label{A3g} g^- \leq \min (- d_t^{(k)}(s), u_t^{(k)}(s) ) \leq 
\max ( - d_t^{(k)}(s), u_t^{(k)}(s) ) \leq g^+
\end{eqnarray}
Furthermore, there exists non-random constants $K > 0$ and $d > 1/2$ such that a.s. 
\begin{eqnarray}
 \label{A3compact} \forall s \geq K \text{ , } g_t \left( s \right) = g_t \left( K \right),
\end{eqnarray} 
\begin{eqnarray}
 \label{A3der} \forall t\geq 0, \text{ } g_t \text{ is differentiable and } \forall s \geq 0 \text{, } 
 \max \big( | ( d_t^{(k)} )' ( s ) |, | ( u_t^{(k)} )'(s) | \big) \leq K,
 \end{eqnarray}
 \begin{eqnarray}
\label{A3sup} \forall \text{ } \left( u, v\right) \in [0,1]^2 \text{ s.t. } 0 < u < v, \text{ }  & 
\| g_{v} - g_{u} \|_{\infty} \leq K | v - u |^{d} ,
 \end{eqnarray}
where $\| (f_1, f_2) \|_{\infty} = \underset{w \geq 0}{\sup} \max \left( | f_1 \left( w \right) |, 
| f_2 \left( w \right) | \right)$. 

\end{assumptionA3}

\begin{A3eq}
 Consider the space $\mathcal{C}$ of constants defined in Assumption $(A3)$
 \begin{eqnarray*}
  \mathcal{C} := \Big\{ (g^- , g^+ , K , d) \text{ s.t. } 0 < g^- < g^+ \text{ , } K > 0 \text{ , } d > \frac{1}{2} 
  \Big\} .
 \end{eqnarray*}
For any $c \in \mathcal{C}$, we define $\mathcal{G} (c)$ to be the functional subspace of 
$\reels^{+} \rightarrow ( \reels^{+}$  $\rightarrow \reels^{-} 
\times \reels^{+} )^2$ such that $\forall g \in \mathcal{G}$, $g$ satisfies (\ref{A3g}), (\ref{A3compact}), (\ref{A3der}) 
and (\ref{A3sup}). When there is no room for confusion, we use $\mathcal{G}$. Assumption (A3) is equivalent to 
$$\exists c \in \mathcal{C} \text{ s.t. } \forall t \in [0,1] \text{ , } g_t \in \mathcal{G} (c).$$
\end{A3eq}
\begin{ass3rk}
 The advised reader will have noticed that Example \ref{hittingbarriernoisejump}, Example \ref{uncertaintyzones}, Example \ref{irregulargrid} and Example \ref{autoregressive}, where time processes are piecewise-constant and may depend on $n$, don't follow Assumption $(A3)$. The adaptation of Theorem \ref{main1} proofs in those examples is discussed in Appendix \ref{adaptationproofs}. We have made the choice not to state more general conditions to keep tractability of Assumption $(A3)$.
\end{ass3rk}
The last assumption is only technical, and also appears in the literature (Mykland and Zhang (2012), Li et al. (2014)).
\begin{assumptionA4}
 The filtration $(\mathcal{F}_t)$ is generated by finitely many Brownian motions.
\end{assumptionA4}
We can now state the main theorem.
\begin{main1} \label{main1}
 Assume $(A1)-(A4)$. Then, there exist processes $AB_t$ and $AV_t$ adapted to 
 $( \mathcal{F}_t )$ such that 
 stably in law as the tick size $\alpha \rightarrow 0$, 
 \begin{eqnarray}
 \label{theorem}
 \alpha^{-1} \left(\langle \widehat{X^{(1)}, X^{(2)}} \rangle_{t,\alpha}^{HY} - 
 \langle X^{(1)}, X^{(2)} \rangle_{t} \right) \rightarrow AB_t + \int_0^t \left( 
 AV_s \right)^{1/2} d Z_s ,
 \end{eqnarray}
 where $Z_t$ is a Brownian motion independent of the underlying $\sigma$-field.  
 The asymptotic bias $AB_t$ and the asymptotic variance $AV_t$ are defined in Section $4.3$ and estimated in Section $5$.
\end{main1}

\begin{pathbiased} \label{pathbiased}
 (path-bias) Note that the asymptotic bias term $AB_t$ on the right-hand side of (\ref{theorem}) doesn't mean that 
 the Hayashi-Yoshida estimator is \emph{biased}, but rather \emph{path-biased}. The latter is a weaker statement which means 
 that once we have seen a path, there is a \emph{bias} for the HY estimator on this specific path of value $AB_t$. 
 In practice, we only get to see one path and thus \emph{bias} and \emph{path-bias} can be confused easily. When doing simulations, we can observe many paths and the reader should keep in mind that the \emph{path-bias} will be different for each path. In addition, note that if we assume that $\sigma_t$ is bounded and bounded away from $0$ on $[0, T]$, there is no \emph{bias} in Theorem \ref{main1} because $\esp [ AB_t ] = 0$.
\end{pathbiased}

\begin{cvrate}
 (convergence rate) At first glance, the convergence rate $\alpha^{-1}$ looks different from the optimal rate of convergence $n^{1/2}$ we obtain in the no-endogeneity case. This is merely a change of perspective because we are looking from the tick size point-of-view. Actually, if for $k=1,2$ we define $N_{t,\alpha}^{(k)}$ as the number of observations before $t$ of the $k$th asset and the sum of observations of both processes 
$N_{t,\alpha}^{(S)} := N_{t,\alpha}^{(1)} + N_{t,\alpha}^{(2)}$, we have that $N_{t,\alpha}^{(S)}$ is exactly of order  $O_p (\alpha^{-2})$. 
Thus, if we define the expected number of observations $n := \esp \big[ N_{t,\alpha}^{(S)} \big]$, we obtain the optimal rate of convergence $n^{\frac{1}{2}}$ in (\ref{theorem}).
\end{cvrate}

\begin{jumps}
\label{rkjumps}
(robustness to jumps in price processes) We assume that we add a jump component to the price process
\begin{eqnarray}
 \label{jumpmodel}
 d X_t^{(k)} = \mu_t^{(k)} dt + \sigma_t^{(k)} dB_t^{(k)} + dJ_t^{(k)}
\end{eqnarray}
for $k=1,2$, where $J_t$ denotes a $2$-dimensional finite activity jump process and $dJ_t^{(k)}$ is either zero (no jump) or a real number indicating the size of the jump at time $t$. We follow exactly the setting of p. 2 in Andersen et al. (2012). We assume that $J_t$ is a general Poisson process independent of the other quantities. Under the same assumptions the conclusion of Theorem \ref{main1} remains valid. The proof can be found in Appendix \ref{proofjumps}. The infinitely many jumps case is complex and beyond the scope of this paper. This was already the case in the $1$-dimensional case (see Remark $4$ in p. $586$ of Li et al. (2014)).
\end{jumps}

\begin{logprice}
(grid on the original non-log scale) Theorem \ref{main1} covers the particular case where $X_t$ corresponds to the log-price and observations are obtained when the price on the original scale hits a boundary. This can be done by a reparametrization of $g_t^{(k)}$ by $\tilde{\tilde{g}}_t^{(k)} (s) := (- \text{exp} (- d_t^{(k)}), \text{exp} (u_t^{(k)}))$.
\end{logprice}
\begin{mainrk}
 (arbitrary number of assets) The authors chose for simplicity to work only with two assets, but they conjecture that this result would stay true for an arbitrary 
 number of assets, and that our proofs would adapt to show it, at the cost of more involved notations and definitions. 
\end{mainrk}

\subsection{Definition of the bias-corrected HY estimator}
Assume that we have a consistent estimator\footnote{$\widehat{AB}_{t,\alpha}$ is consistent means that 
$\alpha^{-1} \widehat{AB}_{t,\alpha} 
= \alpha^{-1} AB_{t,\alpha} + o_p(1)$}
$\widehat{AB}_{t,\alpha}$ of the bias 
$AB_{t,\alpha} := \alpha AB_t$. Such estimator 
will be provided in Section $5$. We define the new estimator $\langle \widehat{ X^{(1)}, X^{(2)}} \rangle_{t,\alpha}^{BC}$  
of high-frequency covariance as the estimate obtained when removing the bias estimate $\widehat{AB}_{t,\alpha}$
from the Hayashi-Yoshida estimator
\begin{eqnarray}
 \label{biascorrected} \langle \widehat{ X^{(1)}, X^{(2)}} \rangle_{t,\alpha}^{BC} := \langle  \widehat{ X^{(1)}, X^{(2)}} \rangle_{t,\alpha}^{HY} - 
 \widehat{AB}_{t,\alpha} .
\end{eqnarray}
With the bias-corrected estimator $\langle \widehat{ X^{(1)}, X^{(2)}} \rangle_{t,\alpha}^{BC}$, we get rid of the 
asymptotic bias and keep the same asymptotic variance as we can see in the 
following corollary.
\begin{maincor}
 Assume $(A1)-(A4)$. Then, stably in law as $\alpha \rightarrow 0$, 
 \begin{eqnarray}
 \label{theoremcor}
 \alpha^{-1} \left(\langle \widehat{X^{(1)}, X^{(2)}} \rangle_{t,\alpha}^{BC} - 
 \langle X^{(1)}, X^{(2)} \rangle_{t} \right) \rightarrow \int_0^t \left( 
 AV_s \right)^{1/2} d Z_s .
 \end{eqnarray}
\end{maincor}

\subsection{Computation of the theoretical asymptotic bias and asymptotic variance}
We warn the reader interested in implementing the bias-corrected estimator that this section is highly technical and we advise her to go directly to Section \ref{estimation} and refer to this section only for the definitions. On the contrary, if the reader wants to understand the main ideas of the proofs, she should take this section as a reference. We also want to emphasize on the fact that the theoretical values of asymptotic bias and asymptotic variance found at the end of this section are rather abstract and don't shed easily light on how the change of parameters $\sigma_t$ and $g_t$ in the model would influence the asymptotic bias and asymptotic variance. The main purpose of this paper is that we don't need to know the theoretical values in order to compute the estimators in Section \ref{estimation}.

\bigskip
We need to introduce some definitions in order to compute the theoretical asymptotic bias $AB_t$ and the 
asymptotic variance term  $AV_t$.
We first need to rewrite the HY estimator (\ref{HY0}) in a different way. For any positive integer $i$, consider the 
$i$th sampling time of the first asset $\tau_{i-1,\alpha}^{(1)}$. We define two random times, $\tau_{i-1,\alpha}^{-}$ 
and $\tau_{i-1,\alpha}^{+}$, which are functions of $\tau_{i-1,\alpha}^{(1)}$ and all the observation times of 
the second asset $\{ \tau_{j,\alpha}^{(2)} \}_{j \geq 0}$, and which correspond respectively to the closest sampling time of the second asset that is strictly smaller than 
$\tau_{i-1,\alpha}^{(1)}$\footnote{Connoisseurs will have noticed that $\tau_{i-1,\alpha}^{-}$ is not a 
$\mathcal{F}_t$-stopping time, which will not be a problem in the proofs}, and the closest sampling time of the second asset that is (not necessarily strictly) 
bigger than $\tau_{i-1,\alpha}^{(1)}$ as
\begin{eqnarray}
\label{tau-0s} \tau_{0,\alpha}^{-} & = & 0, \\
\label{tau-0} \tau_{i-1,\alpha}^{-} & = & \max \{ \tau_{j,\alpha}^{(2)} : \tau_{j,\alpha}^{(2)} < \tau_{i-1,\alpha}^{(1)} \} \text{ for } i \geq 2, \\
\label{tau+0} \tau_{i-1,\alpha}^{+} & = & \min \{ \tau_{j,\alpha}^{(2)} : \tau_{j,\alpha}^{(2)} \geq \tau_{i-1,\alpha}^{(1)} \} \text{ for } i \geq 1.
\end{eqnarray}
We consider $\Delta X_{\tau_{i,\alpha}^{-,+}}^{(2)}$ the increment of the second asset between 
$\tau_{i-1,\alpha}^{-}$ and $\tau_{i,\alpha}^{+}$
\begin{eqnarray}
\label{inc-+0} \Delta X_{\tau_{i,\alpha}^{-,+}}^{(2)} := \Delta X_{[\tau_{i-1,\alpha}^{-},\tau_{i,\alpha}^{+}]}^{(2)}. 
\end{eqnarray}
Rearranging the terms in (\ref{HY0}) gives us (except for a few terms at the 
edge)
\begin{eqnarray}
\label{HY01} \langle \widehat{X^{(1)}, X^{(2)}} \rangle_{t,\alpha} = \sum_{\tau_{i,\alpha}^{+} < t} \Delta X_{\tau_{i,\alpha}^{(1)}}^{(1)} 
\Delta X_{\tau_{i,\alpha}^{-,+}}^{(2)} . 
\end{eqnarray}

The representation in (\ref{HY01}) is very useful in the sense that it gives a natural order between the terms 
in the sum. Nevertheless, any term of this sum is a priori correlated with the other terms. 
We will rearrange once again the terms in (\ref{HY01}), so that each term is only correlated with the previous and 
the next term of the sum. In this case, we say that they are \emph{1-correlated}. For this purpose, we need to introduce 
some notation. 
We remind the reader that $\Tau_{\alpha}$ is the two-dimensional vector of sampling times, where for each $k=1,2$ the $k$th 
component $\Tau_{\alpha}^{(k)}$ is equal to the sequence of sampling times associated with the $k$th asset. We will construct 
a subsequence $\Tau_{\alpha}^{1C}$ of 
$\Tau_{\alpha}^{(1)}$ that also depends on the observation times of the second asset $\Tau_{\alpha}^{(2)}$, and will be 
such that we can write the Hayashi-Yoshida estimator as a 1-correlated sum similar to (\ref{HY01}), except the 
new sampling times $\tau_{i,\alpha}^{1C}$ will replace the original observation times $\tau_{i,\alpha}^{(1)}$. The new sampling 
times $\tau_{i,\alpha}^{1C}$ are obtained using the following algorithm. We define 
$\tau_{0,\alpha}^{1C} := \tau_{0,\alpha}^{(1)}$, 
and recursively for $i$ any nonnegative integer
\begin{eqnarray}
\label{algo1C} \tau_{i+1,\alpha}^{1C} := \min \big\{ \tau_{u,\alpha}^{(1)} : \text{ there exists } j \in \naturels \text{ such that } 
\tau_{i,\alpha}^{1C} \leq \tau_{j,\alpha}^{(2)} < \tau_{u,\alpha}^{(1)} \big\} .
\end{eqnarray}
In words, if we sit at the observation time $\tau_{i,\alpha}^{1C}$ of the first asset, 
we wait first to hit an observation time of the second asset, and we then choose 
the next strictly bigger observation time of the first asset. In analogy with (\ref{tau-0s}), (\ref{tau-0}), (\ref{tau+0}) and (\ref{inc-+0}), we define the following times
\begin{eqnarray}
\label{tau1C-0s} \tau_{0,\alpha}^{1C,-} & := & 0, \\
\label{tau1C-0} \tau_{i-1,\alpha}^{1C,-} & := & \max \{ \tau_{j,\alpha}^{(2)} : \tau_{j,\alpha}^{(2)} < \tau_{i-1,\alpha}^{1C} \} \text{ for } i \geq 2\\
\label{tau1C+0} \tau_{i-1,\alpha}^{1C,+} & := & \min \{ \tau_{j,\alpha}^{(2)} : \tau_{j,\alpha}^{(2)} \geq \tau_{i-1,\alpha}^{1C} \} \text{ for } i \geq 1, \\
\label{inc1C-+0} \Delta X_{\tau_{i,\alpha}^{1C,-,+}}^{(2)} & := & \Delta X_{[\tau_{i-1,\alpha}^{1C,-},\tau_{i,\alpha}^{1C,+}]}^{(2)} \text{ for } i \geq 1.
\end{eqnarray}
First, observe that, except for maybe a few terms at the edge, we can rewrite (\ref{HY01}) as
\begin{eqnarray}
\label{HY02} \widehat{\langle X^{(1)}, X^{(2)} \rangle_{t,\alpha}} = \sum_{\tau_{i,\alpha}^{1C,+} < t} \Delta X_{\tau_{i,\alpha}^{1C}}^{(1)} 
\Delta X_{\tau_{i,\alpha}^{1C,-,+}}^{(2)} .
\end{eqnarray}
Also, we define the following compensated increments of the HY estimator
\begin{eqnarray}
\label{compensated} N_{i,\alpha} = \Delta X_{\tau_{i,\alpha}^{1C}}^{(1)} \Delta X_{\tau_{i,\alpha}^{1C,-,+}}^{(2)} - 
\int_{\tau_{i-1,\alpha}^{1C}}^{\tau_{i,\alpha}^{1C}} \zeta_s^{1,2} ds .
\end{eqnarray}
Note that they are compensated in the sense that they are centered (if we decompose $\Delta X_{\tau_{i,\alpha}^{1C,-,+}}^{(2)}$ 
into a left ($-$), a central and a right ($+$) part and condition the expectation, this is straightforward to show). Similarly, we can show that 
they are 1-correlated. 

\bigskip
The idea of the proof is the following. If we consider the volatility matrix $\sigma_t$ and the grid function $g_t$ 
to be constant over time, we can express the conditional returns of the normalized error of HY as a homogeneous Markov chain 
(of order 1), show that the Markov chain is uniformly ergodic and thus use results in the limit theory of Markov chains 
(see, e.g., Meyn and Tweedie (2009)) to show that it has a stationary distribution. Then, we prove that we can approximate locally the 
returns of the normalized error when the volatility matrix and grid function are not constant by the returns when 
holding them constant on a small block. Finally, using limit theory techniques developed in Mykland and Zhang (2012) together 
with standard probability results of conditional distribution (see, e.g., Breiman (1992)), we can bound uniformly in 
time the error of the returns when holding the volatility matrix and grid function constant. 

\bigskip
Based on the definitions introduced in Appendix \ref{appdef}, we can define the \emph{instantaneous variance} of the normalized HY estimate's error (\ref{psiAV1}), which 
depends on the volatility matrix $\tilde{\sigma}$ and the grid $\tilde{g}$. 
Similarly, we also define the \emph{instantaneous covariance} between the normalized HY's error and the 
first asset price (\ref{psiAC1}), and the \emph{instantaneous covariance} between the error and the 
second asset price (\ref{psiAC2}). Finally, we define the \emph{instantaneous 1-correlated time}, which is the 
approximation of $\esp_{\tau_{i,n}^{1C}} \left[ \Delta \tau_{i+2}^{1C} \right]$, where if $\tau$ is a 
$\left( \mathcal{F}_t \right)$-stopping time, 
$\esp_{\tau} \left[ Y \right]$ is defined as the conditional distribution of $Y$ given $\mathcal{F}_{\tau}$.
\begin{eqnarray}
\label{psiAV1} \psi^{AV} (\tilde{\sigma}, \tilde{g}, x, u) & := & \esp 
\big[ \tilde{N}_2^2 + 2 \tilde{N}_2 \tilde{N}_3 \big],\\
\label{psiAC1} \psi^{AC1} (\tilde{\sigma}, \tilde{g}, x, u) & := & \esp \big[ 
\tilde{N}_2 \Delta \tilde{X}_{\tilde{\tau}_2^{1C}}^{(1)} \big] ,\\
\label{psiAC2} \psi^{AC2} (\tilde{\sigma}, \tilde{g}, x, u) & := & \esp \big[ 
\tilde{N}_2 \Delta \tilde{X}_{\tilde{\tau}_2^{1C,-,+}}^{(2)} \big] ,\\
\label{psitau} \psi^{\tau} (\tilde{\sigma}, \tilde{g}, x, u) & := & \esp \big[ \Delta \tilde{\tau}_2^{1C} \big].
\end{eqnarray}
\begin{N2start}
The reader might expect $\tilde{N}_1$ in lieu of $\tilde{N}_2$ in (\ref{psiAV1}), (\ref{psiAC1}), (\ref{psiAC2}) and (\ref{psitau}). Actually, we cannot use $\tilde{N}_1$ directly from the definition because the corresponding time $\tilde{\tau}_{0}^{1C,-} = 0$. We would need to set it to $-u$ to alter the definition of (\ref{psiAV1}), (\ref{psiAC1}), (\ref{psiAC2}) and (\ref{psitau}), which we have chosen not to do for the sake of clarity.
\end{N2start}
Set $\tilde{Z}_0 := (x,u)$ and for any positive integer $i$ 
\begin{eqnarray}
\label{Zi} \tilde{Z}_i := \big( \Delta \tilde{X}_{[\tilde{\tau}_{i}^{1C,-}, \tilde{\tau}_{i}^{1C}]}^{(4)}, 
\tilde{\tau}_{i}^{1C} - \tilde{\tau}_{i}^{1C,-} \big) .
\end{eqnarray}
For any nonnegative integer $i$, we consider $\tilde{\pi}_{i} ( \tilde{\sigma}, \tilde{g}, x, u )$ the distribution of 
$\tilde{Z}_i$. We also introduce the notation $\Pi (\tilde{\sigma}, \tilde{g}, x, u ) := \{ 
\tilde{\pi}_{i} ( \tilde{\sigma}, \tilde{g}, x, u ) \}_{i \geq 0}$. By the strong Markov property of Brownian motion, we can show that 
$\tilde{Z}_i$ is a homogeneous Markov chain (of order $1$) on the state space 
$\mathcal{S}_{\tilde{g}}$. 
In the following lemma, we show that there exists a stationary distribution of 
$\tilde{\pi}_{i} ( \tilde{\sigma}, \tilde{g}, x, u )$.
\begin{distribexists}
\label{distribexists}
Let $c := (g^-, g^+, K, d)$ be a four-dimensional vector such that 
$c \in \mathcal{C}$ and consider $\tilde{\sigma}$ a constant volatility matrix such that $\tilde{\lambda}^{\min} > 0$ and 
$\tilde{g} \in \mathcal{G} (c) $ a constant grid.
Then, there exists a stationary distribution $\tilde{\pi} (\tilde{\sigma} , \tilde{g} )$. 
\end{distribexists}
The proof of Lemma \ref{distribexists} can be found in the Appendix (proof of Lemma \ref{approxdistrib}). The next definition is the average (regarding the 
stationary distributions) of the instantaneous variance, covariances and 1-correlated time. 
For any $\theta \in \{\text{AV, AC1, AC2, } \tau \}$,
$$\phi^{\theta} \left(\tilde{\sigma}, \tilde{g} \right) := \int_{\reels^2} \psi^{\theta} \left(\tilde{\sigma}, 
\tilde{g}, y, v \right) d \tilde{\pi} 
\left(\tilde{\sigma}, \tilde{g} \right) \left( y, v \right) .$$
We introduce the notation $\phi^{\theta}_s := \phi^{\theta} \left( \sigma_s, g_s \right)$ and consider the following quantities needed to compute the asymptotic bias and variance.
\begin{eqnarray}
\label{k1} k_s^{(1)} & := & \big( \sigma_s^{(1)} \big)^{-2} \phi_s^{AC1} \big( \phi_s^{\tau} \big)^{-1} ,\\
\label{kperp} k_s^{1,\perp} & := & \big(1 - ( \rho_s^{1,2} )^2 \big)^{-1} \big( ( \sigma_s^{(2)} )^{-2} \phi_s^{AC2}
- (\sigma_s^{(1)} \sigma_s^{(2)} )^{-1} \rho_s^{1,2} \phi_s^{AC1} \big) \big( \phi_s^{\tau} \big)^{-1} .
\end{eqnarray}
We express now $AV_s$ the quantity integrated to obtain the asymptotic variance.
\begin{eqnarray}
\label{AV} AV_s & := & \big( \phi_s^{AV} + 2 \big( k_s^{(1)} (\sigma_s^{(1)} )^{-1} 
\sigma_s^{(2)} \rho_s^{1,2} \phi_s^{AC1} 
- ( k_s^{(1)} + k_s^{1, \perp} ) \phi_s^{AC2} \big) \big) \big( \phi_s^{\tau} \big)^{-1}
\end{eqnarray}
\begin{eqnarray*}
& & + \big( \sigma_s^{(1)} \big)^2 \big( k_s^{(1)} \big)^2 
+ \big( \sigma_s^{(2)} \big)^2 \Big( 1 - ( \rho_s^{1,2} )^2 \big) \big( k_s^{1, \perp} \big)^2 .
\end{eqnarray*}
The asymptotic bias is defined as $AB_t := \int_0^t AB_s^{(1)} dX_s^{(1)} + \int_0^t AB_s^{(2)} dX_s^{(2)}$ where
\begin{eqnarray}
\label{AB1} AB_s^{(1)} & := & k_s^{(1)} - k_s^{1, \perp} \rho_s^{1,2} \sigma^{(2)}_{s} \big( \sigma^{(1)}_{s} \big)^{-1}, \\
\label{AB2} AB_s^{(2)} & := & k_s^{1, \perp}.
\end{eqnarray}

\begin{asymptoticbias}
(asymptotic bias) Looking at the expressions for $AB_s^{(1)}$ and $AB_s^{(2)}$, one can be tempted to think that because of the 
$\big(1 - ( \rho_s^{1,2} )^2 \big)^{-1}$ term in $k_s^{1, \perp}$, the bias 
will increase drastically when both assets are highly correlated. In this case, the reader should keep in mind that the second term of $AB_s^{(1)}$, when 
integrated with respect to $X_s^{(1)}$,
and $AB_s^{(2)}$, when integrated with respect to $X_s^{(2)}$, will be roughly of the 
same magnitude, with opposite signs, and thus there is no explosion of asymptotic bias. 
We chose the above asymptotic bias' representation because it is straightforward to build estimators from it. 
We can also express 
the asymptotic bias differently. For this purpose, we can rewrite the log-price process as 
\begin{eqnarray*}
      dX_t^{(1)} & = & \sigma^{(1)}_{t} dB_{t}^{(1)},\\
      dX_t^{(2)} & = & \rho_t^{1,2} \sigma^{(2)}_{t} dB_{t}^{(1)} + 
      \big( 1 - ( \rho_t^{1,2} )^2 \big)^{1/2} \sigma^{(2)}_{t} dB_{t}^{1,\perp},
\end{eqnarray*}
where $B_t^{(1)}$ and $B_t^{1,\perp}$ are independent Brownian motions. Let 
\begin{eqnarray}
\label{X1perp} dX_t^{1,\perp} = \big( 1 - ( \rho_t^{1,2} )^2 \big)^{1/2} \sigma^{(2)}_{t} dB_{t}^{1,\perp}
\end{eqnarray}
be the part of $X_t^{(2)}$ that is not correlated with $X_t^{(1)}$. We can express the asymptotic bias as 
$AB_t = \int_0^t \tilde{AB}_s^{(1)} dX_s^{(1)} + 
 \int_0^t \tilde{AB}_s^{(2)} dB_s^{1,\perp}$. In this case, $\tilde{AB}_s^{(1)} = k_s^{(1)}$ and 
$$\tilde{AB}_s^{(2)} =\underset{n \rightarrow \infty}{\lim} 
 \langle M^n, B^{1,\perp} \rangle_s$$ 
where $M^n$ is defined in the proofs. We can show that this limit exists, 
 and does not explode when both assets are highly correlated.
\end{asymptoticbias}

\section{Estimation of the bias and variance} \label{estimation}
We need to introduce some new notations. We recall that $N_{1,\alpha}^{(1)}$ is the number of observations corresponding to the first asset before $1$ and 
we also define $N_{1,\alpha}^{1C}$ the number of 1-correlated observations before 1, i.e. 
$N_{1,\alpha}^{1C} := \max \{ i \in \naturels \text{ s.t. } \tau_{i,\alpha}^{1C} < 1 \}$. 
In practice, the first step is to transform the returns of the first asset 
$$\big\{ (\Delta X_{\tau_{i,\alpha}^{(1)}}^{(1)}, \Delta \tau_{i,\alpha}^{(1)} ) 
\big\}_{i=1}^{N_{1,\alpha}^{(1)}}$$ into 1-correlated returns
$$\big\{ (\Delta X_{\tau_{i,\alpha}^{1C}}^{1C}, \Delta \tau_{i,\alpha}^{1C} ) 
\big\}_{i=1}^{N_{1,\alpha}^{1C}}$$ 
using algorithm (\ref{algo1C}). Then, for each asset, we will chop the data into $B_n$ blocks and on each block
$i=1, \ldots, B_n$ we will estimate $\widehat{AV}_{i,\alpha}$, $\widehat{AB}_{i,\alpha}^{(1)}$ and 
$\widehat{AB}_{i,\alpha}^{(2)}$, 
pretending that the volatility matrix $\sigma_t$ and grid $g_t$ are block-constant. 

\bigskip
Because there is asynchronicity in 
the observation times, the blocks of each asset are not exactly equal. Let $h_n$ be the block size. For the first asset, 
we consider \emph{block} $1^{(1)} := [0, \tau_{h_n,\alpha}^{1C}]$, \emph{block} $2^{(1)} := 
[\tau_{h_n,\alpha}^{1C}, \tau_{2h_n,\alpha}^{1C}]$, 
etc. For the second asset, 
we let \emph{block} $1^{(2)} := [\tau_{0,\alpha}^{1C,+}, \tau_{h_n,\alpha}^{1C,+}]$, \emph{block} 
$2^{(2)} := [\tau_{h_n,\alpha}^{1C,+}, \tau_{2 h_n,\alpha}^{1C,+}]$, 
etc. In the following, 
we will say 
$j \in \text{block }i^{(1)}$ when $\tau_{j,\alpha}^{(1)} \in \text{block } i^{(1)}$. Similarly, we say
$j \in \text{block }i^{(2)}$ when $\tau_{j,n}^{(2)} \in \text{block } i^{(2)}$. Finally, we define 
$j \in \text{block }i$ if $j \in \{ (i-1) h_n + 1, \ldots, i h_n \} $.
First, we estimate the volatility of 
both assets using the corrected estimator in Li et al. (2014). To do this, we need to define an estimate of the spot volatility on each block for each asset $k=1,2$ by
\begin{eqnarray*}
 \tilde{\sigma}_{i,\alpha}^{(k)} & := & 
 \Big( \sum_{j \in \text{block } i^{(k)}} (\Delta X_{\tau_{j,\alpha}^{(k)}}^{(k)})^2 \Big)^{1/2}.
\end{eqnarray*}
Then, we estimate the asymptotic bias of the volatility via
\begin{eqnarray*}
 \widehat{AB\sigma}_{i,\alpha}^{(k)} & = & \frac{2}{3 (\tilde{\sigma}_{i,\alpha}^{(k)})^2} 
 \sum_{j \in \text{block } i^{(1)}} (\Delta X_{\tau_{j,\alpha}^{(k)}}^{(k)})^3 .
\end{eqnarray*}
We obtain the bias-corrected estimators of volatility on each block:
\begin{eqnarray*}
 \hat{\sigma}_{i,\alpha}^{(k)} & = & \tilde{\sigma}_{i,\alpha}^{(k)} - \widehat{AB\sigma}_{i,\alpha}^{(k)} .
\end{eqnarray*}
Then, we estimate the correlation between both assets using the \emph{naive} HY estimator
\begin{eqnarray*}
 \hat{\rho}_{i,\alpha}^{1,2} & = & \frac{1}{\hat{\sigma}_{i,\alpha}^{(1)} \hat{\sigma}_{i,\alpha}^{(2)}} 
 \sum_{j \in \text{block } i} \Delta X_{\tau_{j,\alpha}^{1C}}^{(1)} 
\Delta X_{\tau_{j,\alpha}^{1C,-,+}}^{(2)} .
\end{eqnarray*}
We then build an estimator of the compensated increments of the HY estimator, 
following the definition in (\ref{compensated}),
\begin{eqnarray*}
   \widehat{N}_{i,\alpha} & = & \Delta X_{\tau_{i,\alpha}^{1C}}^{(1)} \Delta X_{\tau_{i,\alpha}^{1C,-,+}}^{(2)} - 
   \Delta \tau_{i,\alpha}^{1C} 
   \hat{\sigma}_{i,\alpha}^{(1)} 
   \hat{\sigma}_{i,\alpha}^{(2)} \hat{\rho}_{i,\alpha}^{1,2}.
\end{eqnarray*}
The next step is to estimate the instantaneous variance (\ref{psiAV1}), 
both instantaneous covariances (\ref{psiAC1}) and (\ref{psiAC2}) and the instantaneous 1-correlated time
(\ref{psitau}) on each block. This is done by taking the sample average of the corresponding estimated quantities. 
Note that we don't directly estimate $\psi^{AV}$, $\psi^{AC1}$, $\psi^{AC2}$ and $\psi^{\tau}$, but rather a scaling 
version of them, i.e. $\alpha_n^2 \psi^{AV}$, $\alpha_n \psi^{AC1}$, $\alpha_n \psi^{AC2}$ and $\alpha_n \psi^{\tau}$. 
In practice, we can always assume $\alpha_n := 1$ by scaling $g_t$ by the tick size, 
and thus we match the definitions of the following estimators with (\ref{psiAV1})-(\ref{psitau}). For the sake of 
simplicity, we assume that the number of 1-correlated observations of the last block $B_n$ is also $h_n$. In practice, 
this will be most likely different from $h_n$, and thus the denominator of (\ref{estphiAV})-(\ref{estphitau}) will 
have to be changed so that it is equal to the number of 1-correlated observations in this last block. The estimates are given by
\begin{eqnarray}
\label{estphiAV} \hat{\phi}_{i,\alpha}^{AV} & := & h_n^{-1} 
      \sum_{j \in \text{block } i} \hat{N}_{j,\alpha}^2 + 2 \hat{N}_{j, \alpha} \hat{N}_{j+1, \alpha}, \\ 
\hat{\phi}_{i,\alpha}^{AC1} & := &  h_n^{-1} 
      \sum_{j \in \text{block } i} \hat{N}_{j,\alpha} \Delta X_{\tau_{j,\alpha}^{1C}}^{(1)}, \\
\hat{\phi}_{i,\alpha}^{AC2} & := &  h_n^{-1} 
      \sum_{j \in \text{block } i} \hat{N}_{j,\alpha} \Delta X_{\tau_{j,\alpha}^{1C,-,+}}^{(2)}, \\
\label{estphitau} \hat{\phi}_{i,\alpha}^{\tau} & := &  h_n^{-1} 
      \sum_{j \in \text{block } i}  \Delta \tau_{j,\alpha}^{1C}.
\end{eqnarray}
We estimate now the quantities (\ref{k1}) and (\ref{kperp}) as
\begin{eqnarray}
      \hat{k}_{i,\alpha}^{(1)} & := & \big( \hat{\sigma}_{i,\alpha}^{(1)} \big)^{-2} \hat{\phi}_{i, \alpha}^{AC1} 
\big( \hat{\phi}_{i, \alpha}^{\tau} 
\big)^{-1},\\
\hat{k}_{i,\alpha}^{1,\perp} & := & \big(1 - ( \hat{\rho}_{i,\alpha}^{1,2} )^2 \big)^{-1} \big( 
( \hat{\sigma}_{i, \alpha}^{(2)} )^{-2} \hat{\phi}_{i,\alpha}^{AC2}
- (\hat{\sigma}_{i,\alpha}^{(1)} \hat{\sigma}_{i,\alpha}^{(2)} )^{-1} \hat{\rho}_{i,\alpha}^{1,2} 
\hat{\phi}_{i, \alpha}^{AC1} \big) \big( \hat{\phi}_{i,\alpha}^{\tau} \big)^{-1}.
\end{eqnarray}
We follow (\ref{AB1}) and (\ref{AB2}) to estimate the bias integrated terms $AB_{s}^{(1)}$ and $AB_{s}^{(2)}$ on each block 
\begin{eqnarray*}
\widehat{AB}_{i,\alpha}^{(1)} & := & \hat{k}_{i,\alpha}^{(1)} - \hat{k}_{i,\alpha}^{1, \perp} \hat{\rho}_{i,\alpha}^{1,2} 
\hat{\sigma}^{(2)}_{i, \alpha} \big( \hat{\sigma}^{(1)}_{i, \alpha} \big)^{-1}, \\
\widehat{AB}_{i,\alpha}^{(2)} & := & \hat{k}_{i,\alpha}^{1, \perp} .  
\end{eqnarray*}
For the 
variance term $AV_s$, we decide not to use the direct definition in (\ref{AV}) because it can provide negative 
estimates. Instead, we will be using the following estimator
\begin{eqnarray*}
      \widehat{AV}_{i,\alpha} & := & \Big( \big( \sum_{j \in \text{block } i}  \widehat{N}_{j,\alpha} \big) - 
      \hat{k}_{i,\alpha}^{(1)} (X_{\tau_{i h_n, \alpha}^{1C}}^{(1)} - X_{\tau_{(i-1) h_n, \alpha}^{1C}}^{(1)})  \\
&  & - \hat{k}_{i,\alpha}^{\perp} \big((X_{\tau_{i h_n, \alpha}^{1C,+}}^{(2)} - X_{\tau_{(i-1) h_n, \alpha}^{1C,+}}^{(2)}) 
 - \hat{\rho}_{i,\alpha}^{1,2} \hat{\sigma}_{i,\alpha}^{(2)} (\hat{\sigma}_{i,\alpha}^{(1)})^{-1} 
 (X_{\tau_{i h_n, \alpha}^{1C}}^{(1)} - X_{\tau_{(i-1) h_n, \alpha}^{1C}}^{(1)}) \big) \Big)^2.
 \end{eqnarray*}
We define the final estimators of asymptotic bias and asymptotic variance as
\begin{eqnarray}
\label{estAB} \widehat{AB}_{\alpha} & := & \sum_{i=1}^{B_n} \widehat{AB}_{i,\alpha}^{(1)} 
\big( X_{\tau_{i h_n, \alpha}^{1C}}^{(1)} - X_{\tau_{(i-1) h_n, \alpha}^{1C}}^{(1)} \big)
 + \widehat{AB}_{i,\alpha}^{(2)} \big( X_{\tau_{i h_n, \alpha}^{1C,+}}^{(2)} - 
 X_{\tau_{(i-1) h_n, \alpha}^{1C,+}}^{(2)} \big),\\
\label{estAV} \widehat{AV}_{\alpha} & := & \sum_{i=1}^{B_n} \widehat{AV}_{i,\alpha} 
\big( \tau_{i h_n,\alpha}^{1C} - \tau_{(i-1) h_n,\alpha}^{1C} \big).
\end{eqnarray}
As a corollary of Theorem $1$, we obtain the following result, which states the consistency of 
(\ref{estAB}) and (\ref{estAV}).
\begin{estcor}
\label{estcor}
There exists a choice of the block size $h_n$\footnote{the exact assumptions on $h_n$ can be found in the proofs of 
Theorem 1} such that when 
$\alpha \rightarrow 0$, we have
\begin{eqnarray}
 \alpha^{-1} \widehat{AB}_{\alpha} & \overset{\proba}{\rightarrow} & AB_1,\\
 \alpha^{-2} \widehat{AV}_{\alpha} & \overset{\proba}{\rightarrow} & \int_0^1 AV_s ds .
\end{eqnarray}
In particular, in view of Corollary 2, the bias-corrected estimator $\langle \widehat{ X^{(1)}, X^{(2)}} 
\rangle_{1,\alpha}^{BC} := \langle  \widehat{ X^{(1)}, X^{(2)}} \rangle_{1,\alpha}^{HY} - 
 \widehat{AB}_{\alpha}$ is such that 
\begin{eqnarray}
\label{feasiblestat}
 \frac{\langle \widehat{X^{(1)}, X^{(2)}} \rangle_{1,\alpha}^{BC} - 
 \langle X^{(1)}, X^{(2)} \rangle_{1}}{\widehat{AV}_{\alpha}^{1/2}} \rightarrow \mathcal{N} (0,1).
\end{eqnarray}

\end{estcor}
\begin{estrk1} (exchanging $X_t^{(1)}$ and $X_t^{(2)}$) When estimating the asymptotic bias and the asymptotic variance, 
we considered one specific asset to be $X_t^{(1)}$ and the other one to be $X_t^{(2)}$. We could exchange 
$X_t^{(1)}$ and $X_t^{(2)}$, and find new estimators $\tilde{AB}_{\alpha}$ and $\tilde{AV}_{\alpha}$ according 
to the previous definitions. One could then take $\frac{ AB_{\alpha} + \tilde{AB}_{t,\alpha}}{2}$ 
(respectively $\frac{ AV_{\alpha} + \tilde{AV}_{t,\alpha}}{2}$) as final estimators of asymptotic bias (asymptotic variance).
\end{estrk1}
\begin{estrk2} (optimal block size) In practice, the optimal block size $h_n$ is not straightforward to choose. On the one hand, 
$h_n$ should be as small as possible so that the volatility matrix $\sigma_t$ and the grid $g_t$ are almost 
constant on each block, and thus (\ref{estphiAV})-(\ref{estphitau}) are less biased.
On the other hand, we need as many observations as we can on each block, so that the variance of approximations 
(\ref{estphiAV})-(\ref{estphitau}) is not too big. We are facing here the usual bias-variance tradeoff.
\end{estrk2}

\section{Numerical simulations}
We consider four different settings in this part. We describe here the first one. We assume the same setting as the toy model described in Example \ref{hittingbarrier}, in two dimensions. Thus, there exists 
a four-dimensional parameter $\theta := (\theta_u^{(1)}, \theta_d^{(1)}, \theta_u^{(2)}, \theta_d^{(2)})$ such that 
for any $t \geq 0$ and any $s \geq 0$, $u_t^{(1)} (s) := \theta_u^{(1)}$, $d_t^{(1)} (s) := \theta_d^{(1)}$, 
$u_t^{(2)} (s) := \theta_u^{(2)}$ and $d_t^{(2)} (s) := \theta_d^{(2)}$. We assume that the two-dimensional price 
process $(X_t^{(1)}, X_t^{(2)})$ has a null-drift. Also, we assume that the volatility of the first process is $\sigma_t^{(1)} := \tilde{\tilde{\sigma}}^{(1)} $ where $\tilde{\tilde{\sigma}}^{(1)} := 0.016$ and 
the volatility of the second process $\sigma_t^{(2)} := \tilde{\tilde{\sigma}}^{(2)}$ where $\tilde{\tilde{\sigma}}^{(2)} := 0.02$, and that the correlation between both 
assets is $\rho_t^{1,2} := 0.2$. We set $\theta := \big( 0.0007 , 0.0001, 0.0006, 0.0001 \big)$. According to this rule, 
a change of price occurs whenever the price of the first (respectively second) asset increases by $0.07 \%$ ($0.06 \%$) 
or decreases by $0.01 \%$ ($0.01 \%$). Finally, we assume that the price processes $(X_t^{(1)}, X_t^{(2)})$ and the time 
processes $(X_t^{(t,1)}, X_t^{(t,2)})$ are equal. 

\bigskip
The second setting is similar to the first setting, except that we assume now a stochastic volatility Heston model. Specifically, we assume that 
\begin{eqnarray*}
dX_t^{(k)} & := & \mu^{(k)} dt + \sigma_t^{(k)} dB_t^{(k)},\\
d (\sigma_t^{(k)})^2 & := & \kappa^{(k)} \big( (\tilde{\tilde{\sigma}}^{(k)})^2 - (\sigma_t^{(k)})^2 \big) dt + \delta^{(k)} \sigma_t^{(k)} d \tilde{\tilde{B}}_t^{(k)},
\end{eqnarray*}
where the constant high-frequency covariance between $B_t^{(k)}$ and $\tilde{\tilde{B}}_t^{(k)}$ is fixed to $\tilde{\tilde{\rho}}^{(k)}$, and $(\tilde{\tilde{B}}_t^{(1)}, \tilde{\tilde{B}}_t^{(2)})$ are uncorrelated with each other. We choose to work with drift $(\mu^{(1)}, \mu^{(2)}) := (0.03, 0.02)$, and to add leverage effect $(\tilde{\tilde{\rho}}^{(1)}, \tilde{\tilde{\rho}}^{(2)})$ are selected to be $(-0.8, -0.7)$. Finally, $(\kappa^{(1)}, \kappa^{(2)}) := (4.5,5.5)$, the volatility of volatility $(\delta^{(1)},\delta^{(2)}) := (0.4, 0.5)$, and the volatility starting values $(\sigma_0^{(1)}, \sigma_0^{(2)}) := (\tilde{\tilde{\sigma}}^{(1)}, \tilde{\tilde{\sigma}}^{(2)})$.

\bigskip 
We consider now the third setting, which goes one step further than the previous setting. We assume a jump-diffusion model for both the price and the volatility. Formally, we assume that
\begin{eqnarray*}
dX_t^{(k)} & := & \mu^{(k)} dt + \sigma_t^{(k)} dB_t^{(k)} + dJ_t^{(k)},\\
d (\sigma_t^{(k)})^2 & := & \kappa^{(k)} \big( (\tilde{\tilde{\sigma}}^{(k)})^2 - (\sigma_t^{(k)})^2 \big) dt + \delta^{(k)} \sigma_t^{(k)} d \tilde{\tilde{B}}_t^{(k)} + d\tilde{J}_t^{(k)},
\end{eqnarray*}
where the jumps $(J_t^{(1)}, J_t^{(2)}, \tilde{J}_t^{(1)}, \tilde{J}_t^{(2)})$ follow a $4$-dimensional Poisson process with intensity $(\lambda^{(1)}, \lambda^{(2)}, \tilde{\lambda}^{(1)}, \tilde{\lambda}^{(2)})$ $:= (12,11,10,9)$. The jump sizes are taken to be $1$ or $-1$ with probability $\frac{1}{2}$ for  price processes, and $0.0001$ or $-0.0001$ with half-probability for volatility processes.

\bigskip
In the fourth setting, we consider another model of arrival times, namely Example \ref{uncertaintyzones}. We set the tick size $\alpha= 0.0001$ and the friction parameter $\eta = 0.15$. Price and volatility processes are assumed to follow the same model as in the second setting. 

\bigskip
We simulate price processes and observation times for 10 years of $252$ business days. We choose $h_n = n^{\frac{1}{2}}$ for Settings 2 to 4. We provide in Table \ref{num} a summary of the comparison results between HY and the bias-corrected HY. As expected from the theory, the RMSE is 
improved when using the bias-corrected estimator in Example \ref{hittingbarrier}. In Example \ref{uncertaintyzones}, the bias-corrected HY doesn't seem to perform better than HY. We conjecture that there is no asymptotic bias in Example \ref{uncertaintyzones}, and that this is the reason why we don't observe any difference between the two estimators in that simple model. In addition, the sample bias is almost the same when using HY and 
the bias-corrected estimator for the four different settings, which is also expected from Remark \ref{pathbiased}. Furthermore, this sample bias tends to $0$, which comes from the fact that both estimators are consistent. Finally, the standardized feasible statistic (\ref{feasiblestat}) in the first setting is reported in Table \ref{tablefeasiblestat} and plotted in Figure \ref{plotfeasiblestat}.

  \section{Conclusion}
 We have introduced in this paper the HBT model, and we have shown that it is more general than some of the endogenous models of the literature. This model can be extended to a model including more general noise structure in observations, and even noise in sampling times. This is investigated in Potiron (2016).

 \bigskip
Under this model, we have proved the central limit theorem of the Hayashi-Yoshida estimator. Our main theorem states 
that there is an asymptotic bias. Accordingly, we built a bias-corrected HY estimator. We also computed the theoretical standard deviation, and we provided consistent estimates of it. Numerical simulations corroborate the theory. 

\bigskip
The techniques used for the proof of the main theorem could be applied to more general models and to other problems such as the estimation of the integrated variance of noise, integrated betas, etc. In particular, independence between the efficient price process
and the noise is not needed in the model. As long as we can approximate the joint distribution of the noise and the returns by a Markov chain, ideas of our proof can be used.
 
\section{Appendix}
\subsection{Definition of some quantities of approximation}
\label{appdef}
We define in this section some quantities assuming the volatility matrix $\sigma_t$ and the grid function 
$g_t$ are constant. For that purpose, let 
$\tilde{W}_t$ be a four dimensional Wiener process, $c := (g^-, g^+, K, d)$ a four-dimensional vector such that 
$c \in \mathcal{C}$ and $\tilde{\sigma}$ a constant volatility matrix such that the associated $\tilde{\lambda}^{\min}$, which 
is the analog of $\lambda_t^{\min}$ defined in Section $4.1$ when we replace $\sigma_t$ by $\tilde{\sigma}$, is stritcly bigger than $0$ 
and $\tilde{g} \in \mathcal{G} (c) $ a constant grid function. In analogy with the definition of the grid function $g_t$ in $(A3)$, we assume 
that $\tilde{g}$ can be written in terms of the down and up functions of both assets, i.e. 
$\tilde{g} := ( \tilde{g}^{(1)}, \tilde{g}^{(2)}) $ where for each $k=1,2$ we have $\tilde{g}^{(k)} := (\tilde{d}^{(k)}, 
\tilde{u}^{(k)})$. Also, we 
introduce $\mathcal{S}_{\tilde{g}}$ the subspace of $\reels^2$ defined as
$$\mathcal{S}_{\tilde{g}} := \{(y,v) \in \reels \times \reels^+ \text{ s.t. } \tilde{d}^{(2)} (v) \leq y \leq 
\tilde{u}^{(2)} (v) \}.$$
If we set 
$\tilde{X} = \tilde{\sigma} \tilde{W}$ and the corresponding sampling times of both assets $\tilde{\Tau} := 
(\tilde{\Tau}^{(1)}, \tilde{\Tau}^{(2)})$, where for $k=1,2$ we have $\tilde{T}^{(k)} := 
\{ \tilde{\tau}_{i} \}_{i \geq 0}$, we define the observation times of the first asset  
as $\tilde{\tau}_0^{(1)} := 0$ and recursively for $i$ any positive integer
$$\tilde{\tau}_i^{(1)} := \inf \big\{ t > \tilde{\tau}_{i-1}^{(1)} : \Delta \tilde{X}_t^{(3)} \notin [ 
\tilde{d}^{(1)} (t - \tilde{\tau}_{i-1}^{(1)}), 
\tilde{u}^{(1)} (t - \tilde{\tau}_{i-1}^{(1)}) ] \big\} .$$
These stopping times will be seen as approximations of the observation times of the first asset when we hold 
the volatility matrix $\sigma_t$ and the grid $g_t$ constant. We will always start our approximation at a 
\emph{1-correlated} observation 
time $\tau_{i,n}^{1C}$, which corresponds to an observation time of the first asset. As the sampling times of the second asset 
are not synchronized with the ones from the first asset, we need two more quantities $(x,u) \in \mathcal{S}_{\tilde{g}}$ 
to approximate the observation times
of the second asset. They correspond respectively to the increment of the second asset's time process 
$X_t^{(t,2)}$ since the last observation of the second asset occured and the time elapsed since the last observation time 
of the second asset. We define $\tilde{\tau}_0^{(2)} := 0$,  
$$\tilde{\tau}_1^{(2)} := \inf \big\{ t > 0 : x + \Delta \tilde{X}_t^{(4)} \notin [ \tilde{d}_2 (t +u), 
\tilde{u}_2 (t + u) ] \big\},$$
and for any integer $i \geq 2$
$$\tilde{\tau}_i^{(2)} := \inf \big\{ t > \tilde{\tau}_{i-1}^{(2)} : \Delta \tilde{X}_t^{(4)} \notin [ 
\tilde{d}_2 (t - \tilde{\tau}_{i-1}^{(2)}), 
\tilde{u}_2 (t - \tilde{\tau}_{i-1}^{(2)}) ] \big\} .$$
Similarly, we define the analogs of (\ref{tau-0s})-(\ref{tau-0}), (\ref{tau+0}), 
(\ref{inc-+0}), (\ref{algo1C}), (\ref{tau1C-0s})-(\ref{tau1C-0}), (\ref{tau1C+0}), 
(\ref{inc1C-+0}) and (\ref{compensated}) respectively as $\tilde{\tau}_{i-1}^-$, 
$\tilde{\tau}_{i-1}^+$, $\Delta \tilde{X}_{\tilde{\tau}_i^{-,+}}^{(2)}$, $\tilde{\tau}_{i-1}^{1C,-}$, 
$\tilde{\tau}_{i-1}^{1C,+}$, $\Delta \tilde{X}_{\tilde{\tau}_i^{1C,-,+}}^{(2)}$ and $\tilde{N}_i$ by putting tildes on the quantities in the definitions.

\subsection{Preliminary lemmas}
Without loss of generality, we choose to work under the third scenario defined in Section $2.4$, i.e. the asset price is different from the time process for 
both assets. Because we shall prove stable convergence, and because of the local boundedness of $\sigma$ (because by (A1) $\sigma$ is continuous), 
and that $\inf_{t \in (0,1]} \lambda_t^{\min} > 0$
we can without loss of generarality assume that for all $t \in [0,1]$ there exists some nonrandom constants $\sigma^{-}$ and $\sigma^{+}$ such that 
for any eigen-value $\lambda_t$ of $\sigma_t$ we have 
\begin{eqnarray}
\label{smsp}
0 < \sigma^{-} < \lambda_t < \sigma^{+},
\end{eqnarray}
by using a standard localization argument such that the one used in Section 2.4.5 of Mykland and Zhang (2012). One can further supress $\mu$ as in Section $2.2$ (pp. 1407-1409) of Mykland and Zhang (2009), and act as if 
$X$ is a martingale. 

We define the subspace $\mathcal{M}$ of matrices of dimension $4 \times 4$ such that $\forall M \in \mathcal{M}$, for any eigen-value $\lambda_M$ of 
$M$, we have 
\begin{eqnarray}
 \label{M} \sigma^{-} < \lambda_M < \sigma^{+}
\end{eqnarray} 
and $\frac{\left(M M^T \right)^{3,4}}{\left(M M^T \right)^{4,4}}  \in [\rho_-^{3,4}, \rho_+^{3,4} ]$. By (\ref{A2rho34}) of $(A2)$ and (\ref{smsp}),
we will assume in the following that $\forall t \in [0, 1]$, $\sigma_t \in \mathcal{M}$. 

We define $\sigma^p$ the process (of dimension $4 \times 4$) on $\reels^+$ such that 
 $$
   \left \{
   \begin{array}{r c  c l}
      \sigma_t^p & = & \sigma_t & \forall t \in [0,1],\\
      \sigma_t^p & = & \sigma_1 & \forall t \in [1,\infty).
   \end{array}
   \right .
$$
Define now $X^p$ the process such that for all $t \geq 0$
$$
   \left \{
   \begin{array}{r c l}
      d X_t^p & = &  \sigma_t^p dW_t,\\
      X_0^p & = & X_0.
   \end{array}
   \right .
$$
Because $X^p$ and $X$ have the same initial value and follow the same stochastic differential equation on $[0,1]$, they are equal for all $t \in [0,1]$. For simplicity, we keep from now on the notation $X$ for $X^p$.

\bigskip
 In the following, $C$ will be defining a constant which does not depend on $i$ or $n$, but that can vary from a line to another. Also, 
 we are going to use the notation $\tau_{i,n}^{\theta}$ as a subtitute of $\tau_{i,\alpha_n}^{\theta}$, where $\theta$ can take various names, such that 
 $(1),(2)$ and so on. 
Let $h : \naturels \rightarrow \naturels$ a (not strictly) increasing non-random sequence such that
 \begin{eqnarray}
 \label{infty}
 h_n \rightarrow + \infty, \\
 \label{hnalphan}
 h_n \alpha_n \rightarrow 0.
 \end{eqnarray}
 To keep notation as simple as possible, we define $\tau_{i,n}^{h} := \tau_{i h_n,n}^{1C}$, $\tau_{i,n}^{h,-} :=
\tau_{i h_n,n}^{1C,-}$, $\tau_{i,n}^{h,+} := \tau_{i h_n,n}^{1C,+}$. We also let $A_n := \{ i \geq 1 \text{ s.t. } \tau_{i-1,n}^h \leq t \}$, where 
$t \in [0,1]$. Also, we recall the notation $(X_t^{(3)}, X_t^{(4)}) := (X_t^{(t,1)}, X_t^{(t,2)})$
Finally, for $\theta \in \{ (1), (2), 1C, h \}$, we define $s_n^{\theta} = \underset{\tau_{i,n}^{\theta} < T}{\sup} \Delta \tau_{i,n}^{\theta}$. 
We show that these quantities tend to 0 almost surely in the following lemma. 
\begin{snk}
 \label{snk}
 We have $s_n^{\theta} \overset{a.s.}{\rightarrow} 0$.
\end{snk}

\begin{proof}
 We can follow the proof of Lemma $4.5$ in Robert and Rosenbaum (2012) to prove that for $k \in \{1, 2 \}$, $s_n^{(k)} \overset{a.s.}\rightarrow 0$. 
 Then, we can notice that a.s. $s_n^{1C} < s_n^{(1)} + s_n^{(2)}$ to deduce that $s_n^{1C} \overset{a.s.}\rightarrow 0$. To show that 
 $s_n^h \rightarrow 0$, define the process $Z$ such that $Z_0 = 0$ and $\forall i > 0$ we have
 $$
   Z_t := \left \{
   \begin{array}{r l}
      \Delta X_{[\tau_{i-1,n}^{1C},t]}^{(2)} + Z_{\tau_{i-1,n}^{1C}} & \forall t \in [\tau_{i-1,n}^{1C}, \tau_{i-1,n}^{1C,+}],\\
    \Delta X_{[\tau_{i-1,n}^{1C,+},t]}^{(1)} + Z_{\tau_{i-1,n}^{1C,+}} & \forall t \in [\tau_{i-1,n}^{1C,+}, \tau_{i,n}^{1C}].
   \end{array}
   \right .
$$
Substituting $X$ in Lemma 4.5 of Robert and Rosenbaum's proof by our $Z$, we can follow the same reasoning. The only main change will be that in their notation 
$M_n \leq C h_n \alpha_n$, but this tends to 0 by (\ref{hnalphan}).
\end{proof}

Let $f$ be a random process, $s$ a random number, we define
$$S \left(f, s \right) :=  \underset{0 \leq u,v \leq 1, \mid u-v \mid \leq s}{\sup} \Big| 
f_u - f_v \Big|.$$

\begin{s}
 \label{s}
 Let $f$ be a bounded random process such that for all non-random sequence $\left( q_n \right)_{n \geq 0}$, if $q_n \rightarrow 0$, 
 then $S \left(f, q_n \right) \overset{\proba}{\rightarrow} 0$. Let also a random sequence $\left( s_n \right)_{n \geq 0}$ such that $s_n \overset{\proba}{\rightarrow} 0$. Then we have $\forall l \geq 1$ that
  $$ S \left( f, s_n \right) \overset{\mathbf{L}^l}{\rightarrow} 0.$$
\end{s}

\begin{proof}
As $f$ is bounded, convergence in $\proba$ implies convergence in $\mathbf{L}^l$ for any $l \geq 1$. Hence it is sufficient to show 
that $S \left( f, s_n \right) \overset{\proba}{\rightarrow} 0$. Let $\eta > 0$ and $\epsilon > 0$, we want to show that $\exists N > 0$ such that $\forall n \geq N$, we have
$$\proba \left( S \left( f, s_n \right) > \eta \right) < \epsilon.$$
$\exists$ non-random $\chi > 0$ such that $\proba \left( S \left( f, \chi \right) > \eta \right) < \frac{\epsilon}{2}$. Also, 
$\exists N > 0$ such that $\forall n \geq N$, $\proba \left( s_n \geq \chi \right) < \frac{\epsilon}{2}$. Thus 
\begin{eqnarray*}
\proba \left( S \left( f, s_n \right) > \eta \right) & = & \proba \left( S \left( f, s_n \right) > \eta,s_n > \chi \right) 
+ \proba \left( S \left( f, s_n \right) > \eta,s_n \leq \chi \right) \\ 
& \leq & \proba \left( s_n > \chi \right) +  \proba \left( S \left( f, \chi \right) > \eta \right) < \epsilon.
\end{eqnarray*}
\end{proof}
We aim to define the approximations of observation times on blocks 
$$\left( K_{i,n} := [ \tau_{i,n}^h, 
\tau_{i+1,n}^h ] \right)_{i \geq 0}.$$ 
We need some definitions first. Let $( C_t^{(i)} )_{i \geq 0}$ a sequence of independent 4-dimensional Brownian motions (i.e. for each $i$, 
$C_t^{(i)}$ is a 4-dimensional Brownian motion), independent of everything we have 
defined so far. We define $\forall i, n \geq 0$, 
 $$
   S_t^{i,n} := \left \{
   \begin{array}{r l}
      \Delta W_{[\tau_{i,n}^{h},\tau_{i,n}^{h} + .]} & \forall t \in [0, \Delta \tau_{i+1,n}^{h}],\\
    \Delta W_{[\tau_{i,n}^{h},\tau_{i+1,n}^{h}]} + C_{t- \Delta \tau_{i+1,n}^{h}}^{(i)} & \forall t \geq \Delta \tau_{i+1,n}^{h},
   \end{array}
   \right .
$$
and
$$\left( \tilde{\tau}_{i,j,n}^k \right)_{j \geq 0 ; k=1,2} = \tilde{\Tau} \left( S^{i,n}, 
\sigma_{\tau_{i,n}^h}, \alpha_n g_{\tau_{i,n}^h}, \Delta X_{[\tau_{i,n}^{h,-}, \tau_{i,n}^{h}]}^{(4)}, \tau_{i,n}^h - \tau_{i,n}^{h,-} 
\right).$$
To keep symmetry in notations, we define for all integers $i$ and $n$ positive integers, 
$\left( \tau_{i,j,n}^{(1)} \right)_{j \geq 0}$ consisting of the observation times of the process 1 after $\tau_{i,n}^h$, 
substracting the value of $\tau_{i,n}^h$, i.e. $\tau_{i,j,n}^{(1)} = \tau_{i^* + j, n}^{(1)} - \tau_{i^*, n}^{(1)}$ where $i^*$ is the (random) index 
on the original grid of process 1 corresponding to $\tau_{i,n}^h$ ($\tau_{i^*,n}^{(1)} = \tau_{i,n}^h$). For process 2, we define
$\tau_{i,0,n}^{(2)} = 0$ and for integers $j \geq 1$, $\tau_{i,j,n}^{(2)} = \tau_{j^* + j - 1, n}^{(2)} - \tau_{i^*, n}^{(1)}$, where 
$j^*$ is the index on the original grid of process 2 corresponding to the smallest observation time of process 2 bigger (not necessarily strictly) 
than $\tau_{i,n}^h$. 
We also define $\tau_{i,j,n}^{-}$, $\tau_{i,j,n}^{+}$, $\tau_{i,j,n}^{1C}$, $\tau_{i,j,n}^{1C,-}$, $\tau_{i,j,n}^{1C,+}$, 
$\tilde{\tau}_{i,j,n}^{-}$, $\tilde{\tau}_{i,j,n}^{+}$, $\tilde{\tau}_{i,j,n}^{1C}$, $\tilde{\tau}_{i,j,n}^{1C,-}$, $\tilde{\tau}_{i,j,n}^{1C,+}$ 
following the construction we used to define (\ref{tau-0s}), (\ref{tau-0}), (\ref{tau+0}), (\ref{algo1C}), (\ref{tau1C-0s}), (\ref{tau1C-0}) and (\ref{tau1C+0}). We also set
$$\left( \tilde{\pi}_{i,j,n} \right)_{j \geq 0} = \varPi \left( S^{i,n}, 
\sigma_{\tau_{i,n}^h}, \alpha_n g_{\tau_{i,n}^h}, \Delta X_{[\tau_{i,n}^{h,-}, \tau_{i,n}^h]}^{(4)}, \tau_{i,n}^{h} - \tau_{i,n}^{h,-} 
\right).$$

\begin{esptauk}
\label{esptauk}
For $\theta \in \{ (1), (2), 1C \}$, any real $l > 0$, any positive integer $i$ and $n$, any non-negative integer $j$, we have  $0 < C_l^- < C_l^+$ such that
\begin{eqnarray}
\label{esptauk1} C_{l}^- \alpha_n^{2l} < \esp \left[ \left( \Delta \tilde{\tau}_{i,j,n}^{\theta} \right)^{l} \right] \leq C_{l}^+ \alpha_n^{2l},
\end{eqnarray}
where $\Delta \tilde{\tau}_{i,j,n}^{\theta} := \tilde{\tau}_{i,j,n}^{\theta} - \tilde{\tau}_{i,j-1,n}^{\theta}$ and 
\begin{eqnarray}
\label{esptauk2} C_{l}^- \alpha_n^{2l} < \esp \left[ \left( \Delta \tau_{i,n}^{(k)} \right)^{l} \right] \leq C_{l}^+ \alpha_n^{2l}.
\end{eqnarray}
\end{esptauk}

\begin{proof}
For $\theta \in \{(1), (2) \}$, because of (\ref{esptauk}) and (\ref{smsp}), we can deduce (\ref{esptauk1}) using well-known result on exit zone of a Brownian 
motion (see for instance Borodin and Salminen (2002)). (\ref{esptauk2}) can be 
deduced using Dubins-Schwarz theorem for continuous local martingale (see, e.g. Th. $V.1.6$ in Revuz and Yor (1999)). 
If $\theta = 1C$ writing $ \Delta \tilde{\tau}_{i,j,n}^{\theta} = \left( \tilde{\tau}_{i,j-1,n}^{\theta,+} - \tilde{\tau}_{i,j-1,n}^{\theta} \right) 
+ \left( \tilde{\tau}_{i,j,n}^{\theta,+} - \tilde{\tau}_{i,j-1,n}^{\theta,+} \right)$ and working those two terms, we can obtain (\ref{esptauk1}) and 
(\ref{esptauk2}).
\end{proof}

Now, we define for $\theta \in \{(1), (2), 1C, h \}$ the number of observation times before $t$ as
$$N_{t,n}^{\theta} = \sup \{i : \tau_{i,n}^{\theta} < t \}.$$
We have the following lemma
\begin{numb}
 \label{numb}
For $\theta \in \{(1), (2), 1C \}$, we have that the sequence $\left( \alpha_n^{2} N_{t,n}^{\theta} \right)_{n \geq 1}$ is tight.
\end{numb}

\begin{proof}
 Here for $\theta \in \{(1), (2)\}$ we can follow the proof of Lemma 4.6 in Robert and Rosenbaum (2012) together with Lemma \ref{snk}. Also, by definition we have 
 $N_{t,n}^{1C} \leq N_{t,n}^{(1)}$ so we also deduce the tightness of $\left( \alpha_n^{2} N_{t,n}^{1C} \right)_{n \geq 1}$.
\end{proof}

\begin{sumtight}
\label{sumtight}
 Let $\left( U_{i,n} \right)_{i, n \geq 1}$ an array of positive random variables and $\theta \in \{(1), (2), 1C \}$. If 
 \begin{eqnarray}
 \label{sumtightcond} \forall u > 0, & \sum_{i=1}^{\llcorner 
 u \alpha_n^{-2} \lrcorner} U_{i,n} \overset{\proba}{\rightarrow} 0
 \end{eqnarray}
 then $\sum_{i=1}^{N_{t,n}^{\theta}} U_{i,n} \overset{\proba}{\rightarrow} 0$. 
 Also, if $\forall u > 0$, $\sum_{i=1}^{\llcorner 
 u \alpha_n^{-2} h \left( n \right)^{-1} \lrcorner} 
 U_{i,n} \overset{\proba}{\rightarrow} 0$, then $\sum_{i=1}^{N_{t,n}^{h}} U_{i,n} \overset{\proba}{\rightarrow} 0$.
\end{sumtight}

\begin{proof}
Let $\epsilon > 0$ and $ u > 0$.
\begin{eqnarray*}
 \proba \left( \sum_{i=1}^{N_{t,n}^{\theta}} U_{i,n} > \epsilon \right) & = & 
 \proba \Bigg( \sum_{i=1}^{\llcorner u \alpha_n^{-2} \lrcorner} U_{i,n} + 
 \sum_{i= \llcorner u \alpha_n^{-2} \lrcorner + 1}^{N_{t,n}^{\theta}} U_{i,n} \mathbf{1}_{ \{ \llcorner u \alpha_n^{-2} \lrcorner < N_{t,n}^{\theta} \} } \\
 & &
 -\sum_{i= N_{t,n}^{\theta} + 1}^{\llcorner u \alpha_n^{-2} \lrcorner} U_{i,n} \mathbf{1}_{ \{ \llcorner u \alpha_n^{-2} \lrcorner > N_{t,n}^{\theta} \} } > 
 \epsilon \Bigg) \\
& \leq & 
 \proba \left( \sum_{i=1}^{\llcorner u \alpha_n^{-2} \lrcorner} U_{i,n} + 
 \sum_{i= \llcorner u \alpha_n^{-2} \lrcorner + 1}^{N_{t,n}^{\theta}} U_{i,n} \mathbf{1}_{ \{ \llcorner u \alpha_n^{-2} \lrcorner < N_{t,n}^{\theta} \} } 
  > \epsilon \right) \\
& \leq & \proba \left( \sum_{i=1}^{\llcorner u \alpha_n^{-2} \lrcorner} U_{i,n} > \frac{\epsilon}{2} \right)
+ \proba \left( \sum_{i= \llcorner u \alpha_n^{-2} \lrcorner + 1}^{N_{t,n}^{\theta}} U_{i,n} \mathbf{1}_{ \{ \llcorner u \alpha_n^{-2} \lrcorner < 
N_{t,n}^{\theta} \} } 
  > \frac{\epsilon}{2} \right) \\
& \leq & \proba \left( \sum_{i=1}^{\llcorner u \alpha_n^{-2} \lrcorner} U_{i,n} > \frac{\epsilon}{2} \right)
+ \proba \left( \llcorner u \alpha_n^{-2} \lrcorner < N_{t,n}^{\theta} \right).
 \end{eqnarray*}
 We take the  $\underset{n \rightarrow \infty}{\limsup}$ and use (\ref{sumtightcond}). We obtain
 $$\underset{n \rightarrow \infty}{\limsup} \proba \left( \sum_{i=1}^{N_{t,n}^{\theta}} U_{i,n} > \epsilon \right) 
 \leq \underset{n \rightarrow \infty}{\limsup} \text{ } \proba \left( \llcorner u \alpha_n^{-2} \lrcorner < N_{t,n}^{\theta} \right).$$
 We now tend $u \rightarrow \infty$ and conclude using Lemma \ref{numb}. 
 The second statement is proved in the same way.
\end{proof}

\begin{scale}
\label{scale}
For any $\alpha > 0$, $\sigma \in \mathcal{M}, g \in \mathcal{G}, (x, u) \in \mathcal{S}_g$, we have that 
\begin{eqnarray*}
\psi^{AV} \left(\sigma, g, x, u \right) & = & \alpha^{-4} \psi^{AV} \left(\sigma, \alpha g, \alpha x, \alpha^2 u \right), \\
\psi^{AC1} \left(\sigma, g, x, u \right) & = & \alpha^{-3} \psi^{AC1} \left(\sigma, \alpha g, \alpha x, \alpha^2 u \right), \\
\psi^{AC2} \left(\sigma, g, x, u \right) & = & \alpha^{-3} \psi^{AC2} \left(\sigma, \alpha g, \alpha x, \alpha^2 u \right), \\
\psi^{\tau} \left(\sigma, g, x, u \right) & = & \alpha^{-2} \psi^{\tau} \left(\sigma, \alpha g, \alpha x, \alpha^2 u \right).
\end{eqnarray*}
\end{scale}

\begin{proof}
 For any Brownian motion $\left( W_t \right)_{t \geq 0}$, by the scale property we have that $\left( W_t \right)_{t \geq 0}
 \overset{\mathcal{L}}{=} \left( \alpha^{-1} W_{\alpha^2 t} \right)_{t \geq 0}$. Thus, if we define $\tau = \inf \{t > 0 \text{ s.t. } W_t 
  \notin [d (t), u (t)] \}$ and $\tau_{\alpha} = \inf \{t > 0 \text{ s.t. } W_{t} \notin [\alpha d (t), \alpha u(t)] \}$, we have that 
\begin{eqnarray*} 
\tau \overset{\mathcal{L}}{=} \inf \{t > 0 \text{ s.t. } W_{\alpha^2 t} 
  \notin [\alpha d (t), \alpha u (t)] \} \overset{\mathcal{L}}{=} \alpha^{-2} \tau_{\alpha}.
\end{eqnarray*}
We deduce that
\begin{eqnarray}
\label{scale1} 
\left( \tau, W_{\tau} \right) \overset{\mathcal{L}}{=} \left( \alpha^{-2} \tau_{\alpha}, W_{\alpha^{-2} \tau_{\alpha}} \right) \overset{\mathcal{L}}{=} \left( \alpha^{-2} \tau_{\alpha}, \alpha W_{\tau_{\alpha}} \right).
\end{eqnarray}
We can prove the lemma based on the way we proved (\ref{scale1}), at the cost of 2-dimension definitions that would be more involved and 
straightforward applications of Strong Markov property of Brownian motions that we won't write, so that we don't lose ourselves in the technicality of this 
proof. 
\end{proof}
We introduce the number of points in the $i$th block in the $k$th process as the following
$$N_{i,n}^{(k)} = \max \{j \geq 0 \text{ s.t. } \tau_{i, n}^h + \tau_{i,j,n}^{(k)} \leq \tau_{i+1,n}^h \}. $$
We also introduce the total number of points in the $i$th block $N_{i,n} = N_{i,n}^{(1)} + N_{i,n}^{(2)}$. 
We show now that we can control uniformly 
the error of the approximations of the observation times.
\begin{tauapp}
\label{tauapp}
 Let $l \geq 1$, we have that 
 \begin{eqnarray}
 \label{tauapp1} \sup_{i \geq 0 \text{ , } 2 \leq j \leq h_n} \esp \left[ \Big| \Delta \tau_{i,j,n}^{1C} - 
 \Delta \tilde{\tau}_{i,j,n}^{1C} \Big|^l \right] = o_p \left( \alpha_n^{2l} \right)
 \end{eqnarray}
 and
 \begin{eqnarray}
 \label{tauapp2} \sup_{i \geq 0 \text{ , } 2 \leq j \leq h_n} \esp \left[ \Big| \Delta \tau_{i,j,n}^{1C,-,+} - 
 \Delta \tilde{\tau}_{i,j,n}^{1C,-,+} \Big|^l \right] = o_p \left( \alpha_n^{2l} \right).
 \end{eqnarray}

\end{tauapp}

\begin{proof}
We introduce the notation $o_p^U$ where $U$ stands for ``uniformly in $i \geq 0$'', 
meaning that the $\sup$ of the rests is of the given order

\bigskip
First step : We define $\tilde{s}_n^h = \underset{i \in A_n}{\sup} \tilde{\tau}_{i,h_n,n}^{1C}$. We show in this step that 
\begin{eqnarray}
 \label{tauappstep1} \tilde{s}_n^h \overset{\proba}{\rightarrow} 0.
\end{eqnarray}
We define the accumulated time of approximated durations, i.e. 
$$\tilde{\tau}_{i,n}^h = \sum_{l=0}^{l=i} \tilde{\tau}_{l,h_n,n}^{1C}.$$
Using Lemma \ref{esptauk} together with Lemma \ref{numb}, $\exists M > 0$ such that 
$$\proba \left( \tilde{\tau}_{N_{n}^h, n}^h \leq M \right) \rightarrow 1.$$
We define $Z_0^n = 0$ and $\forall t \in [\tilde{\tau}_{i-1,n}^h, \tilde{\tau}_{i,n}^h]$, 
$$Z_t^n = Z_{\tilde{\tau}_{i-1,n}^h}^n + S_{t - \tilde{\tau}_{i-1,n}^h}^{i-1,n}.$$ 
A slight modification of the proof of Lemma \ref{snk} will conclude.

\bigskip
Second step : We show that we can do a localization in the number of observations in the i-th block, i.e. there exists a non-random $M_n$ such that
\begin{eqnarray}
\label{step2} \proba \left( \max \left( N_{i,n}^{(1)}, N_{i,n}^{(2)} \right) > M_n \right) 
\end{eqnarray}
converges uniformly (in $i$) towards $0$ and $M_n$ increasing at most linearly 
with $h_n$, i.e. we have $M_n \leq \beta h_n$ where $\beta > 0$. 

To prove (\ref{step2}), we need some definitions. Define for $i \geq 0$ the order of observation times $O_{i,k,n}$ and the order of the approximated observation times $\tilde{O}_{i,k,n}$ in the 
following way. Let $\Tau_{i,n}^O := \left( \tau_{i,j,n}^O \right)_{j \geq 0}$ the sorted set of all observation times (corresponding to process 1 and 2) strictly greater than $\tau_{i,n}^h$. Then for 
$j \geq 1$, we will set $O_{i,j,n} = 1$ if the j-th observation time in $\Tau_{i,n}^O$ corresponds to an observation of the first process and 
$O_{i,j,n} = 2$ if it corresponds to an observation of the 
second process. Similarly, we set $\tilde{\Tau}_{i,n}^O$ the sorted set of all approximated times $\left( \tilde{\tau}_{i,j,n}^{(k)} \right)_{j \geq 0,k=1,2}$ . 
$\tilde{O}_{i,j,n}$ are defined in the same way. There exists a $p > 0$ such that for all integers $i,j,n$ : 
\begin{eqnarray}
 \label{p} \proba \left( O_{i,j+1,n} = 1 \Big| O_{i,j,n} = 2 \right) \geq p \text{ and } \proba \left( O_{i,j+1,n} = 2 \Big| O_{i,j,n} = 1 \right) \geq p.
\end{eqnarray}
Indeed, let $l$ the (random) index such that $\tau_{i,l,n}^{(1)} = \tau_{i,j,n}^O$. Conditionally on $\Big\{ O_{i,j,n} = 1 \Big\}$,
we know that $O_{i,j+1,n} = 2$ if $\Delta X_{[\tau_{i,n}^h + \tau_{i,j,n}^l, .]}^{(4)}$ crosses $g^+$ or $- g^+$ 
before $\Delta X_{[\tau_{i,n}^h + 
\tau_{i,j,n}^O,.]}^{(3)}$ crosses $g^-$ or $-g^-$. Using (\ref{A2rho34}) of (A2) and (\ref{smsp}), we can easily bound away from 0 this probability, thus we 
deduce (\ref{p}). Now, using (\ref{algo1C}) together with (\ref{p}) and strong Markov property of Brownian motions, we deduce (\ref{step2}).

\bigskip
Third step : let $g=(d,u)$ such that $(g,g) \in \mathcal{G}$, $\sigma \in [ \sigma^-, \sigma^+ ]$ and $\epsilon \leq \frac{g^-}{2}$. We define $\tau \left(g, \sigma, \epsilon \right) 
= \inf \{ t > 0 : \sigma W_t = u (t) + \epsilon \text{ or }  \sigma W_t = d(t) - \epsilon \}$, where $W_t$ is a standard Brownian motion. We show that 
\begin{eqnarray}
\label{tauappstep2} \esp \left[ \Big| \tau \left(g, \sigma, \epsilon \right) - \tau \left(g, \sigma, 0 \right) \Big|^l \right] \leq \gamma^{(l)} \left( \epsilon \right)
\end{eqnarray}
where $\gamma^{(l)} \left( \epsilon \right) \overset{\epsilon \rightarrow 0}{\rightarrow} 0$.

In order to show (\ref{tauappstep2}), let $$\tau^1 \left(g, \sigma, \epsilon \right) = \inf \{ t > 0 : \sigma W_{t + \tau \left(g, \sigma, 0 \right)} = \min \left( u ( \tau(g, \sigma, 0)) 
+Kt + \epsilon, g^+ \right)$$ 
$$\text{  or } \sigma W_{t + \tau \left(g, \sigma, 0 \right)} = \max \left( d ( \tau(g, \sigma, 0)) 
- Kt - \epsilon , g^- \right) \}.$$ By (\ref{A3g}) and (\ref{A3der}) of (A3), we have $\tau \left(g, \sigma, \epsilon \right) - \tau \left(g, \sigma, 0 \right) 
\leq \tau^1 \left(g, \sigma, \epsilon \right)$. Conditionally on $\Big\{ \tau \left(g, \sigma, \epsilon \right) \Big\}$, using strong Markov property of 
Brownian motions, we can show that $\esp_{\tau \left(g, \sigma, \epsilon \right)} \left[ \Big| \tau^1 \left(g, \sigma, \epsilon \right) \Big|^l \right] \overset{\epsilon \rightarrow 0}{\rightarrow} 0$ 
using Theorem $2$ in Potzelberger and Wang (2001) for instance.

\bigskip
Fourth step : let $k \in \{1 , 2 \}$. We show here that
\begin{eqnarray}
\label{tauappstep3} \sum_{j \leq M_n} \esp \left[ \Big| \tau_{i,j,n}^{(k)} - \tilde{\tau}_{i,j,n}^{(k)} \Big|^l \right] = o_p^U \left( 
\alpha_n^{2l} \right).
\end{eqnarray}
The idea is to show that by recurrence in $j$, $\esp \left[ \Big| 
\tau_{i,j,n}^{(k)} - \tilde{\tau}_{i,j,n}^{(k)} \Big|^l \right]$ can be arbitrarily small when $n$ grows. It is then a straightforward analysis exercise to use 
the localization in second step and choose a different sequence $h$ if necessary, that will still be non-random increasing and following 
(\ref{infty}) and (\ref{hnalphan}), so that the sum in (\ref{tauappstep3}) will be also arbitrarily small. Let's start with $j=1$ and $k=1$.
$$\esp \left[ \Big| \tau_{i,1,n}^{(k)} - \tilde{\tau}_{i,1,n}^{(k)} \Big|^l \right] = \esp \left[ \Big| \tau_{i,1,n}^{(k)} - 
\tilde{\tau}_{i,1,n}^{(k)} \Big|^l \mathbf{1}_{E_{i,n}} \right] + \esp \left[ \Big| \tau_{i,1,n}^{(k)} - 
\tilde{\tau}_{i,1,n}^{(k)} \Big|^l \mathbf{1}_{E_{i,n}^C} \right],$$
where $E_{i,n} = E_{i,n}^{(1)} \cap E_{i,n}^{(2)}$ with $$E_{i,n}^{(1)}  =  \Bigg\{ \underset{s \in [\tau_{i,n}^h, \tau_{i,n}^h + \tau_{i,1,n}^{(1)} \vee \tilde{\tau}_{i,1,n}^{(1)}]}{\sup} \big|  \Delta X_{[\tau_{i,n}^h,s]}^{(1)} - \Delta \tilde{X}_{[\tau_{i,n}^h,s]}^{(1)} \big| < \eta_{1,n} \Bigg\},$$ 
 $$E_{i,n}^{(2)}  =  \Bigg\{ \underset{s \in [\tau_{i,n}^h, \tau_{i,n}^h + \tau_{i,1,n}^{(1)} \vee \tilde{\tau}_{i,1,n}^{(1)}]}{\sup} \big\| 
 g_s^{(1)} - g_{\tau_{i,n}^h}^{(1)} \big\|_{\infty} < \eta_{1,n} \Bigg\},$$ $\eta_{1,n} = q_n \alpha_n$, $q_n = \max \left( \alpha_n^{d-1/2}, z_n^{1/2} \right)$ and 
$z_n = \underset{1 \leq u,v \leq 4}{\sup} \left( \esp \left[ \left( S \left( \sigma^{u,v}, s_n^h \vee \tilde{s}_n^h \right) \right)^2 \right] \right)^{1/2}.$
By (\ref{scale1}) and (\ref{tauappstep2}), 
$$\esp \left[ \Big| \tau_{i,1,n}^{(k)} - 
\tilde{\tau}_{i,1,n}^{(k)} \Big|^l \mathbf{1}_{E_{i,n}} \right] \leq C \alpha_n^{2l} \left( \gamma^{(l)} \left(2 q_n \right) + \gamma^{(l)} \left(- 2 q_n \right) 
\right).$$
Using Cauchy-Schwarz inequality and Lemma \ref{esptauk},
$$\esp \left[ \Big| \tau_{i,1,n}^{(k)} - 
\tilde{\tau}_{i,1,n}^{(k)} \Big|^l \mathbf{1}_{E_{i,n}^C} \right] \leq C \alpha_n^{2l} \proba \left( E_{i,n}^C \right)^{1/2} \leq 
C \alpha_n^{2l} \left( \proba \left( \left( E_{i,n}^{(1)} \right)^C \right) + \proba \left( \left( E_{i,n}^{(2)} \right)^C \right) \right)^{1/2}.$$
On the one hand, 
\begin{eqnarray*}
\proba \left( \left( E_{i,n}^{(1)} \right)^C \right) & \leq & \left( \eta_{1,n} \right)^{-1} \esp \left[ 
\underset{s \in [\tau_{i,n}^h, \tau_{i,n}^h + \tau_{i,1,n}^{(1)} \vee \tilde{\tau}_{i,1,n}^{(1)}]}{\sup} \big| 
 \Delta X_{[\tau_{i,n}^h,s]}^{(1)} - \Delta \tilde{X}_{[\tau_{i,n}^h,s]}^{(1)} \big| \right]\\
 & \leq & C \left( \eta_{1,n} \right)^{-1} \underset{1 \leq u, v \leq 4}{\max} \esp \left[ \left( \int_{\tau_{i,n}^h}^{\tau_{i,n}^h + \tau_{i,1,n}^{(1)} \vee 
 \tilde{\tau}_{i,1,n}^{(1)}} \left( \sigma_s^{u,v} - \sigma_{\tau_{i,n}^h}^{u,v} \right)^2 ds \right)^{1/2} \right] \\
 & \leq & C \left( \eta_{1,n} \right)^{-1} \underset{1 \leq u, v \leq 4}{\max} \esp \left[ \left( \left( \tau_{i,1,n}^{(1)} \vee 
 \tilde{\tau}_{i,1,n}^{(1)} \right) S \left( \sigma^{u,v}, s_n^h \vee \tilde{s}_n^h \right)^2 \right)^{1/2} \right] \\
 & \leq & C \left( \eta_{1,n} \right)^{-1} \left( \esp \left[ \tau_{i,1,n}^{(1)} \vee \tilde{\tau}_{i,1,n}^{(1)} \right] \right)^{1/2} z_n \\
 & \leq & C z_n^{1/2}.
\end{eqnarray*}
where we used Markov inequality in the first inequality, conditional Burkholder-Davis-Gundy inequality in the second inequality, Cauchy-Schwarz inequality 
in the fourth inequality, Lemma \ref{esptauk} in the last inequality. On the other hand,
\begin{eqnarray*}
\proba \left( \left( E_{i,n}^{(2)} \right)^C \right) & \leq & \left( \eta_{1,n} \right)^{-1} \esp \left[ 
\underset{s \in [\tau_{i,n}^h, \tau_{i,n}^h + \tau_{i,1,n}^{(1)} \vee \tilde{\tau}_{i,1,n}^{(1)}]}{\sup} \big\| 
 g_s^{(1)} - g_{\tau_{i,n}^h}^{(1)} \big\|_{\infty} \right] \\
 & \leq & C \left( \eta_{1,n} \right)^{-1} \esp \left[ \left( \tau_{i,1,n}^{(1)} \vee \tilde{\tau}_{i,1,n}^{(1)} \right)^d \right] \\
 & \leq & C \alpha_n^{d-1/2}.
\end{eqnarray*}
where we used Markov inequality in the first inequality, (\ref{A3sup}) of (A3) in the second inequality, Lemma \ref{esptauk} in the last inequality. 
In summary, we have
$$\esp \left[ \Big| 
\tau_{i,j,n}^{(k)} - \tilde{\tau}_{i,j,n}^{(k)} \Big|^l \right] \leq C \alpha_n^{2l} \left( \gamma^{(l)} \left(2 q_n \right) + \gamma^{(l)} \left(- 2 q_n \right) 
+ z_n^{1/2} + \alpha^{d-1/2} \right).$$
which we can make arbitrarily small, because $z_n \rightarrow 0$ by first step together with Lemma \ref{s} and the continuity of $\sigma$ (A1). 
The case with $k=2$ is very similar. Finally, for $j > 1$, the same kind of computation 
techniques, using in addition (\ref{A3der}) of (A3), will work.

\bigskip
Fifth step : Prove that uniformly (in $i$) 
\begin{eqnarray}
\label{tauappstep4} \proba \left( \forall j \leq M_n, O_{i,j,n} = \tilde{O}_{i,j,n} \right) \rightarrow 1.
\end{eqnarray}
To show (\ref{tauappstep4}), let $j \leq M_n$. We define the (random) index $v$ such that $\tau_{i,v,n}^{O} = \tau_{i,j,n}^{(k)}$. 
Modifying suitably $h$ if needed, there exists (using fourth step) a sequence 
$\left( \epsilon_n \right)$ such that
\begin{eqnarray}
 \label{5a} \proba \left( \Big| \tau_{i,j,n}^{(k)} - \tilde{\tau}_{i,j,n}^{(k)} \Big| \leq \alpha_n^2 \epsilon_n \right) & \rightarrow & 1,\\
\label{5b} \proba \left( \Big| \tau_{i,v+1,n}^{O} - \tau_{i,v,n}^{O} \Big| \leq \alpha_n^2 \epsilon_n \right) & \rightarrow & 0.
 \end{eqnarray}
Using (\ref{5a}) and (\ref{5b}), we can verify (\ref{tauappstep4}) by recurrence.

\bigskip
Sixth step : We prove here (\ref{tauapp1}) and (\ref{tauapp2}). Using Lemma \ref{esptauk} and (\ref{tauappstep4}) we obtain
$$\esp \left[ \Big| \Delta \tau_{i,j,n}^{1C} - 
 \Delta \tilde{\tau}_{i,j,n}^{1C} \Big|^l \right] = \esp \left[ \Big| \Delta \tau_{i,j,n}^{1C} - 
 \Delta \tilde{\tau}_{i,j,n}^{1C} \Big|^l \mathbf{1}_{ \{ \forall j \leq M_n, O_{i,j,n} = \tilde{O}_{i,j,n} \} } \right] + o_p^U \left( \alpha_n^{2l} \right).$$
The first term on the right part of the inequality can be bounded by
$$C \Bigg(\esp \left[ \Big| \tau_{i,j,n}^{1C} - \tilde{\tau}_{i,j,n}^{1C} \Big|^l \mathbf{1}_{ \{ \forall j \leq M_n, O_{i,j,n} = \tilde{O}_{i,j,n} \} } \right]$$ 
$$+ \esp \left[ \Big| \tau_{i,j-1,n}^{1C} - \tilde{\tau}_{i,j-1,n}^{1C} \Big|^l 
\mathbf{1}_{ \{ \forall j \leq M_n, O_{i,j,n} = \tilde{O}_{i,j,n} \} } \right] \Bigg).$$
Both terms can be treated with the same trick. Using the second step and Lemma \ref{esptauk}, the first term is equal to
$$\sum_{v \leq M_n} \esp \left[ \Big| \tau_{i,j,n}^{1C} - \tilde{\tau}_{i,j,n}^{1C} \Big|^l 
\mathbf{1}_{ \{ \forall j \leq M_n, O_{i,j,n} = \tilde{O}_{i,j,n} \} } \mathbf{1}_{ \{ \tau_{i,j,n}^{1C} = \tau_{i,v,n}^{(1)} \} } \right] 
+ o_p^U \left(\alpha_n^{2l} \right).$$
The sum is obviously bounded by
$$\sum_{v \leq M_n} \esp \left[ \Big| \tau_{i,j,n}^{1C} - \tilde{\tau}_{i,j,n}^{1C} \Big|^l \right]. $$
and using (\ref{tauappstep3}), we prove (\ref{tauapp1}). We can deduce (\ref{tauapp2}) with the same 
kind of computations.
\end{proof}

Let $M^n$ the interpolated normalized error, i.e.
$$M_t^n = \alpha_n^{-1} \left( \sum_{i \geq 1} \Delta X_{[\tau_{i-1,n}^{1C} \wedge t,\tau_{i,n}^{1C} \wedge t]}^{(1)} 
\Delta X_{[\tau_{i-1,n}^{1C,-} \wedge t,\tau_{i,n}^{1C,+} \wedge t]}^{(2)} - \int_0^t \sigma_s^{(1)} \sigma_s^{(2)} \rho_s^{1,2} ds \right).$$
$M_t^{n}$ corresponds exactly to the normalized error of the Hayashi-Yoshida estimator if we observe the price of both assets at time $t$. 
We recall the definition of 
$$N_{i,n} = \Delta X_{\tau_{i,n}^{1C}}^{(1)} 
\Delta X_{\tau_{i,n}^{1C,-,+}}^{(2)} - \int_{\tau_{i-1,n}^{1C}}^{\tau_{i,n}^{1C}} \sigma_s^{(1)} \sigma_s^{(2)} \rho_s^{1,2} ds.$$ 

\begin{termsvanishing}
\label{termsvanishing}
 We have $$\sum_{i \in A_n}
\esp_{{\tau}_{i-1,n}^h} \left[ \left( \Delta M_{\tau_{i,n}^h}^n \right)^2 \right]$$
$$= \alpha_n^{-2} \sum_{i \in A_n}
\esp_{\tau_{i-1,n}^h} \left[ \sum_{u=2}^{h_n} \left( N_{(i-1)h_n + u} \right)^2 + 2 N_{(i-1)h_n + u} 
N_{(i-1)h_n + u+1} \right] + o_p (1).$$
\end{termsvanishing}

\begin{proof}
We obtain this equality noting that $\left( N_{i,n} \right)_{n \geq 0}$ are centered and 1-correlated, and that the terms left
converge to 0 in probability.
\end{proof}

We introduce the observation time at the start of a block, where ``s'' stands for ``start''
$$\tau_{i,n}^{s} = \sup \{ \tau_{j, n}^{h} \text{ s.t. } \tau_{j, n}^{h} < \tau_{i, n}^{1C} \}.$$
\begin{holdingconst}
\label{holdingconst}
We have 
 $$\alpha_n^{-2} \sum_{i \in A_n}
\esp_{\tau_{i-1,n}^h} \left[ \sum_{u=2}^{h_n} \left( N_{(i-1)h_n + u} \right)^2 + 2 N_{(i-1)h_n + u} N_{(i-1)h_n + u+1} \right]$$
$$= \alpha_n^{2} \sum_{i \in A_n} \sum_{j=0}^{h_n - 2} \int_{\reels^2} \psi^{AV} \left(\sigma_{\tau_{i-1,n}^h}, g_{\tau_{i-1,n}^h}, 
\alpha_n^{-1} x, \alpha_n^{-2} v\right) d \tilde{\pi}_{i-1,j,n} \left( x,v \right) + o_p (1).$$
\end{holdingconst}

\begin{proof}
First step : approximating with holding volatility constant. Set
$$\tilde{N}_{i,n} = \left( \sigma_{\tau_{i-1,n}^{s}} \Delta W_{\tau_{i,n}^{1C}} \right)^{(1)} 
\left( \sigma_{\tau_{i-1,n}^{s}} \Delta W_{\tau_{i,n}^{1C,-,+}} \right)^{(2)} - \int_{\tau_{i-1,n}^{1C}}^{\tau_{i,n}^{1C}} \zeta_{\tau_{i-1,n}^{s}}^{1,2} ds$$ 
where $A^{(i)}$ is the i-th component of the vector A. We want to show that :
$$\alpha_n^{-2} \sum_{i \in A_n}
\esp_{\tau_{i-1,n}^h} \left[ \sum_{u=2}^{h_n} \left( N_{(i-1)h_n + u} \right)^2 + 2 N_{(i-1)h_n + u} 
N_{(i-1)h_n + u+1} \right]$$
$$= \alpha_n^{-2} \sum_{i \in A_n}
\esp_{\tau_{i-1,n}^h} \left[ \sum_{u=2}^{h_n} \left( \tilde{N}_{(i-1)h_n + u} \right)^2 + 2 \tilde{N}_{(i-1)h_n + u} 
\tilde{N}_{(i-1)h_n + u+1} \right] + o_p \left( 1 \right).$$
Noting $F_{i,n} = \left( N_{i,n} \right)^2 + 2 N_{i,n} N_{i+1,n}$ and $\tilde{F}_{i,n} = \left( \tilde{N}_{i,n} \right)^2 + 
2 \tilde{N}_{i,n} \tilde{N}_{i+1,n}$, it is sufficient to show that 
$$ \alpha_n^{-2} \sum_{i \geq 1} \esp_{\tau_{i-1,n}^{s}} \left[ \bigg| F_{i,n} - \tilde{F}_{i,n} \bigg| \mathbf{1}_{ \{ \tau_{i-1,n}^s < t \} } \right] \overset{\proba}{\rightarrow} 0,$$
that we can rewrite as $ \alpha_n^{-2} \sum_{i \geq 1}^{N_{t,n}^{(1)}} \esp_{\tau_{i-1,n}^{s}} \left[ \Big| F_{i,n} - \tilde{F}_{i,n} \Big| 
\mathbf{1}_{ \{ \tau_{i-1,n}^s < t \} } \right] \overset{\proba}{\rightarrow} 0$. 
Using Lemma \ref{sumtight}, it is sufficient to show that $\forall u > 0$ :
$$ \alpha_n^{-2} \sum_{i=1}^{u \alpha_n^{-2}} 
\esp_{\tau_{i-1,n}^{s}} \left[ \Big| F_{i,n} - \tilde{F}_{i,n} \Big| \mathbf{1}_{ \{ \tau_{i-1,n}^s < t \} } \right] \overset{\proba}{\rightarrow} 0.$$
Thus, it is sufficient to show the convergence $\mathbf{L}^1$ of this quantity, i.e. that 
$$ \alpha_n^{-2} \sum_{i=1}^{u \alpha_n^{-2}} 
\esp \left[ \Big| F_{i,n} - \tilde{F}_{i,n} \Big| \mathbf{1}_{ \{ \tau_{i-1,n}^s < t \} } \right] \rightarrow 0.$$ 
We have that 
$$\bigg| F_{i,n} - \tilde{F}_{i,n} \bigg| \leq B_{i,n}^{(1)} + 2 B_{i,n}^{(2)},$$
where $B_{i,n}^{(1)} = \bigg| N_{i,n}^2 - \tilde{N}_{i,n}^2 \bigg|$ and $B_{i,n}^{(2)} = \bigg| N_{i-1,n} N_{i,n} - \tilde{N}_{i-1,n} 
\tilde{N}_{i,n} \bigg|$. We have that 
$$B_{i,n}^{(1)} \leq C_{i,n}^{(1)} + C_{i,n}^{(2)} + C_{i,n}^{(3)},$$
where
\begin{eqnarray*}
C_{i,n}^{(1)} & = & \bigg| \left( \Delta X_{\tau_{i,n}^{1C}}^{(1)} 
\Delta X_{\tau_{i,n}^{1C,-,+}}^{(2)} \right)^2 - \left( \left( \sigma_{\tau_{i-1,n}^{s}} \Delta W_{\tau_{i,n}^{1C}} \right)^{(1)} 
\left( \sigma_{\tau_{i-1,n}^{s}} \Delta W_{\tau_{i,n}^{1C,-,+}} \right)^{(2)} \right)^2 \bigg|,\\
C_{i,n}^{(2)} & = & \bigg| \left( \int_{\tau_{i-1,n}^{1C}}^{\tau_{i,n}^{1C}} \zeta_s^{1,2} ds \right)^2 
-  \left( \int_{\tau_{i-1,n}^{1C}}^{\tau_{i,n}^{1C}} \zeta_{\tau_{i-1,n}^{s}}^{1,2} ds \right)^2 \bigg|,
\end{eqnarray*}
\begin{eqnarray*}
C_{i,n}^{(3)} & = & 2 \bigg| \Delta X_{\tau_{i,n}^{1C}}^{(1)} 
\Delta X_{\tau_{i,n}^{1C,-,+}}^{(2)} \int_{\tau_{i-1,n}^{1C}}^{\tau_{i,n}^{1C}} \zeta_s^{1,2} ds \\ 
& & - \left( \sigma_{\tau_{i-1,n}^{s}} \Delta W_{\tau_{i,n}^{1C}} \right)^{(1)} 
\left( \sigma_{\tau_{i-1,n}^{s}} \Delta W_{\tau_{i,n}^{1C,-,+}} \right)^{(2)} 
\int_{\tau_{i-1,n}^{1C}}^{\tau_{i,n}^{1C}} \zeta_{\tau_{i-1,n}^{s}}^{1,2} ds \bigg|.
\end{eqnarray*}
Let's show that $\alpha_n^{-2} \sum_{i=1}^{u \alpha_n^{-2}} 
\esp \left[ C_{i,n}^{(1)} \mathbf{1}_{  \{ \tau_{i-1,n}^s < t \}} \right] \rightarrow 0$. 
We can write it as $C_{i,n}^{(1)} \leq D_{i,n}^{(1)} + D_{i,n}^{(2)}$, where 	
\begin{eqnarray*}
D_{i,n}^{(1)} & = & \bigg| \left( \Delta X_{\tau_{i,n}^{1C}}^{(1)} 
\Delta X_{\tau_{i,n}^{1C,-,+}}^{(2)} \right)^2 \left( \left( \sigma_{\tau_{i-1,n}^{s}} \Delta W_{\tau_{i,n}^{1C}} \right)^{(1)} 
\Delta X_{\tau_{i,n}^{1C,-,+}}^{(2)} \right)^2 \bigg|, \\
D_{i,n}^{(2)} & = & \bigg| \left( \left( \sigma_{\tau_{i-1,n}^{s}} \Delta W_{\tau_{i,n}^{1C}} \right)^{(1)} 
\Delta X_{\tau_{i,n}^{1C,-,+}}^{(2)} \right)^2 - \\ & & - 
\left( \left( \sigma_{\tau_{i-1,n}^{s}} \Delta W_{\tau_{i,n}^{1C}} \right)^{(1)} 
\left( \sigma_{\tau_{i-1,n}^{s}} \Delta W_{\tau_{i,n}^{1C,-,+}} \right)^{(2)} \right)^2 \bigg|.
\end{eqnarray*}
We want to show that $\alpha_n^{-2} \sum_{i=1}^{u \alpha_n^{-2}} 
\esp \left[ D_{i,n}^{(1)} \mathbf{1}_{  \{ \tau_{i-1,n}^s < t \}} \right] 
\rightarrow 0 $. We define :
\begin{eqnarray*}
E_{i,n}^{(1)} & = & \Delta X_{\tau_{i,n}^{1C}}^{(1)} 
\Delta X_{\tau_{i,n}^{1C,-,+}}^{(2)}, \\
E_{i,n}^{(2)} & = & \left( \sigma_{\tau_{i-1,n}^{s}} \Delta W_{\tau_{i,n}^{1C}} \right)^{(1)} 
\Delta X_{\tau_{i,n}^{1C,-,+}}^{(2)}.
\end{eqnarray*}

Using Cauchy-Schwarz inequality, we deduce :
\begin{eqnarray*}
 \esp \left[ D_{i,n}^{(1)} \mathbf{1}_{  \{ \tau_{i-1,n}^s < t \}} \right] & = & \esp \left[ \left( E_{i,n}^{(1)} + E_{i,n}^{(2)} \right) 
 \left( E_{i,n}^{(1)} - E_{i,n}^{(2)} \right) \mathbf{1}_{  \{ \tau_{i-1,n}^s < t \}} \right]\\
& \leq & \left( \esp \left[ \left( E_{i,n}^{(1)} + E_{i,n}^{(2)} \right)^2 \right] 
\esp \left[ \left( E_{i,n}^{(1)} - E_{i,n}^{(2)} \right)^2 \mathbf{1}_{  \{ \tau_{i-1,n}^s < t \}} \right] 
 \right)^{1/2}.
 \end{eqnarray*}
 
Using Cauchy-Schwarz inequality together with Burkholder-Davis-Gundy inequality and Lemma \ref{esptauk}, we obtain that :
$$\esp \left[ \left( E_{i,n}^{(1)} + E_{i,n}^{(2)} \right)^2 \right] = O^U \left( \alpha_n^4 \right).$$
where $U$ stands for ``uniformly in $1 \leq i \leq u \alpha_n^{-2}$''. Another application of Cauchy-Schwarz inequality gives us
$$\esp \left[ \left( E_{i,n}^{(1)} - E_{i,n}^{(2)} \right)^2 \mathbf{1}_{  \{ \tau_{i-1,n}^s < t \}} \right]$$ 
$$\leq \left( \esp \left[ \left( 
\Delta X_{\tau_{i,n}^{1C}}^{(1)}  - \left( \sigma_{\tau_{i-1,n}^{s}} \Delta W_{\tau_{i,n}^{1C}} \right)^{(1)} \right)^4 
\mathbf{1}_{  \{ \tau_{i-1,n}^s < t \}} \right] \esp \left[ \left( 
\Delta X_{\tau_{i,n}^{1C,-,+}}^{(2)} \right)^4 \right] \right)^{1/2}.$$
Using once again Cauchy-Schwarz inequality together with Burholder-Davis-Gundy inequality and Lemma \ref{esptauk}, we obtain that :
$$\esp \left[ \left( 
\Delta X_{\tau_{i,n}^{1C,-,+}}^{(2)} \right)^4 \right] = O^U \left( \alpha_n^4 \right).$$
Similarly, we compute using conditional Burkholder-Davis-Gundy in first inequality, Cauchy-Schwarz in third inequality, Lemma \ref{snk}, Lemma \ref{s} 
and Lemma \ref{esptauk} together with the continuity of $\sigma$ (A1) in last equality.

$$ \esp \left[ \left( 
\Delta X_{\tau_{i,n}^{1C}}^{(1)}  - \left( \sigma_{\tau_{i-1,n}^{s}} 
\Delta W_{\tau_{i,n}^{1C}} \right)^{(1)} \right)^4 \mathbf{1}_{  \{ \tau_{i-1,n}^s < t \}} \right]$$  
\begin{eqnarray*}
& = & \esp \left[ \mathbf{1}_{  \{ \tau_{i-1,n}^s < t \}} \esp_{\tau_{i-1,n}^{1C}} \left[ \left( 
\Delta X_{\tau_{i,n}^{1C}}^{(1)}  - \left( \sigma_{\tau_{i-1,n}^{s}} 
\Delta W_{\tau_{i,n}^{1C}} \right)^{(1)} \right)^4 \right] \right]\\
& = & \esp \left[ \mathbf{1}_{  \{ \tau_{i-1,n}^s < t \}} \esp_{\tau_{i-1,n}^{1C}} \left[ \left( \int_{\tau_{i-1,n}^{1C}}^{\tau_{i,n}^{1C}} 
\left( \left( \sigma_{s} - \sigma_{\tau_{i-1,n}^{s}} \right) dW_s \right)^{(1)} \right)^4 \right] \right]\\
& \leq & C \sup_{1 \leq j, l \leq 4} \esp \left[ \mathbf{1}_{  \{ \tau_{i-1,n}^s < t \}} \esp_{\tau_{i-1,n}^{1C}} \left[\left( 
\int_{\tau_{i-1,n}^{1C}}^{\tau_{i,n}^{1C}} 
\left(\sigma_{s}^{j,l} - \sigma_{\tau_{i-1,n}^{s}}^{j,l} \right)^2 ds \right)^2 \right] \right]\\
& = & C \sup_{1 \leq j, l \leq 4} \esp \left[ \mathbf{1}_{  \{ \tau_{i-1,n}^s < t \}} \left( \int_{\tau_{i-1,n}^{1C}}^{\tau_{i,n}^{1C}} 
\left(\sigma_{s}^{j,l} - \sigma_{\tau_{i-1,n}^{s}}^{j,l} \right)^2 ds \right)^2 \right]
\end{eqnarray*}
\begin{eqnarray*}
& \leq & C \sup_{1 \leq j, l \leq 4} \esp \left[ \left( \Delta \tau_{i,n}^{1C} S \left( \sigma^{j,l}, s_n^{h} \right)^2 \right)^2 \right] + o^U \left(\alpha_n^4 \right) \\
& \leq & C \left( \esp \left[ \left( \Delta \tau_{i,n}^{1C} \right)^4 \right] \esp \left[ \sup_{1 \leq j, l \leq 4} \left( S \left( \sigma^{j,l}, s_n^{h} 
\right) \right)^8 \right] \right)^{1/2} + o^U \left(\alpha_n^4 \right)\\
& = & O^U \left(\alpha_n^4 \right).
\end{eqnarray*}
With the same kind of computations, we show that $\alpha_n^{-2} \sum_{i=1}^{u \alpha_n^{-2}} 
\esp \left[ D_{i,n}^{(2)} \mathbf{1}_{  \{ \tau_{i-1,n}^s < t \}} \right] \rightarrow 0$, and we also can show $\alpha_n^{-2} \sum_{i=1}^{u \alpha_n^{-2}} 
\esp \left[ C_{i,n}^{(2)} \mathbf{1}_{  \{ \tau_{i-1,n}^s < t \}} \right] \rightarrow 0$, 
$\alpha_n^{-2} \sum_{i=1}^{u \alpha_n^{-2}} 
\esp \left[ C_{i,n}^{(3)} \mathbf{1}_{  \{ \tau_{i-1,n}^s < t \}} \right] \rightarrow 0$ 
(thus we have also that $\alpha_n^{-2} \sum_{i=1}^{u \alpha_n^{-2}} 
\esp \left[ B_{i,n}^{(1)} \mathbf{1}_{  \{ \tau_{i-1,n}^s < t \}} \right] \rightarrow 0$) 
and 
$$\alpha_n^{-2} \sum_{i=1}^{u \alpha_n^{-2}} 
\esp \left[ B_{i,n}^{(2)} \mathbf{1}_{  \{ \tau_{i-1,n}^s < t \}} \right] \rightarrow 0.$$

\bigskip
Second step : approximating using $\left( \tilde{\tau}_{i,j,n} \right)_{i,j,n \geq 0}$ instead of 
$\left( \tau_{i,n} \right)_{i,n \geq 0}$. We set
$$\tilde{\tilde{N}}_{i,j,n} = \left( \sigma_{\tau_{i,n}^{h}} \Delta W_{\tilde{\tau}_{i,j,n}^{1C}} \right)^{(1)} 
\left( \sigma_{\tau_{i,n}^{h}} \Delta W_{\tilde{\tau}_{i,j,n}^{1C,-,+}} \right)^{(2)} - \int_{\tilde{\tau}_{i,j-1,n}^{1C}}^{\tilde{\tau}_{i,j,n}^{1C}} 
\zeta_{\tau_{i,n}^{h}}^{1,2} ds.$$ 
We want to show that
$$\alpha_n^{-2} \sum_{i \in A_n}
\esp_{\tau_{i-1,n}^h} \left[ \sum_{u=2}^{h_n} \left( \tilde{N}_{(i-1)h_n + u} \right)^2 + 2 \tilde{N}_{(i-1)h_n + u} 
\tilde{N}_{(i-1)h_n + u+1} \right]$$
$$= \alpha_n^{-2} \sum_{i \in A_n}
\esp_{\tau_{i-1,n}^h} \left[ \sum_{u=2}^{h_n} \left( \tilde{\tilde{N}}_{i-1,u,n} \right)^2 + 2 \tilde{\tilde{N}}_{i-1,u,n} 
\tilde{\tilde{N}}_{i,u+1,n} \right] + o_p \left( 1 \right).$$
Using the same kind of computations as in the first step together with Lemma \ref{tauapp}, we conclude.

Third step : express the result as a function of $\psi^{AV}$. Using Lemma \ref{scale} in last equality, we deduce for any integer $u$ such that $2 \leq u \leq h_n$ that
\begin{eqnarray*}
& & \esp_{\tau_{i-1,n}^h} \left[ \left( \tilde{\tilde{N}}_{i-1,u,n} \right)^2 + 2 \tilde{\tilde{N}}_{i-1,u,n} 
\tilde{\tilde{N}}_{i-1,u+1,n} \right]\\
& = & \int_{\reels^2} \psi^{AV} \left(\sigma_{\tau_{i-1,n}^h}, \alpha_n g_{\tau_{i-1,n}^h}, 
x, v \right) d \tilde{\pi}_{i,u-2,n} \left( x,v \right) \\
& = & \alpha_n^4 \int_{\reels^2} \psi^{AV} \left(\sigma_{\tau_{i-1,n}^h}, g_{\tau_{i-1,n}^h}, \alpha_n^{-1} x, \alpha_n^{-2} v \right) 
d \tilde{\pi}_{i,u-2,n} \left( x,v \right).
\end{eqnarray*}

\end{proof}

\begin{approxdistrib}
\label{approxdistrib}
$\forall \sigma \in \mathcal{M}, g \in \mathcal{G}, \exists \pi \left(\sigma, g \right)$ distribution such that :
$$\alpha_n^{2} \sum_{i \in A_n}
\sum_{j=0}^{h_n - 2} \int_{\reels^2} \psi^{AV} \left(\sigma_{\tau_{i-1,n}^h}, g_{\tau_{i-1,n}^h}, 
\alpha_n^{-1} x, \alpha_n^{-2} u \right) d \tilde{\pi}_{i-1,j,n} \left( x,u \right)$$
$$= \alpha_n^{2} \sum_{i \in A_n} 
h_n \phi^{AV} \left(\sigma_{\tau_{i-1,n}^h}, g_{\tau_{i-1,n}^h} \right) + o_p (1).$$
\end{approxdistrib}

\begin{proof}
 We define the transition functions of the Markov chains $\left( \tilde{Z}_i \left( \sigma, g \right) \right)_{i \geq 0}$ defined in (\ref{Zi}). For 
 $\left(x, u \right) \in \mathcal{S}_g$, $B \in \mathcal{B} \left( \mathcal{S}_g \right)$ (borelians of $\mathcal{S}_g$) 
 $$P \left( \sigma, g \right) \left( \left(x, u \right), B \right) = \proba \left( \tilde{Z}_1 \left( \sigma, g \right) \in B \bigg| 
 \tilde{Z}_0 \left( \sigma, g \right) = \left(x, u \right)\right).$$
 First step : We prove that $\forall \sigma \in \mathcal{M}$, $\forall g \in \mathcal{G}$, the state space $\mathcal{S}_g$ is $\nu$-small, i.e. 
 there exists a non-trivial measure $\nu$ 
 on $\mathcal{B} ( \reels^2 )$ such that $\forall (x,u) \in \mathcal{S}_g, \forall B \in \mathcal{B} (\mathcal{S}_g)$, $P \left( \sigma, g \right) \left((x,u), B \right) \geq \nu \left( B \right)$.
Let $B = [x_a, x_b] \times [u_a, u_b]$. We are choosing $\nu$ such that $\nu = 0$ outside $[- \frac{g^-}{4}, \frac{g^-}{4}] \times [3, 4]$. Thus, 
without loss of generality, 
we have that $[x_a, x_b] \times [u_a, u_b] \subset [- \frac{g^-}{4}, \frac{g^-}{4}] \times [3, 4]$. We want to show 
that $\exists c > 0$ such that 
uniformly
$$P \left( \sigma, g \right) \left( \left(x, u \right), B \right) \geq c \left( x_b - x_a \right) \left( u_b - u_a \right).$$

There are two useful ways to rewrite $(\tilde{X}^{(3)}, \tilde{X}^{(4)})$. The first one is :
\begin{eqnarray}
\label{W11}      \tilde{X}_t^{(3)} & := & \sigma^{(3)} \tilde{B}_{t}^{(3)},\\
\label{W12}      \tilde{X}_t^{(4)} & := & \rho^{3,4} \sigma^{(4)} \tilde{B}_{t}^{(3)} + 
      \left( 1 - \left( \rho^{3,4} \right)^2 \right)^{1/2} \sigma^{(4)} \tilde{B}_{t}^{3,\perp}.
\end{eqnarray}
where $\tilde{B}^{(3)}$ and $\tilde{B}^{3,\perp}$ are independent, $\rho^{3,4} \in [\rho_-^{3,4}, \rho_+^{3,4}]$ 
and $\max \left( - \rho_-^{3,4}, \rho_+^{3,4} \right) < 1$ (because $\sigma \in \mathcal{M}$), 
\begin{eqnarray}
\label{delta}
\delta = \left(1 - \max \left( \left(\rho_-^{3,4} \right)^2, \left( \rho_+^{3,4} \right)^2 \right) \right)^{1/2}.
\end{eqnarray} 
The other way to rewrite it is : 
\begin{eqnarray}
 \label{W21}     \tilde{X}_t^{(4)} & := & \sigma^{(4)} \tilde{B}_{t}^{(4)},\\
 \label{W22}     \tilde{X}_t^{(3)} & := & \rho^{3,4} \sigma^{(3)} \tilde{B}_{t}^{(4)} + 
      \left( 1 - \left( \rho^{3,4} \right)^2 \right)^{1/2} \sigma^{(3)} \tilde{B}_{t}^{4,\perp}.
\end{eqnarray}

where $\tilde{B}^{(4)}$ and $\tilde{B}^{4,\perp}$ are independent. For $\left( B_t \right)_{t \geq 0}$ a standard 
Brownian motion, $a < x < b$, we denote the exiting-zone time of the Brownian motion
$$\tau_x^{a,b} = \inf \{ t > 0 \text{ s.t. } x + B_t = a \text{ or } x + B_t = b \}$$
and $p_1(x,a,b,t)$ the density of $\tau_x^{a,b}$. We also define $p_2 (x,a,b,s,y)$ the distribution of $B_s + x$ conditioned on $\{ \tau_x^{a,b} \geq s \}$. 
Finally, let $p_3 (x,a,b,t)$ the distribution of $\tau_x^{a,b}$ conditioned on $\{ B_{\tau_x^{a,b}} = b \}$. All the formulas can be found in Borodin and Salminen 
(2002). Consider the spaces $C_1 = C_3 = \{ \left(x,a,b,t \right) \in \reels^4 \text{ s.t. } a \leq x \leq b, t > 0 \}$, 
$C_2 = \{ \left(x,a,b,t,y \right) \in \reels^5 \text{ s.t. } a \leq x \leq b\text{ , } a < y < b \text{ , } t > 0 \}$. The functions $p_i$ are 
continuous on $C_i$ and positive. Thus, for all compact set $ K_i \subset C_i$, we have 
\begin{eqnarray}
\label{compact} \underset{k \in K_i}{\inf} p_i (k) > 0.
\end{eqnarray}
We can bound below
\begin{eqnarray*}
 P \left( \sigma, g \right) \left( \left(x, u \right), B \right) & \geq & \proba \left( E_0 \bigcap E_1 
 \bigcap E_2 \bigcap E_3 \bigcap E_4  \bigg| \tilde{Z}_0 = (x,u) \right),
\end{eqnarray*}
where
\begin{eqnarray*}
 E_0 & = & \Big\{ \underset{0 \leq s \leq \tilde{\tau}_1^{(2)}}{\sup} \Big| \tilde{X}_s^{(3)} \Big| < \frac{\epsilon \sigma^- \min(\sigma^-,1)}
 {15 \sigma^+} \text{ , } 
 \tilde{\tau}_1^{(2)} \leq K \Big\},\\
 E_1 & = & \Big\{ \underset{\tilde{\tau}_1^{(2)} \leq s \leq K+1}{\sup} \Big| \tilde{X}_s^{(3)} \Big| < \frac{\epsilon \sigma^-}{10 \sigma^+} \text{ , } 
 \underset{\tilde{\tau}_1^{(2)} \leq s \leq K+1}{\sup} \Big| \Delta \tilde{B}_{[\tilde{\tau}_1^{(2)}, s]}^{3, \perp} \Big| < \frac{g^- \sigma^-}
 {4 \left(\sigma^+ \right)^2} \Big\},\\
 E_2 & = & \Big\{ \underset{K+1 \leq s \leq \tilde{\tau}_2^{(2)}}{\sup} \Big| \tilde{X}_s^{(3)} \Big| \leq \frac{\epsilon}{5} \text{ , } 
 \tilde{\tau}_2^{(2)} \in [K+2, K+3] \Big\}, \\
  E_3 & = & \Big\{ \forall s \in [\tilde{\tau}_2^{(2)}, K+4] \text{ } \tilde{X}_{s}^{(3)} \in [d_1 (K), u_1 (K)] \text{ , } 
 \tilde{X}_{K+4}^{(3)} \in [u_1 (K) - 2 \epsilon, u_1 (K) - \epsilon] \Big\} \\
 & & \bigcap \Big\{ \underset{\tilde{\tau}_2^{(2)} \leq s \leq K+4}{\sup} \Big| \Delta \tilde{X}_{[\tilde{\tau}_2^{(2)},s]}^{(4)} \Big| < \frac{g^-}{12} \Big\}, \\
 E_4 & = & \Big\{ \tilde{\tau}_1^{(1)} \in [u_a + \tilde{\tau}_2^{(2)}, u_b + \tilde{\tau}_2^{(2)}] \text{ , } \underset{K+4 \leq s \leq \tilde{\tau}_1^{(1)}}{\inf} \Delta \tilde{X}_{[K+4,s]}^{(3)} > 
 - 2 \epsilon  \Big\} \\
 & & \bigcap \Big\{  \underset{K+4 \leq s \leq \tilde{\tau}_1^{(1)}}{\sup} \Big| 
 \Delta \tilde{X}_{[\tilde{\tau}_2^{(2)},s]}^{(4)} \Big| < g^- \text{ , } \Delta \tilde{X}_{[\tilde{\tau}_2^{(2)},\tilde{\tau}_1^{(1)}]}^{(4)} \in [x_a, x_b] \Big\},
 \end{eqnarray*}
 where $\epsilon = \frac{g^- \sigma^-}{24 \sigma^+}$. Using extensively Bayes formula, we can rewrite 
 $$\proba \left( E_0 \bigcap E_1 
 \bigcap E_2 \bigcap E_3 \bigcap E_4 \bigcap \{ \tilde{Z}_1 \in B \} \bigg| \tilde{Z}_0 = (x,u) \right) = I \times II \times III \times IV \times V,$$ 
 where 
 $I= \proba \left( E_0 \bigg| \{ \tilde{Z}_0 = (x,u) \} \right)$, $II= \proba \left( E_1 \bigg| E_0 \bigcap \{ \tilde{Z}_0 = (x,u) \} \right)$, 
 and also $III= \proba \left( E_2 \bigg| E_1 \bigcap E_0 \bigcap \{ \tilde{Z}_0 = (x,u) \} \right)$, $IV= \proba \left( E_3 \bigg| E_2 \bigcap E_1 \bigcap E_0 
 \bigcap \{ \tilde{Z}_0 = (x,u) \} \right)$ and $V = \proba \left( E_4 \bigg| E_3 \bigcap E_2 \bigcap E_1 \bigcap E_0 
 \bigcap \{ \tilde{Z}_0 = (x,u) \} \right)$. 
 
 We prove that $I$ is uniformly bounded away from $0$. Using (\ref{M}), (\ref{W11}), (\ref{W12}) and (\ref{delta}), we deduce that $E_0^{(1)} \bigcap E_0^{(2)} \subset 
 E_0$ where
 \begin{eqnarray*}
  E_0^{(1)} & = & \Big\{ \underset{0 \leq s \leq K}{\sup} \Big| \tilde{B}_s^{(3)} \Big| < \frac{\epsilon \sigma^- \min(\sigma^-,1)}
 {15 \left( \sigma^+ \right)^2} \Big\}, \\
 E_0^{(2)} & = & \Big\{ \underset{0 \leq s \leq K}{\sup} \Big| \frac{x}{\sigma^{(4)} \left(1 - \left( \rho^{3,4} \right)^2 \right)^{1/2}} + 
 \tilde{B}_s^{3, \perp} \Big| \geq \frac{g^+}{\delta \sigma^{-}} + 
 \frac{\epsilon \sigma^- \min(\sigma^-,1)}
 {15 \left( \sigma^+ \right)^2} \Big\}.
 \end{eqnarray*}
Conditionally on $\{ \tilde{Z}_0 = (x,u) \}$, $E_0^{(1)}$ and $E_0^{(2)}$ are independent. Thus, we deduce
$$I \geq \proba \left( E_0^{(1)} \bigg| \{ \tilde{Z}_0 = (x,u) \} \right)
\proba \left( E_0^{(2)} \bigg| \{ \tilde{Z}_0 = (x,u) \} \right).$$
Using Markov property of Brownian motions, we obtain that the right part of the inequality is equal to
$$\left(1 - \int_0^K p_1 \left( 0, - \frac{\epsilon \sigma^- \min(\sigma^-,1)}
 {15 \left( \sigma^+ \right)^2}, 
\frac{\epsilon \sigma^- \min(\sigma^-,1)}
 {15 \left( \sigma^+ \right)^2}, t \right) dt \right)
\int_0^K p_1 \left( y_0^{(1)}, 
- y_0^{(2)}, y_0^{(2)}, t \right) dt, $$
where $y_0^{(1)} = \frac{x}{\sigma^{(4)} \left(1 - \left( \rho^{3,4} \right)^2 \right)^{1/2}}$, 
$y_0^{(2)} = \frac{g^+}{\delta \sigma^{-}} + 
 \frac{\epsilon \sigma^- \min(\sigma^-,1)}
 {15 \left( \sigma^+ \right)^2}$, which is uniformly (in $x$, $\sigma$ and $g$) bounded away from 0 using (\ref{M}) and (\ref{compact}).

We prove that $II$ is uniformly bounded away from $0$. Conditionally on $E_0 \bigcap \{ \tilde{Z}_0 = (x,u) \}$, the two quantities of $E_1$ are independent. 
Thus, we bound below $II$ (the same way we did for $I$) by
$$\left(1 - \int_{\tilde{\tau}_1^{(2)}}^{K+1} p_1 \left( \tilde{B}_{\tilde{\tau}_1^{(2)}}^{(3)}, - \frac{\epsilon \sigma^-}
  {10 \sigma^{+} \sigma^{(3)}}, \frac{\epsilon \sigma^-}{10 \sigma^{+} \sigma^{(3)}}, t \right) dt \right)$$
$$ \left(1 - \int_{\tilde{\tau}_1^{(2)}}^{K+1} p_1 \left(0, - \frac{g^- \sigma^{-}}{4 \sigma^+ \sigma^{(4)}}, \frac{g^- \sigma^{-}}{4 \sigma^+ \sigma^{(4)}}, t \right) dt \right), $$
which is uniformly bounded away from 0 using (\ref{M}) together with (\ref{compact}).
 
 We prove that $III$ is uniformly bounded away from $0$. Using (\ref{M}), (\ref{W11}), (\ref{W12}) and (\ref{delta}), we deduce that $E_2^{(1)} \bigcap E_2^{(2)} \subset 
 E_2$ where
 \begin{eqnarray*}
  E_2^{(1)} & = & \Big\{ \underset{K+1 \leq s \leq K+3}{\sup} \Big| \tilde{B}_s^{(3)} \Big| \leq \frac{\epsilon}{5 \sigma^{+}} \Big\}, \\
 E_2^{(2)} & = & \Big\{ \underset{K+1 \leq s \leq K+2}{\sup} \Big| \Delta \tilde{B}_{[\tilde{\tau}_1^{(2)}, s]}^{3, \perp} \Big| < \frac{g^-}{2 \sigma^{+}} 
 \text{ , } \underset{K+2 \leq s \leq K+3}{\sup} \Big| \Delta \tilde{B}_{[\tilde{\tau}_1^{(2)}, s]}^{3, \perp} \Big| \geq \frac{g^+}{\delta \sigma^{-}} + 
 \frac{\epsilon}{5 \sigma^+ \delta} \Big\}.
 \end{eqnarray*}
Conditionally on $E_1 \bigcap E_0 \bigcap \{ \tilde{Z}_0 = (x,u) \}$, $E_2^{(1)}$ and $E_2^{(2)}$ are independent. Thus, we deduce
$$III \geq \proba \left( E_2^{(1)} \bigg| E_1 \bigcap E_0 \bigcap \{ \tilde{Z}_0 = (x,u) \} \right)
\proba \left( E_2^{(2)} \bigg| E_1 \bigcap E_0 \bigcap \{ \tilde{Z}_0 = (x,u) \} \right).$$
Using Markov property of Brownian motions, we obtain that the right part of the inequality conditioned on $\{ \tilde{B}_{K+1}^{(3)} \text{ , } 
\Delta \tilde{B}_{[\tilde{\tau}_1^{(2)}, K+1]}^{3,\perp} \Big| 
E_1 \bigcap E_0 \bigcap \{ \tilde{Z}_0 = (x,u) \} \}$ is equal to
$$\left(1 - \int_0^2 p_1 \left(\tilde{B}_{K+1}^{(3)} , - \frac{\epsilon}{5 \sigma^+}, \frac{\epsilon}{5 \sigma^+}, t \right) dt \right)
\left( 1 - \int_0^1 p_1 \left(\Delta \tilde{B}_{[\tilde{\tau}_1^{(2)}, K+1]}^{3,\perp}, - \frac{g^+}{2 \sigma^+}, \frac{g^+}{2 \sigma^+}, t \right) dt \right)$$
$$\times \int_{-\frac{g^-}{2 \sigma^+}}^{\frac{g^-}{2 \sigma^+}} \int_1^2 p_1 \left(y, - \left( \frac{g^+}{\delta \sigma^{-}} + 
 \frac{\epsilon}{5 \sigma^+ \delta} \right) , \frac{g^+}{\delta \sigma^{-}} + 
 \frac{\epsilon}{5 \sigma^+ \delta}, t \right) dt dq(y),$$
where $q$ is the (conditional) distribution of $\Delta \tilde{B}_{[\tilde{\tau}_1^{(2)}, K+1]}^{3,\perp} + B_1$ 
conditioned on $$\bigg\{ \tau_{\Delta \tilde{B}_{[\tilde{\tau}_1^{(2)}, K+1]}^{3,\perp}}^{-\frac{g^-}{2 \sigma^+}, \frac{g^-}{2 \sigma^+}} \geq 1 \bigg\}.$$ 
Using the definition of $E_1$ together with (\ref{M}) and (\ref{compact}), we have $III$ which is uniformly bounded away from 0.

We prove that $IV$ is uniformly bounded away from $0$. Using (\ref{W21}) and (\ref{W22}), we deduce that $E_3^{(1)} \bigcap E_3^{(2)} \subset 
 E_3$ where
 \begin{eqnarray*}
  E_3^{(1)} & = & \Big\{ \underset{\tilde{\tau}_2^{(2)} \leq s \leq K+4}{\sup} \Big| \Delta \tilde{B}_{[\tilde{\tau}_2^{(2)}, s]}^{(4)} \Big| < 
  \frac{\epsilon \sigma^-}{5 \sigma^{+} \sigma^{(4)}} \Big\}, \\
 E_3^{(2)} & = & \Big\{ \forall s \in [\tilde{\tau}_2^{(2)}, K+4] \text{ } \Delta \tilde{B}_{[\tilde{\tau}_2^{(2)}, s]}^{4, \perp} \in 
 [y_3^{(1)}, y_3^{(2)}] \text{ , } 
 \Delta \tilde{B}_{[\tilde{\tau}_2^{(2)}, K+4]}^{4, \perp} \in [y_3^{(3)}, y_3^{(4)}] \Big\},
 \end{eqnarray*}
 with $y_3^{(1)} = \frac{d_1 (K) + 2 \epsilon/5}{\sigma^{(4)} \left( 1 - \left( \rho^{3,4} \right)^2 \right)^{1/2}}$, 
 $y_3^{(2)} = \frac{u_1 (K) - 2 \epsilon/5}{\sigma^{(4)} \left( 1 - \left( \rho^{3,4} \right)^2 \right)^{1/2}}$, 
 $y_3^{(3)} = \frac{u_1 (K) - 8 \epsilon/5}{\sigma^{(4)} \left( 1 - \left( \rho^{3,4} \right)^2 \right)^{1/2}}$, 
 as well as $y_3^{(4)} = \frac{u_1 (K) - 7 \epsilon/5}{\sigma^{(4)} \left( 1 - \left( \rho^{3,4} \right)^2 \right)^{1/2}}$.
Conditionally on $E_2 \bigcap E_1 \bigcap E_0 \bigcap \{ \tilde{Z}_0 = (x,u) \}$, $E_3^{(1)}$ and $E_3^{(2)}$ are independent. Thus, we deduce
$$IV \geq \proba \left( E_3^{(1)} \bigg| E_2 \bigcap E_1 \bigcap E_0 \bigcap \{ \tilde{Z}_0 = (x,u) \} \right)$$
$$\proba \left( E_3^{(2)} \bigg| E_2 \bigcap E_1 \bigcap E_0 \bigcap \{ \tilde{Z}_0 = (x,u) \} \right).$$
Using Markov property of Brownian motions, we obtain that the right part of the inequality conditioned on $\{ \tilde{\tau}_2^{(2)} \bigg| 
E_2 \bigcap E_1 \bigcap E_0 \bigcap \{ \tilde{Z}_0 = (x,u) \} \}$ is equal to
$$\left( 1 - \int_0^{K+4 - \tilde{\tau}_2^{(2)}} p_1 \left(0, - \frac{\epsilon \sigma^-}{5 \sigma^+ \sigma^{(4)}}, \frac{\epsilon \sigma^-}{5 \sigma^+ \sigma^{(4)}}, 
t \right) dt \right) \left( 1 - \int_0^{K+4 - \tilde{\tau}_2^{(2)}} p_1 \left( 0, y_3^{(1)}, y_3^{(2)}, t \right) dt \right)$$ 
$$ \times \int_{y_3^{(3)}}^{y_3^{(4)}} p_2 \left( 0, y_3^{(1)}, y_3^{(2)}, K + 4 - \tilde{\tau}_2^{(2)}, y \right) dy, $$
which is uniformly bounded away from 0 using (\ref{M}), (\ref{delta}) and (\ref{compact}).

We prove that $V > c (x_b - x_a) (u_b - u_a)$. Using (\ref{W11}) and (\ref{W12}), we deduce that $E_4^{(1)} \bigcap E_4^{(2)} \subset 
 E_4$ where
 \begin{eqnarray*}
  E_4^{(1)} & = & \Big\{ \tilde{\tau} \in [u_a + \tilde{\tau}_2^{(2)}, u_b + \tilde{\tau}_2^{(2)}] \text{ , } \tilde{X}_{\tilde{\tau}}^{(3)} = u_1 (K) \Big\}, \\
 E_4^{(2)} & = & \Big\{ \underset{K+4 \leq s \leq \tilde{\tau}}{\sup} \Big| \Delta \tilde{B}_{[K+4,s]}^{3, \perp} \Big| < \frac{5 g^-}
 {6 \sigma^{(4)} \left( 1 - \left( \rho^{3,4} \right)^2 \right)^{1/2}} \text{ , } \Delta \tilde{B}_{[L+4,\tilde{\tau}]}^{3, \perp} 
 \in [y_4^{(1)}, y_4^{(2)}] \Big\},
 \end{eqnarray*}
$\tilde{\tau} = \inf \{ t > K+4 \text{ : } \tilde{X}_t^{(3)} = u_1 (K) \text{ or } \Delta \tilde{X}_{[K+4,t]}^{(3)} = - 2 \epsilon \}$, 
$$y_4^{(1)} = \frac{x_a - \Delta \tilde{X}_{[\tilde{\tau}_2^{(2)}, K+4]}^{(4)} - \rho^{3,4} \sigma^{(4)} \left( \sigma^{(3)} \right)^{-1} 
\left( u_1 (K) - \tilde{X}_{K+4}^{(3)} \right)}
{\sigma^{(4)} \left( 1 - \left( \rho^{3,4} \right)^2 \right)^{1/2}},$$ 
and $y_4^{(2)} = \frac{x_b - \Delta \tilde{X}_{[\tilde{\tau}_2^{(2)}, K+4]}^{(4)} - \rho^{3,4} \sigma^{(4)} \left( \sigma^{(3)} \right)^{-1} 
\left( u_1 (K) - \tilde{X}_{K+4}^{(3)} \right)}
{\sigma^{(4)} \left( 1 - \left( \rho^{3,4} \right)^2 \right)^{1/2}}$.
We have
$$V = \proba \left( \tilde{X}_{\tilde{\tau}}^{(3)} = u_1 (K) 
 \right) $$ $$\times \proba \left( E_4^{(1)} \bigcap E_4^{(2)} \bigg| E_3 \bigcap E_2 \bigcap E_1 \bigcap E_0 
 \bigcap \{ \tilde{Z}_0 = (x,u) \} \bigcap \{ \tilde{X}_{\tilde{\tau}}^{(3)} = u_1 (K) \} \right). $$
 The first term on the right part of the equation is uniformly bounded away from 0 (Borodin and Salminen (2002)). 
 Because $\tilde{\tau}$ is a function of $\tilde{X}^{(3)}$ and $\tilde{B}^{3,\perp}$ is independent with $\tilde{X}^{(3)}$, $\tilde{\tau}$ and 
 $\tilde{B}^{3, \perp}$ are independent. 
 Thus the second term on the right conditioned on 
 $$\{ y_4^{(1)}, y_4^{(2)}, X_{K+4}^{(3)}, \tilde{\tau}_2^{(2)} \Big| E_3 \bigcap E_2 \bigcap E_1 \bigcap E_0 
 \bigcap \{ \tilde{Z}_0 = (x,u) \} \}$$
 can be expressed as :
 $$\int_{u_a + \tilde{\tau}_2^{(2)} - (K+4)}^{u_b + \tilde{\tau}_2^{(2)} - (K+4)} \int_{y_4^{(1)}}^{y_4^{(2)}} p_3 \left(\frac{X_{K+4}^{(3)}}{\sigma^{(3)}}, \frac{X_{K+4}^{(3)} - 2 \epsilon}{\sigma^{(3)}}, 
 \frac{u_1(K)}{\sigma^{(3)}}, t \right) p_2 \left( 0, - \frac{5 g^-}
 {y_4^{(3)}}, \frac{5 g^-}
 {y_4^{(3)}}, t,y \right) dt dy,$$
 where $y_4^{(3)} = 6 \sigma^{(4)} \left( 1 - \left( \rho^{3,4} \right)^2 \right)^{1/2}$. We have that $y_4^{(1)}$ and $y_4^{(2)}$ are dominated by $\frac{3 g^-}
 {4 \sigma^{(4)} \left( 1 - \left( \rho^{3,4} \right)^2 \right)^{1/2}}$. Using this together with (\ref{M}), (\ref{delta}) and (\ref{compact}), we 
 deduce that $V \geq c (x_b - x_a) (u_b - u_a)$.
 
\bigskip
 Second step : We prove that $\Big\rVert \psi^{AV} \Big\rVert_{\infty} := \underset{\sigma \in \mathcal{M}, g \in \mathcal{G}, (x, u) 
 \in \mathcal{S}_g}{\sup} \Big| \psi^{AV} \left( \sigma, g, x, u \right) \Big| < \infty$.
 To show this, we bound the term as
 $$\esp \left[ \left( \Delta \tilde{X}_{\tilde{\tau}_2^{1C}}^{(1)} 
\Delta \tilde{X}_{\tilde{\tau}_2^{1C,-,+}}^{(2)} - \tilde{\zeta}^{1,2} \Delta \tilde{\tau}_2^{1C} \right)^2 \right] 
\leq 2 \esp \left[ \left( \Delta \tilde{X}_{\tilde{\tau}_2^{1C}}^{(1)} 
\Delta \tilde{X}_{\tilde{\tau}_2^{1C,-,+}}^{(2)} \right)^2 + \left( \tilde{\zeta}^{1,2} \Delta \tilde{\tau}_2^{1C} \right)^2 \right].$$
The second term in the right hand-side of the inequality is uniformly bounded using (\ref{M}) and Lemma \ref{esptauk}. Using successively 
Cauchy-Schwarz and Burholder-Davis-Gundy inequality, (\ref{M}) and Lemma \ref{esptauk}, we can also bound uniformly the first term. 
The other term of (\ref{psiAV1}) can be bounded in the same way.

\bigskip
 Third step : Define $q = (\sigma, g, x, u)$ and $$\mathcal{Q} = \left\{ (\sigma, g, x ,u) \text{ s.t. } \sigma \in \mathcal{M}, g \in \mathcal{G}, (x, u) 
 \in \mathcal{S}_g \right\}.$$ Prove that $\forall q \in \mathcal{Q}$, 
 there exists a measure $\tilde{\pi} \left( \sigma, g \right)$ such that
 $$\underset{q \in \mathcal{Q}}{\sup} \bigg| \sum_{l=0}^{n-1}  \int_{\reels^2} \psi^{AV} \left(\sigma, g, y, v \right) d \tilde{\pi}_l 
 \left(\sigma, g, x, u \right) \left(y, v \right) - n \int_{\reels^2} \psi^{AV} \left(\sigma, g, y, v \right) d \tilde{\pi} 
 \left(\sigma, g \right) \left(y, v \right) \bigg|$$ $$= n o_p(1).$$
 To show this, we use first step together with $Th.16.0.2$ $(v)$ (Meyn and Tweedie (2009)). We obtain that there exists $\tilde{\pi} \left( \sigma, g \right)$ where 
 $$\Big\rVert P^n\left( \sigma, g \right) \left(\left(x, u \right), . \right) - \tilde{\pi} \left( \sigma, g \right) \Big\rVert_{TV} \leq 2 r^{n}$$
 and $r = 1 - \nu \left( \reels^2 \right)$. Thus, we deduce :
 $$\bigg| \int_{\reels^2} \psi^{AV} \left(\sigma, g, y, v \right) d \tilde{\pi}_l 
 \left(\sigma, g, x, u \right) \left(y, v \right) - \int_{\reels^2} \psi^{AV} \left(\sigma, g, y, v \right) d \tilde{\pi} 
 \left(\sigma, g \right) \left(y, v \right) \bigg|$$
 \begin{eqnarray}
 \label{mark1}
 \leq \Big\rVert \psi^{AV} \Big\rVert_{\infty} \Big\rVert \tilde{\pi}_l \left(\sigma, g, x, u \right) - \tilde{\pi} 
 \left(\sigma, g \right) \Big\rVert_{TV} \leq 2 \Big\rVert \psi^{AV} \Big\rVert_{\infty} r^l.
 \end{eqnarray}
 We want to show that $\forall \epsilon > 0$, $\exists N > 0$ such that $\forall n \geq N$ :
 $$\bigg| \sum_{l=0}^{n-1} \int_{\reels^2} \psi^{AV} \left(\sigma, g, y, v \right) d \tilde{\pi}_l 
 \left(\sigma, g, x, u \right) \left(y, v \right) - n \int_{\reels^2} \psi^{AV} \left(\sigma, g, y, v \right) d \tilde{\pi} 
 \left(\sigma, g \right) \left(y, v \right) \bigg| $$
 \begin{eqnarray}
 \label{mark2}
 < \epsilon n.
 \end{eqnarray}
 
The rest is a straightforward analysis exercise. Let $\epsilon > 0$. $\exists N_1 > 0$ such that $r^{N_1} < \frac{\epsilon}{2}$. Choosing 
$N > 8 N_1 \epsilon^{-1} \rVert \psi^{AV} \rVert_{\infty}^{-1}$, we first use the triangular inequality, and then 
split the sum of the left part of (\ref{mark2}) in two parts, one up to $N_1$ and the other one up to N. We use (\ref{mark1}) in the second part to obtain (\ref{mark2}).

\bigskip
Fourth step : Proving the Lemma. Let $w > 0$. From Lemma \ref{sumtight}, we just have to show that
$$ \alpha_n^{2} \sum_{i=1}^{\llcorner 
 w \alpha_n^{-2} h \left( n \right)^{-1} \lrcorner} 
 \bigg| \sum_{j=0}^{h_n - 2} \int_{\reels^2} \psi^{AV} \left(\sigma_{\tau_{i-1,n}^h}, g_{\tau_{i-1,n}^h}, 
\alpha_n^{-1} y, \alpha_n^{-2} v \right) d \tilde{\pi}_{i-1,j,n} \left( y,v \right)$$ 
$$- h_n \phi^{AV} \left(\sigma_{\tau_{i-1,n}^h}, g_{\tau_{i-1,n}^h} \right) \bigg|$$
tends to 0 in probability. Using third step together with standard results on regular conditional distributions (see for instance 
Section $4.3$ (pp. $77-80$) in Breiman (1992)), we prove the lemma.
 
 \end{proof}

\begin{approxtau}
We have
 \label{approxtau}
$$\alpha_n^2 \sum_{i \in A_n} \esp_{\tau_{i-1,n}^h} \left[ 
 \left( \sigma_{\tau_{i-1,n}^h}^{(1)} \right)^2  \left( \sigma_{\tau_{i-1,n}^h}^{(2)} \right)^2 
h_n \phi^{AV} \left(\sigma_{\tau_{i-1,n}^h}, g_{\tau_{i-1,n}^h} \right) \Delta \tau_{i,n}^h \left( 
\esp_{\tau_{i-1}^h} \left[ \Delta \tau_{i,n}^h \right] \right)^{-1} \right]$$
$$ = \sum_{i \in A_n} \esp_{\tau_{i-1,n}^h} \left[  \phi^{AV} \left(\sigma_{\tau_{i-1,n}^h}, g_{\tau_{i-1,n}^h} \right) \Delta \tau_{i,n}^h
\left( \phi_{\tau_{i-1,n}^h}^{\tau} \right)^{-1} \right] + o_p (1).$$
\end{approxtau}

\begin{proof}
 First step : Defining 
 $$
   \left \{
   \begin{array}{r c l}
      u_{i,n} & := & \sum_{j=0}^{h_n - 2} \int_{X} \psi^{\tau} \left(\sigma_{\tau_{i-1,n}^h}, g_{\tau_{i-1,n}^h}, 
x, u \right) d \tilde{\pi}_{i-1,j,n} \left( x,u \right), \\
      A_0 & := & \alpha_n^2 \sum_{i \in A_n} \esp_{\tau_{i-1,n}^h} \left[ h_n \phi^{AV} \left(\sigma_{\tau_{i-1,n}^h}, g_{\tau_{i-1,n}^h} \right) 
      \Delta \tau_{i,n}^h \left( 
  \esp_{\tau_{i-1}^h} \left[ \Delta \tau_{i,n}^h \right] \right)^{-1} \right], \\
A_1 & := & \alpha_n^2 \sum_{i \in A_n} \esp_{\tau_{i-1,n}^h} \left[ h_n \phi^{AV} \left(\sigma_{\tau_{i-1,n}^h}, g_{\tau_{i-1,n}^h} \right) 
\Delta \tau_{i,n}^h \left( u_{i,n} \right)^{-1} \right].
   \end{array}
   \right .
$$
we have that $A_0 = A_1 + o_p \left( 1 \right)$. To show this, in light of Lemma \ref{tauapp}, we have that 
$$\bigg| \esp_{\tau_{i-1,n}^h} \left[ \Delta \tau_{i,n}^h \right] - u_{i,n} \bigg| \leq h \left( n \right) C_n,$$
where $C_n$ tends to 0 in probability. From this, we can easily show that $A_0 = A_1 + o_p \left( 1 \right)$.

\bigskip
Second step : We have that 
$$A_1 = \sum_{i \in A_n} \esp_{\tau_{i-1,n}^h} \left[  \phi^{AV} \left(\sigma_{\tau_{i-1,n}^h}, g_{\tau_{i-1,n}^h} \right) \Delta \tau_{i,n}^h
\left( \phi_{\tau_{i,n}^h}^{\tau} \right)^{-1} \right] + o_p (1).$$
To prove it, we can mimic the proof of Lemma \ref{approxdistrib}, together with Lemma \ref{tauapp}.
\end{proof}

\subsection{Computation of the limits of $\langle M^n \rangle_t$, $\langle M^n, X^{(1)} \rangle_t$ 
and $\langle M^n, X^{(2)} \rangle_t$}

\begin{eqnarray*}
\langle M^n \rangle_t & = & \sum_{i \in A_n}
\esp_{{\tau}_{i-1,n}^h} \left[ \left( \Delta M_{\tau_{i,n}^h}^n \right)^2 \right] + o_p (1) \\
& = & \alpha_n^{-2} \sum_{i \in A_n}
\esp_{\tau_{i-1,n}^h} \left[ \sum_{u=2}^{h_n} \left( N_{(i-1)h_n + u} \right)^2 + 2 N_{(i-1)h_n + u} 
N_{(i-1)h_n + u+1} \right] + o_p (1) \\
& = & \alpha_n^{2} \sum_{i \in A_n} \sum_{j=0}^{h_n - 2} \int_{\reels^2} \psi^{AV} \left( \sigma_{\tau_{i-1,n}^h}, g_{\tau_{i-1,n}^h}, 
\alpha_n^{-1} x, \alpha_n^{-2} u \right) d \tilde{\pi}_{i-1,j,n} \left( x,u \right) + o_p (1),
\end{eqnarray*}
where we used Lemma 2.2.11 of Jacod and Protter (2012) in first equality, Lemma \ref{termsvanishing} in second equality, Lemma \ref{holdingconst} 
in third equality.

We deduce (using Lemma \ref{approxdistrib} in first equality and Lemma \ref{approxtau} 
in third equality)
\begin{eqnarray*}
\langle M^n \rangle_t & = & \alpha_n^{2} \sum_{i \in A_n} 
h_n \phi_{\tau_{i-1,n}^h}^{AV} + o_p (1)\\
 & = & \alpha_n^2 \sum_{i \in A_n} \esp_{\tau_{i-1,n}^h} \left[ h_n \phi_{\tau_{i-1,n}^h}^{AV} \Delta \tau_{i,n}^h \left( 
\esp_{\tau_{i-1}^h} \left[ \Delta \tau_{i,n}^h \right] \right)^{-1} \right] + o_p (1)\\
& = & \sum_{i \in A_n} \esp_{\tau_{i-1,n}^h} \left[ 
\phi_{\tau_{i-1,n}^h}^{AV} \Delta \tau_{i,n}^h
\left( \phi_{\tau_{i,n}^h}^{\tau} \right)^{-1} \right] + o_p (1).
\end{eqnarray*}

Using Lemma 2.2.11 of Jacod and Protter (2012) again, we deduce 
\begin{eqnarray*}
\langle M^n \rangle_t & = & \sum_{i \in A_n} 
\phi_{\tau_{i-1,n}^h}^{AV} \Delta \tau_{i,n}^h
\left( \phi_{\tau_{i,n}^h}^{\tau} \right)^{-1} + o_p (1).\\
\end{eqnarray*}

Using Lemma \ref{snk} together with Prop. $I.4.44$ (page 51) in Jacod and Shiryaev (2003), we obtain
\begin{eqnarray} \label{mm}
   \langle M^n \rangle_t \rightarrow \int_0^{t} \phi_{s}^{AV} \left( \phi_{s}^{\tau} \right)^{-1} ds.
\end{eqnarray}

Using the same approximations and computations, we also compute
\begin{eqnarray} \label{mx1}
   \langle M^n, X^{(1)} \rangle_t & \rightarrow & \int_0^{t} \phi_{s}^{AC1}  
\left( \phi_s^{\tau} \right)^{-1} ds, \\
\label{mx2}
\langle M^n, X^{(2)} \rangle_t & \rightarrow & \int_0^{t} \phi_{s}^{AC2}  
\left( \phi_s^{\tau} \right)^{-1} ds. 
\end{eqnarray}

\subsection{Computation of the asymptotic bias and variance}
We follow the idea in 1-dimension in pp. 155-156 of Mykland and Zhang (2012), and define an auxiliary 
martingale
$$\tilde{M}_t^n = M_t^n - \int_0^t k_s^{(1)} dX_s^{(1)} - \int_0^t k_s^{1,\perp} dX_s^{1,\perp},$$
where $X_t^{1,\perp}$ is defined in (\ref{X1perp}). Using \eqref{mx1}, we deduce 
\begin{eqnarray*}
\langle \tilde{M}^n, X^{(1)} \rangle_t & = & 
\langle M^n, X^{(1)} \rangle_t - \int_0^t k_s^{(1)} d \langle X^{(1)} \rangle_s\\
& \overset{ \proba }{\rightarrow} & \int_0^{t}  
\phi_s^{AC1} \left( \phi_s^{\tau} \right)^{-1} ds - \int_0^t k_s^{(1)} \left( 
\sigma_s^{(1)} \right)^2 ds. 
\end{eqnarray*}
Hence, we choose 
$$k_s^{(1)} = \left( \sigma_s^{(1)} \right)^{-2} \phi_s^{AC1} \left( \phi_s^{\tau} \right)^{-1}.$$
By the same techniques that we used to compute \eqref{mx1}, we have that 
\begin{eqnarray} \label{mx12}
   \langle M^n, \int_0^{.} \rho_s^{1,2} \sigma_s^{(2)} dB_s^{(1)} \rangle_t \rightarrow \int_0^{t} \left( \sigma_s^{(1)} \right)^{-1} \sigma_s^{(2)} 
  \rho_s^{1,2} 
\phi_s^{AC1} \left( \phi_s^{\tau} \right)^{-1} ds. 
\end{eqnarray}
Using \eqref{mx2} and \eqref{mx12} we compute
\begin{eqnarray*}
\langle \tilde{M}^n, X^{1, \perp} \rangle_t & = & 
\langle M^n, X^{1,\perp} \rangle_t - \int_0^t k_s^{1,\perp} d \langle X^{1,\perp} \rangle_s\\
& = & \langle M^n, X^{(2)} - \int_0^{.} \rho_s \sigma_s^{(2)} dB_s^{(1)} \rangle_t - \int_0^t k_s^{1,\perp} d \langle X^{1,\perp} \rangle_s\\
& = & \langle M^n, X^{(2)} \rangle - \langle M^n , \int_0^{.} \rho_s \sigma_s^{(2)} dB_s^{(1)} \rangle_t - 
\int_0^t k_s^{1,\perp} d \langle X^{1,\perp} \rangle_s\\
& \overset{ \proba }{\rightarrow} & \int_0^{t} \left(\phi_s^{AC2} - \left(\sigma_s^{(1)} \right)^{-1} \sigma_s^{(2)} \rho_s^{1,2} \phi_s^{AC1} \right) 
\left( \phi_s^{\tau} \right)^{-1} ds \\
& & - \int_0^t k_s^{1, \perp} \left(1 - \left( \rho_s^{1,2} \right)^2 \right) 
\left( \sigma_s^{(2)} \right)^2 ds.
\end{eqnarray*}
Hence, we choose 
$$k_s^{1,\perp} = \left(1 - \left( \rho_s^{1,2} \right)^2 \right)^{-1} \left( \left( \sigma_s^{(2)} \right)^{-2} \phi_s^{AC2}
- \left(\sigma_s^{(1)} \sigma_s^{(2)} \right)^{-1} \rho_s^{1,2} \phi_s^{AC1} \right) \left( \phi_s^{\tau} \right)^{-1}.$$
By $(A4)$, there exists $S > 0$ such that the $S$ Brownian motions $\{ D^{(1)}, ... , D^{(S)} \}$ generate the filtration $\left( \mathcal{F}_t \right)_{t \geq 0}$. To show that 
$\langle \tilde{M}^n, D^{(s)} \rangle_t$ tends to 0 in probability, we decompose $D^{(s)} = D^{s,1} + D^{s,2}$ where $D^{s,1}$ belongs to the space spanned 
by $\{ X^{(1)}, X^{(2)} \}$, $D^{s,2}$ is orthogonal to this space. By what precedes, we have clearly $\langle \tilde{M}^n, D^{s,1} \rangle_t$ tends to 0 
in probability. Also, $D^{s,2}$ is a martingale that is, conditionally on the observations times of both processes, independent of $\tilde{M}^n$. Thus we also 
deduce that $\langle \tilde{M}^n, D^{s,2} \rangle_t$ converges to 0 in probability.

We can now compute
\begin{eqnarray*}
\langle \tilde{M}^n \rangle_t & = & 
\langle M^n - \int_0^. k_s^{(1)} dX_s^{(1)} - \int_0^. k_s^{1,\perp} dX_s^{1,\perp} \rangle_t\\
& = & \langle M^n \rangle_t + \int_0^{t} \left( \sigma_s^{(1)} \right)^2 \left( k_s^{(1)} \right)^2 ds + 
\int_0^{t} \left( \sigma_s^{(2)} \right)^2 \left( 1 - \left( \rho_s^{1,2} \right)^2 \right) \left( k_s^{1, \perp} \right)^2 ds \\
& - & 2 \int_0^t k_s^{(1)} d \langle X^{(1)}, M^n \rangle_s - 2 \int_0^t k_s^{1, \perp} d \langle X^{1, \perp}, M^n \rangle_s\\
& \overset{ \proba }{\rightarrow} & \int_0^{t} \left( \phi_s^{AV} + 2 \left( k_s^{(1)} \left(\sigma_s^{(1)} \right)^{-1} \sigma_s^{(2)} \rho_s^{1,2} \phi_s^{AC1} 
- \left( k_s^{1} + k_s^{1, \perp} \right) \phi_s^{AC2} \right) \right) \left( \phi_s^{\tau} \right)^{-1}\\
& + & \left( \sigma_s^{(1)} \right)^2 \left( k_s^{(1)} \right)^2 
+ \left( \sigma_s^{(2)} \right)^2 \left( 1 - \left( \rho_s^{1,2} \right)^2 \right) \left( k_s^{1, \perp} \right)^2 ds.
\end{eqnarray*}
By letting
\begin{eqnarray*}
 AV_s & = & \left( \phi_s^{AV} + 2 \left( k_s^{(1)} \left(\sigma_s^{(1)} \right)^{-1} \sigma_s^{(2)} \rho_s^{1,2} \phi_s^{AC1} 
- \left( k_s^{(1)} + k_s^{1, \perp} \right) \phi_s^{AC2} \right) \right) \left( \phi_s^{\tau} \right)^{-1} \\
& & + \left( \sigma_s^{(1)} \right)^2 \left( k_s^{(1)} \right)^2 
+ \left( \sigma_s^{(2)} \right)^2 \left( 1 - \left( \rho_s^{1,2} \right)^2 \right) \left( k_s^{1, \perp} \right)^2,
\end{eqnarray*}
we deduce using Theorem 2.28 in Mykland and Zhang (2012) that stably in law as $\alpha_n \rightarrow 0$,  
 $$\alpha_n^{-1} \left(\widehat{RCV}_{t,n} - RCV_t \right) \rightarrow \int_0^t k_s^{(1)} dX_s^{(1)} + 
 \int_0^t k_s^{1, \perp} dX_s^{1, \perp} + \int_0^t \left( 
 AV_s \right)^{1/2} d \tilde{W}_s.$$
We have just shown Theorem \ref{main1}. Now, we express the asymptotic bias $AB_t = \int_0^t k_s^{(1)} dX_s^{(1)} + 
 \int_0^t k_s^{1, \perp} dX_s^{1, \perp}$ differently as
 \begin{eqnarray*}
  AB_t & = & \int_0^t k_s^{(1)} dX_s^{(1)} + \int_0^t k_s^{1, \perp} (1 - \left( \rho_s^{1,2} \right)^2)^{1/2} \sigma^{(2)}_{s} dB_{s}^{1,\perp} \\
  & = & \int_0^t k_s^{(1)} dX_s^{(1)} - \int_0^t k_s^{1, \perp} \rho_s^{1,2} \sigma^{(2)}_{s} dB_{s}^{(1)} 
  + \int_0^t k_s^{1, \perp} \rho_s^{1,2} \sigma^{(2)}_{s} dB_{s}^{(1)} \\
& & + 
  \int_0^t k_s^{1, \perp} \left( 1 - \left( \rho_s^{1,2} \right)^2 \right)^{1/2} \sigma^{(2)}_{s} dW_{s}^{1,\perp} \\
 & = & \int_0^t \left( k_s^{(1)} - k_s^{1, \perp} \rho_s^{1,2} \sigma^{(2)}_{s} \left( \sigma^{(1)}_{s} \right)^{-1} \right) dX_s^{(1)}
  + \int_0^t k_s^{1, \perp} dX_s^{(2)}.
 \end{eqnarray*}
 We thus deduce the expression of $AB_s^{(1)}$ and $AB_s^{(2)}$.
 
 \bigskip
 The proof of Corollary \ref{estcor} follows in the same way as the proof of Theorem \ref{main1}. We hold constant the asymptotic variance and the asymptotic bias on blocks of size $h_n$. Moreover, we can see that $\widehat{AB}_{i,\alpha}^{(1)}$, $\widehat{AB}_{i,\alpha}^{(2)}$ and $\widehat{AV}_{i,\alpha}$ are uniformly consistent estimators under the constant model.
 
 \subsection{Discussion on the adaptation of Theorem \ref{main1} proofs for more general models}
 \label{adaptationproofs}
We discuss in this section how to adapt the proofs of Theorem \ref{main1} when considering Example \ref{hittingbarriernoisejump} up to Example \ref{autoregressive}. In that case, the HBT can be defined for each $k=1,2$ as $\tau_{0,n} := 0$ and recursively as
\begin{eqnarray}
 \label{generateobstimesdiscussion} \tau_{i,n}^{(k)} := \inf \Big\{ t > \tau_{i-1,n}^{(k)} : \Delta X_{[\tau_{i-1,n}^{(k)}, t]}^{(t,k)} \notin \big[ \alpha_n d_{t,n}^{(k)} \big( t - \tau_{i-1,n}^{(k)} \big), \alpha_n u_{t,n}^{(k)} \big(t - \tau_{i-1,n}^{(k)} \big) 
\big] \Big\}
\end{eqnarray}
for any positive integer $i$. In (\ref{generateobstimesdiscussion}), the grid $g_{t,n}^{(k)} := (d_{t,n}^{(k)}, u_{t,n}^{(k)})$ depends on $n$, thus the term $g_t^{(k)}$ in the asymptotic variance obtained in Theorem \ref{main1} will have a different interpretation. Indeed, $g_t^{(k)}$ will be seen as a (possibly multidimensional) continuous time-varying parameter which generates (\ref{generateobstimesdiscussion}) instead of the scaled grid function itself. In particular, the approximations will not be carried with holding $g_{t,n}$ constant on each block, but rather with holding $g_t$ constant on each block. Also, for any fixed $t \in [0,1]$, $g_t^{(k)}$ will not be a function on $\reels^+$, but a simple vector. The reader can refer to Potiron (2016) for the notion of time-varying parameter. Note that Assumption (A$3$) is only used in Lemma \ref{tauapp} and Lemma \ref{approxdistrib}. Thus, Lemma \ref{tauapp} and Lemma \ref{approxdistrib} are the only parts in the proof which need to be adapted.

\subsubsection{Example \ref{hittingbarriernoisejump} (hitting constant boundaries of the jump size)}
 For each asset $k=1, 2$ we define the jump sizes as $L_{i,n}^{(k)}$. We assume that $L_{i,n}^{(1)}$ and $L_{i,n}^{(2)}$ are independent of each other. We have that $g_{t,n}^{(k)} (s) := ( - L_{i-1,n}^{(k)}, L_{i-1,n}^{(k)})$ for $t \in (\tau_{i-1,n}^{(k)}, \tau_{i,n}^{(k)}]$. We also have a non-time varying parameter $g_t := 1$. 
 
 \bigskip
 As the jump size $L_{i,n}^{(k)}$ are IID and independent of the other quantities, we can consider the same $L_{i,n}^{(k)}$ when making local approximations. Note that in Lemma \ref{tauapp}, the proof is made recursively for each observation time of the block. Thus a "jump" of $g_{t,n}$ is not a problem when it happens exactly at observation times, as long as the same jump is also made in the approximation block. Since $L_{i,n}^{(k)}$ is assumed to be bounded, it is straightforward to adapt the proof of Lemma \ref{tauapp}.
 
 \bigskip
 We discuss now how to adapt the proof of Lemma \ref{approxdistrib}. To do that, we consider the Markov chain $\tilde{Z}_i := \big( \Delta \tilde{X}_{[\tilde{\tau}_{i}^{1C,-}, \tilde{\tau}_{i}^{1C}]}^{(2)}, 
\tilde{\tau}_{i}^{1C} - \tilde{\tau}_{i}^{1C,-}$,  $L_{i'}^{(1)}$, $L_{j'}^{(2)} \big)$, where $i'$ is such that  $\tilde{\tau}_{i'}^{(1)} = \tilde{\tau}_{i}^{1C}$, $j'$ is such that $\tilde{\tau}_{j'}^{(2)}  = \tilde{\tau}_{i}^{1C,-}$, $L_{i}^{(1)}$ and $L_{i}^{(2)}$ are IID sequences independent of each other which follows respectively the distribution of $L_{i,1}^{(1)}$ and $L_{i,1}^{(2)}$. Then, everything follows the same way as in the proof of Lemma \ref{approxdistrib}. 

\subsubsection{Example \ref{uncertaintyzones} (model with uncertainty zones)}
 This model is very similar to Example \ref{hittingbarriernoisejump}, except that the sequence $L_{i,n}^{(k)}$ is obtained as a function of $\chi_{\tau_{i,n}}^{(k)}$, where $\chi_t^{(k)}$ corresponds to the continuous time-varying parameter $\chi_t$ of the $k$th asset introduced in p. 5 of Robert and Rosenbaum (2012). We thus consider $g_t^{(k)} := \chi_t^{(k)}$. The proof of Lemma \ref{tauapp} can be extended using the convenient construction of $L_{i,n}^{(k)}$ provided in p. 11 of Robert and Rosenbaum (2012). We extend this construction in two-dimension assuming that $(W_t')^{(1)}$ and $(W_t')^{(2)}$ are independent. As Example \ref{uncertaintyzones} is slightly more involved than Example \ref{hittingbarriernoisejump}, the Markov chain $\tilde{Z}_i$ needs to include also the type of previous price change (increment or decrement) for each asset. We thus consider $\tilde{Z}_i := \big( \Delta \tilde{X}_{[\tilde{\tau}_{i}^{1C,-}, \tilde{\tau}_{i}^{1C}]}^{(2)}, 
\tilde{\tau}_{i}^{1C} - \tilde{\tau}_{i}^{1C,-}$,  $L_{i'}^{(1)}$, $L_{j'}^{(2)}, \text{sign} (  \Delta \tilde{X}_{\tilde{\tau}_{i'}^{(1)}}^{(1)} ), \text{sign} (  \Delta \tilde{X}_{\tilde{\tau}_{j'}^{(2)}}^{(2)} ) \big)$, and can follow the same line of reasoning as in Lemma \ref{approxdistrib}.
 
\subsubsection{Example \ref{irregulargrid} (times generated by hitting an irregular grid model)}
In this case, the parameter $g_t^{(k)} := 1$ is non time-varying. Lemma \ref{tauapp} can be adapted easily. To show Lemma \ref{approxdistrib}, a further condition is needed on $q_j^{(k)} := p_j^{(k)} - p_{j-1}^{(k)}$. We assume that there exists a positive number $Q^{(k)}$ such that for any non-negative number $j$ and any $l \in \{0, \cdots, Q^{(k)}-1\}$ we have $q_{j Q^{(k)} + l}^{(k)} = q_{l}^{(k)}$. We also define the Markov chain $\tilde{Z}_i := \big( \Delta \tilde{X}_{[\tilde{\tau}_{i}^{1C,-}, \tilde{\tau}_{i}^{1C}]}^{(2)}, 
\tilde{\tau}_{i}^{1C} - \tilde{\tau}_{i}^{1C,-}, l^{(1)}, l^{(2)} \big)$, where $l^{(1)}$ is the index such that there exists a non-negative number $m$ with $p_{m Q^{(1)} + l^{(1)}}^{(1)} = \tilde{X}_{\tilde{\tau}_{i}^{1C}}^{(1)}$, and $l^{(2)}$ is the index such that there exists a non-negative number $m$ with $p_{m Q^{(2)} + l^{(2)}}^{(2)} = \tilde{X}_{\tilde{\tau}_{i}^{1C,-}}^{(2)}$. Under this assumption, we can show Lemma \ref{approxdistrib}.
\subsubsection{Example \ref{autoregressive} (structural autoregressive conditional duration model)}
We assume that the mixing variables $\tilde{d}_{\tau_{i,n}}^{(k)}$ and $\tilde{c}_{\tau_{i,n}}^{(k)}$ are interpolated by time-varying continuous stochastic parameters $(\tilde{d}_{t}^{(k)}, \tilde{c}_{t}^{(k)})$. We have that $g_t^{(k)} := (\tilde{d}_{t}^{(k)}, \tilde{c}_{t}^{(k)})$. The central limit theorem in Example \ref{autoregressive} can be obtained as a straightforward corollary of Theorem \ref{main1}. If we define for any $s \geq 0$ the grid functions $g_t^{(k)} (s) := (\tilde{d}_{t}^{(k)}, \tilde{c}_{t}^{(k)})$, the only difference between the HBT model (\ref{generateobstimes}) and the structural ACD model (\ref{autoregressiveACD}) is that we hold the grid between two observations in the latter model. In view of this specific assumption which implicates that the quantities of approximation are closer to the approximated quantities than under the HBT model, the proof of Lemma \ref{tauapp} simplifies. The proof of Lemma \ref{approxdistrib} remains unchanged as it deals only with quantities of approximation.

\subsection{Jump case: proof of Remark \ref{rkjumps}}
 \label{proofjumps}
 We update in this section the proof in the jump case model (\ref{jumpmodel}). The idea is to exclude all the blocks where we observe a jump. Such blocks will be finitely counted, and we will have at most one jump (either for $Y_t^{(1)}$ or for $Y_t^{(2)}$ but not for both prices at the same time) in each block. This is the main difference with the one-dimensional case.
 
 \bigskip
 We introduce the notation 
 $$A_n^{(no)} := \big\{ i \geq 1 \text{ s.t. } \tau_{i-1,n}^h \leq t \text{ and there is no jumps on } [\tau_{i-1,n}^h, \tau_{i,n}^h] \big\}.$$ 
 The proof of Lemma \ref{snk} can be adapted because of the finiteness of jumps. The proof of Lemma \ref{s} remains unchanged. Lemma \ref{esptauk} and Lemma \ref{numb} remains true in view of the finiteness of jumps. Lemma \ref{sumtight} and Lemma \ref{scale} don't need any change. We modify Lemma \ref{tauapp} as follows. Let $l \geq 1$, we have that 
 \begin{eqnarray*}
 \sup_{i \in A_n^{(no)} \text{ , } 2 \leq j \leq h_n} \esp \left[ \Big| \Delta \tau_{i,j,n}^{1C} - 
 \Delta \tilde{\tau}_{i,j,n}^{1C} \Big|^l \right] = o_p \left( \alpha_n^{2l} \right)
 \end{eqnarray*}
 and
 \begin{eqnarray*}
 \sup_{i \in A_n^{(no)} \text{ , } 2 \leq j \leq h_n} \esp \left[ \Big| \Delta \tau_{i,j,n}^{1C,-,+} - 
 \Delta \tilde{\tau}_{i,j,n}^{1C,-,+} \Big|^l \right] = o_p \left( \alpha_n^{2l} \right)
 \end{eqnarray*}
The proof remains unchanged in view of the independence assumption between jumps and the other quantities. Lemma \ref{termsvanishing} stays true with no further change. We introduce the new following lemma to be inserted between Lemma \ref{termsvanishing} and Lemma \ref{holdingconst} in the proofs.
\begin{lemmajumps}
\label{lemmajumps}
We have
$$\alpha_n^{-2} \sum_{i \in A_n}
\esp_{\tau_{i-1,n}^h} \left[ \sum_{u=2}^{h_n} \left( N_{(i-1)h_n + u} \right)^2 + 2 N_{(i-1)h_n + u} 
N_{(i-1)h_n + u+1} \right]$$
$$= \alpha_n^{-2} \sum_{i \in A_n^{(no)}}
\esp_{\tau_{i-1,n}^h} \left[ \sum_{u=2}^{h_n} \left( N_{(i-1)h_n + u} \right)^2 + 2 N_{(i-1)h_n + u} 
N_{(i-1)h_n + u+1} \right] + o_p (1)$$
\end{lemmajumps}
\begin{proof}
This is a simple consequence to the fact that we have at most one jump in $\Delta X_{\tau_{i,n}^{1C}}^{(1)}$ or
$\Delta X_{\tau_{i,n}^{1C,-,+}}^{(2)}$ asymptotically, together with the finiteness of jumps.
\end{proof}
Starting from Lemma \ref{holdingconst} up to the end of the proof of Theorem \ref{main1}, in view of Lemma \ref{lemmajumps}, we can use "$i \in A_n^{(no)}$" in lieu of "$i \in A_n$". We have thus proved that Theorem \ref{main1} is robust to jumps. 

 \bibliographystyle{plain}

 \begin{figure}[p!]
\includegraphics[width=\linewidth]{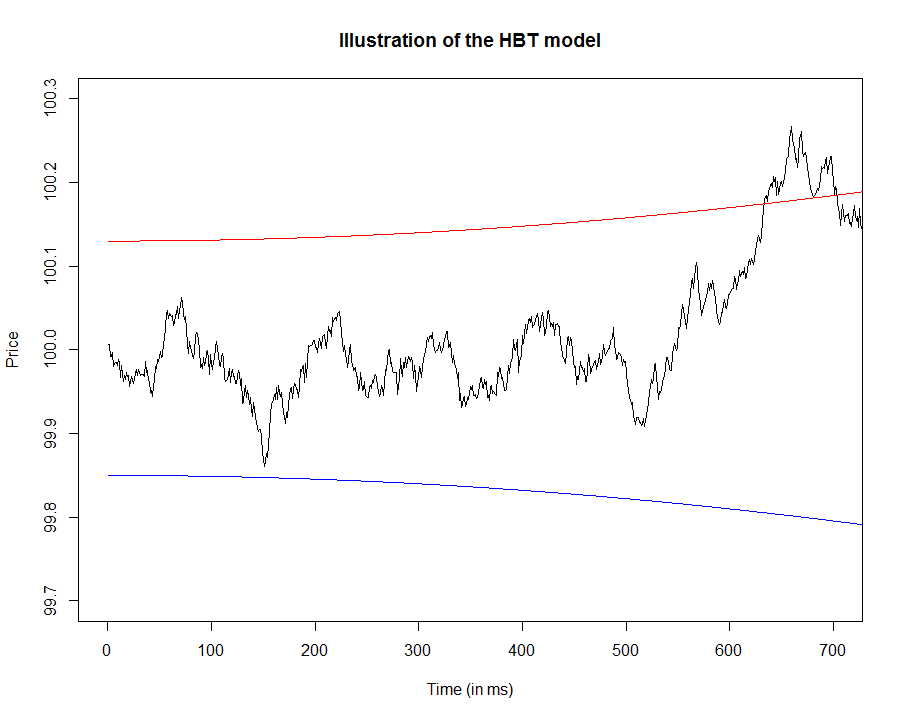}
\centering
\caption{This is an illustration of the HBT model when starting at time $\tau_0 =0$ and with $X_0 = 100$. The black stochastic process represents $X_t$, the red line stands for $100 + u_t (t)$ and the blue line for $100 + d_t (t)$. Furthermore, we assume that $X_t^{(t)} = X_t$. The second observation $\tau_1$ is obtained when $X_t$ crosses the red line for the first time.}
\label{illustration}
\end{figure}

   \begin{figure}[p!]
\includegraphics[width=\linewidth]{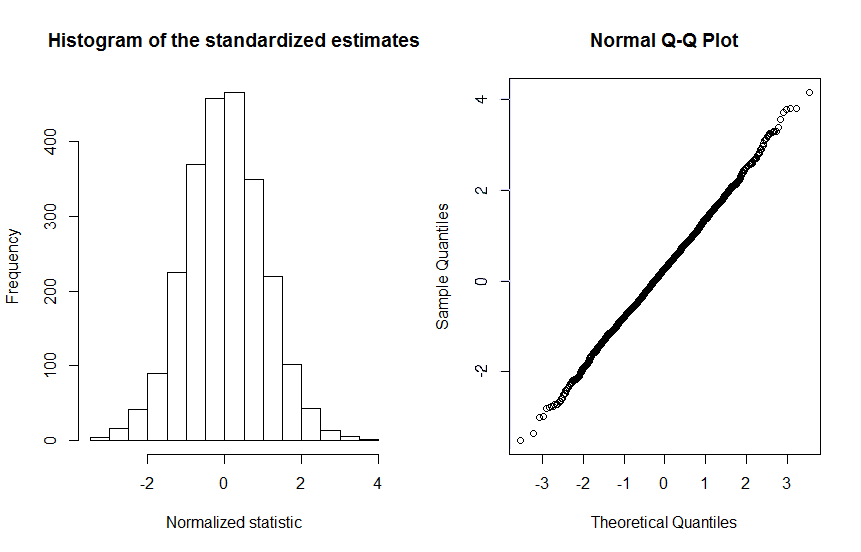}
\centering
\caption{Histogram and Normal QQ-plot of the standardized estimates (\ref{feasiblestat}) in setting $1$ on a $10$-year period of observations.}
\label{plotfeasiblestat}
\end{figure}

\begin{table}[p!]
 \centering
  \begin{tabular}{| c | c | c | c | c | c |}
   \hline
     No. years & estim & setting & sample bias &  RMSE & \% Reduced RMSE\\
  \hline
    1 &  HY & 1 &$5.41e-07$ &   $1.36e-05$ & -\\
    1 & BCHY & 1 & $5.43e-07$ &  $1.19e-05$ & $13 \%$\\
    5 &  HY & 1 & $1.10e-07$ &   $1.42e-05$ & -\\
    5 & BCHY & 1 & $1.07e-07$ &  $1.26e-05$ & $11 \%$\\
    10 &  HY & 1 & $5.54e-08$ &   $1.39e-05$ & -\\
    10 & BCHY & 1 & $5.53e-08$ &  $1.20e-05$ & $14 \%$\\
    1 &  HY & 2 &$5.47e-07$ &   $1.66e-05$ & -\\
    1 & BCHY & 2 & $5.44e-07$ &  $1.50e-05$ & $9 \%$\\
    5 &  HY & 2 & $1.13e-07$ &   $1.71e-05$ & -\\
    5 & BCHY & 2 & $1.15e-07$ &  $1.58e-05$ & $8 \%$\\
    10 &  HY & 2 & $5.58e-08$ &   $1.70e-05$ & -\\
    10 & BCHY & 2 & $5.60e-08$ &  $1.57e-05$ & $8 \%$\\
    1 &  HY & 3 &$5.61e-07$ &   $1.80e-05$ & -\\
    1 & BCHY & 3 & $5.62-07$ &  $1.67e-05$ & $7 \%$\\
    5 &  HY & 3 & $1.14e-07$ &   $1.81e-05$ & -\\
    5 & BCHY & 3 & $1.12e-07$ &  $1.68e-05$ & $7 \%$\\
    10 &  HY & 3 & $5.56e-08$ &   $1.80e-05$ & -\\
    10 & BCHY & 3 & $5.55e-08$ &  $1.68e-05$ & $7 \%$\\
    1 &  HY & 4 & $4.41e-07$ &   $1.10e-05$ & -\\
    1 & BCHY & 4 & $4.44e-07$ &  $1.11e-05$ & $- 1 \%$\\
    5 &  HY & 4 & $8.81e-08$ &   $1.10e-05$ & -\\
    5 & BCHY & 4 & $8.80e-08$ &  $1.09e-05$ & $1 \%$\\
    10 &  HY & 4 & $4.39e-08$ &   $1.08e-05$ & -\\
    10 & BCHY & 4 & $4.43e-08$ &  $1.08e-05$ & $0 \%$\\
    
   \hline
  \end{tabular}
  \caption{Summary statistics based on simulated endogenous data of $1$, $5$ and $10$ years. The RMSE in the table corresponds to 
  the square root of the squared distance between the estimated value and the true value $6.4 e-05$. HY
  stands for the usual Hayashi-Yoshida estimator (\ref{HY}), and BCHY represents the bias-corrected 
  estimator (\ref{biascorrected}).}
\label{num}
 \end{table}

  \begin{table}[p!]
 \centering
  \begin{tabular}{| c | c | c | c | c | c | c |}
   \hline
    No. years &  $0.5$ \% & $2.5$ \% & $5$ \% 
    & $95$ \% & $97.5$ \% & $99.5$ \% \\
  \hline
    1 & -2.48 & -1.99 & -1.59 & 1.66 & 2.13 & 2.57 \\
    5 & - 2.60 & -1.96 & -1.64 & 1.64 & 2.05 & 2.62  \\
    10 &- 2.68 & - 1.98 & -1.60 & 1.65 & 2.01 & 2.73  \\
   \hline
  \end{tabular}
  \caption{In this table, we report the finite sample quartiles of the feasible standardized statistic (\ref{feasiblestat}) in setting $1$. The benchmark quartiles are those for the limit distribution $\mathcal{N} (0, 1)$.}
\label{tablefeasiblestat}
 \end{table}

 \end{document}